\pdfoutput=1
\RequirePackage{ifpdf}
\ifpdf 
\documentclass[pdftex]{sigma}
\else
\documentclass{sigma}
\fi

\usepackage{tikz}
\usetikzlibrary{decorations.markings}

 \numberwithin{equation}{section}

\newtheorem{Thm}{Theorem}[section]
\newtheorem{Cor}[Thm]{Corollary}
\newtheorem{Prop}[Thm]{Proposition}
\newtheorem{Lemma}[Thm]{Lemma}

\theoremstyle{definition}

\newtheorem{Remark}[Thm]{Remark}

\newcommand{\Tr}{\operatorname{Tr}}
\def\Re{\operatorname{Re}}
\def\Im{\operatorname{Im}}
\def\bigO{{\cal O}}

\begin{document}

\allowdisplaybreaks

\renewcommand{\thefootnote}{$\star$}

\newcommand{\arXivNumber}{1508.06734}

\renewcommand{\PaperNumber}{031}

\FirstPageHeading

\ShortArticleName{Random Matrices with Merging Singularities and the Painlev\'e V Equation}

\ArticleName{Random Matrices with Merging Singularities\\ and the Painlev\'e V Equation\footnote{This paper is a~contribution to the Special Issue
on Asymptotics and Universality in Random Matrices, Random Growth Processes, Integrable Systems and Statistical Physics in honor of Percy Deift and Craig Tracy.
The full collection is available at \href{http://www.emis.de/journals/SIGMA/Deift-Tracy.html}{http://www.emis.de/journals/SIGMA/Deift-Tracy.html}}}

\Author{Tom CLAEYS and Benjamin FAHS}
\AuthorNameForHeading{T.~Claeys and B.~Fahs}
\Address{Institut de Recherche en Math\'ematique et Physique, Universit\'e catholique de Louvain,\\
 Chemin du Cyclotron 2, B-1348 Louvain-La-Neuve, Belgium}
\Email{\href{mailto:tom.claeys@uclouvain.be}{tom.claeys@uclouvain.be}, \href{mailto:benjamin.fahs@uclouvain.be}{benjamin.fahs@uclouvain.be}}

\ArticleDates{Received September 08, 2015, in f\/inal form March 18, 2016; Published online March 23, 2016}

\Abstract{We study the asymptotic behavior of the partition function and the correlation kernel in random matrix ensembles of the form
$\frac{1}{Z_n} \big|\det \big( M^2-tI \big)\big|^{\alpha} e^{-n\Tr V(M)}dM$,
where~$M$ is an $n\times n$ Hermitian matrix, $\alpha>-1/2$ and $t\in\mathbb R$,
in double scaling limits where $n\to\infty$ and simultaneously~$t\to 0$. If $t$ is proportional to~$1/n^2$, a transition takes place which can be described in terms of a family of solutions to the Painlev\'e V equation.
These Painlev\'e solutions are in general transcendental functions, but for certain values of~$\alpha$, they are algebraic, which leads to explicit asymptotics of the partition function and the correlation kernel.}

\Keywords{random matrices; Painlev\'e equations; Riemann--Hilbert problems}

\Classification{60B20; 35Q15; 33E17}

\begin{flushright}
\begin{minipage}{70mm}
\it Dedicated to Percy Deift and Craig Tracy\\ on the occasion of their 70th birthdays.
\end{minipage}
\end{flushright}

\renewcommand{\thefootnote}{\arabic{footnote}}
\setcounter{footnote}{0}

 \section{Introduction}
We consider a random matrix model on $n \times n$ Hermitian matrices, def\/ined by a probability measure of the form
\begin{gather} \label{matrixmeasure}
\frac{1}{Z_n} \left|\det \left( M^2-tI \right)\right|^{\alpha} e^{-n\Tr V(M)}dM \qquad \textrm{for} \quad \alpha>-1/2, \quad t\in \mathbb{R}.
\end{gather}
Here $dM=\prod\limits_{1\leq i \leq n} dM_{ii}\prod\limits_{1\leq i<j\leq n} d\Re M_{ij} d\Im M_{ij}$
is the Lebesgue measure on $n\times n$ Hermitian matrices and $Z_n=Z_n(t,\alpha,V)$ is a normalizing constant. The conf\/ining potential $V\colon \mathbb{R} \to \mathbb{R}$ is a real analytic function which increases suf\/f\/iciently fast asymptotically, namely
\begin{gather} \label{Cond:V}
\lim_{x\to\pm\infty}\frac{V(x)}{\log (x^2+1)} = +\infty .
\end{gather}

The measure (\ref{matrixmeasure}) is invariant under unitary conjugation and induces a probability measure on the eigenvalues $x_i$, $i=1,2,\dots, n$, given by
 \begin{gather} \label{Def:eigs}
 \frac{1}{\widehat{Z}_n} \prod_{1\leq i<j \leq n} (x_j-x_i)^2 \prod_{j=1}^n w_n(x_j)dx_j,\qquad w_n(x)=e^{-nV(x)} \big|x^2-t\big|^{\alpha},
 \end{gather}
where
\begin{gather}
\widehat{Z}_n= \widehat{Z}_n(t,\alpha,V)=\int_{-\infty}^{\infty}\cdots\int_{-\infty}^{\infty}\prod_{1\leq i<j \leq n} (x_j-x_i)^2
\prod_{j=1}^n w_n(x_j)dx_j \label{Defn PartitionFunction}
\end{gather}
 is the partition function associated to the ensemble.
It can alternatively be written as the Hankel determinant
\begin{gather}\label{Hankel}
\widehat Z_n(t,\alpha,V)= n!\det\left(\int_{-\infty}^{+\infty}
x^{i+j}w_n(x)dx\right)_{i,j=0,\ldots, n-1}.
\end{gather}

While $\alpha$ is kept f\/ixed and does not depend on the matrix size~$n$, we will be particularly interested in double scaling limits where $t$ is $n$-dependent and approaches $0$ as $n\to\infty$.
In such a case, the weight function $w_n$ has the feature that it has two merging singularities approaching~$0$ as $n\to\infty$.
If $t>0$, the two singularities $\pm\sqrt{t}$ are real; if $t<0$, the singularities $\pm i\sqrt{-t}$ are purely imaginary.
Random matrix ensembles with such merging root-type singularities have applications in fractional Brownian motion with Hurst index $H=0$~\cite{Fyodorov}, and in quantum chromodynamics~\cite{Akemann1, DN}. The partition function~$\widehat Z_n$ appears in the study of impenetrable bosons in a harmonic well as the one body density matrix (up to a~constant)~\cite{Forresterbose}.

\looseness=-1
Although the singularities have no ef\/fect on the macroscopic limiting density of the eigen\-va\-lues as $n\to\infty$, they do have an ef\/fect on the partition function $\widehat Z_n$ and on the microscopic correlations between eigenvalues
near the origin. We will obtain asymptotics for the partition function and for the eigenvalue correlation kernel near the origin, in double scaling limits where $n\to\infty$ and $t\to 0$ simultaneously. The asymptotics for the partition function are described in terms of solutions to the f\/ifth Painlev\'e equation; the limit of the correlation kernel can also be expressed in terms of functions related to the Painlev\'e~V equation. As $t\to 0$ at a fast rate, the limiting kernel degenerates to a Bessel kernel, and as $t\to 0$ at a slow rate, we recover the sine kernel.

\subsection{Statement of results}

The macroscopic large $n$ behavior of the eigenvalues is described by the limit of the mean eigenvalue distribution as $n \to \infty$, which we denote by $\mu_V$, and which is independent of~$t$ and~$\alpha$, for $t$, $\alpha$ bounded.
It is characterized~\cite{Deift, DKM} as the unique equilibrium measure~$\mu_V$ which minimizes
\begin{gather}\label{energy}
\iint \log \frac{1}{|x-y|} d\mu (x) d\mu (y) +\int V(x) d\mu(x),
\end{gather}
over the space of Borel probability measures $\mu$ on $\mathbb R$. The equilibrium measure is alternatively characterized by the Euler--Lagrange variational conditions \cite{SaffTotik}
\begin{gather}
2\int \log|x-y|d\mu_V(y)=V(x)-\ell, \qquad \text{for} \quad x\in \operatorname{supp}\mu_V,\label{EulLag1}\\
2\int \log|x-y|d\mu_V(y)\leq V(x)-\ell, \qquad \text{for} \quad x\in \mathbb{R}{\setminus} \operatorname{supp} \mu_V, \label{EulLag2}
\end{gather}
for some constant $\ell$ depending on $V$.

Throughout the paper we will assume that $V$ is such that the equilibrium measure is supported on a single interval $[a,b]$.
Then, the equilibrium measure can be written in the form \cite{DKM}
\begin{gather}\label{eq density}
d\mu_V(x)=\psi_V(x)dx=\sqrt{(b-x)(x-a)} h(x) dx,\qquad x\in[a,b],
\end{gather}
with $h$ a real analytic function.
We require $V$ to be such that the following generic conditions hold: the density $\psi_V$ is strictly positive on $(a,b)$, it vanishes like a square root at the endpoints~$a$,~$b$ (in other words, $h(x)>0$ for~$x\in[a,b]$), and the variational inequality~(\ref{EulLag2}) is strict on $\mathbb R{\setminus}[a,b]$. If all those assumptions hold, we say that~$V$ is one-cut regular (see~\cite{KM} for a general classif\/ication of singularities). In addition, we require zero to be contained in the interior of the support of $\mu_V$, i.e., $a<0<b$.

If zero were outside of the support of $\mu_V$, no eigenvalues are expected near the origin for large~$n$, and the singularities of the weight will not play a signif\/icant role for large~$n$.
The case where~$0$ is an edge point of the support of $V$ is also interesting, but will not be studied here~-- we refer to~\cite{IKO} for results about the local eigenvalue correlations when a~singularity lies near a~soft edge, and to~\cite{XuZhao} when a singularity approaches a hard
edge in a modif\/ied Jacobi ensemble.

\subsubsection*{Painlev\'e V functions}

In order to describe the asymptotic behavior in the matrix model \eqref{matrixmeasure} for $n$ large and $t$ small, we need a special 1-parameter form of the more general 4-parameter Jimbo--Miwa--Okamoto $\sigma$-form of the f\/ifth Painlev\'e equation (see \cite{Jimbo, JimboMiwa}, and \cite{Forrestertau} for the explicit relation to the usual Painlev\'e~V equation), which is given by
\begin{gather} \label{PainleveV}
(s\sigma''(s))^2=(\sigma(s)-s\sigma'(s)+2(\sigma'(s))^2+2\alpha\sigma'(s))^2-4 (\sigma'(s))^2(\sigma'(s)+\alpha)^2.
\end{gather}
For $\alpha>-1/2$, we are interested in two particular solutions $\sigma_\alpha^+$ and $\sigma_\alpha^-$ to this equation.
The f\/irst one, $\sigma_\alpha^-(s)$, is analytic and real for $s\in (0,+\infty)$, and it has the asymptotics
\begin{gather}\label{asymsigma-}
\sigma_\alpha^-(s)=\begin{cases}
\alpha^2+o(1), & s\to 0,\\
s^{-1+2\alpha}e^{-s}\dfrac{1}{\Gamma(\alpha)^2}\big(1+\mathcal O \big(s^{-1}\big)\big), &s \to +\infty,
\end{cases}
\end{gather}
where $\Gamma$ is Euler's $\Gamma$-function.
The existence of such a solution was proved in \cite{CIK}.
The second solution $\sigma_\alpha^+(s)$ is analytic and such that $\sigma_\alpha^+(s)+\frac{\alpha s}{2}$ is real-valued for $s\in -i\mathbb{R}_{>0}$, satisfying in addition:
\begin{gather*} \nonumber
\sigma_\alpha^+(s)= \begin{cases}
\alpha^2+o(1),& s\to -i0^+, \\
\dfrac{\alpha^2}{2}-\dfrac{\alpha s}{2}+\mathcal O\big(|s|^{-1}\big), & s \to -i\infty.
\end{cases} \end{gather*}
Its existence was proved in \cite{CK}. It should be noted that uniqueness of the solutions $\sigma_\alpha^\pm$ satisfying the above asymptotic conditions was not proven in \cite{CIK, CK}.
Those Painlev\'e solutions $\sigma_\alpha^-$ and $\sigma_\alpha^+$ can be characterized in terms of Riemann--Hilbert (RH) problems; we will do this in Section \ref{section: PV} for $\sigma_\alpha^-$ and in Section \ref{sec: PV2} for $\sigma_\alpha^+$. For general values of $\alpha>-1/2$, $\sigma_\alpha^\pm$ are transcendental functions, but there are special values of $\alpha$ where they are algebraic.
For $\sigma_\alpha^-$, this is the case if $\alpha$ is an integer, while for~$\sigma_\alpha^+$ only if~$\alpha$ is an even integer. For those values of $\alpha$, $\sigma_\alpha^\pm$ can be constructed via a recursive procedure, involving Schlesinger type transformations~\cite{FIKN}, which we will explain in detail in Sections~\ref{sec: rec-} and~\ref{sec: rec+}.

For $\alpha=0,1,2$, the expressions are simple and we have
\begin{gather}
\sigma_0^\pm(s)=0,\label{sigma0-intro}\\
\sigma_1^-(s)=\frac{s}{e^s-1},\label{sigma1-intro}\\
\sigma_2^\pm(s)=s\frac{-2+e^s(2-2s+s^2)}{1-e^s(2+s^2)+e^{2s}}.\label{sigma2-intro}
\end{gather}

\subsubsection*{Asymptotics for the partition function}

\begin{Thm}\label{Thm PartitionFunction} Let $V$ be real analytic, satisfying \eqref{Cond:V} and one-cut regular with $0$ in the interior of the support $[a,b]$ of the equilibrium measure $\mu_V$. Denote $z_0=\sqrt{t}$ when $t$ is positive and $z_0=i\sqrt{-t}$ when $t$ is negative. Define
\begin{gather*}
s_{n,t}=-2\pi i n\int_{-z_0}^{z_0}h(s)\left((s-a)(b-s)\right)^{1/2}ds,
\end{gather*}
with $h$ defined in~\eqref{eq density} and
where the square root is taken to be positive at $0$ and analytic in a~neighbourhood of~$0$, in such a way that~$s_{n,t}$ is positive when $t<0$ and negative imaginary when~$t>0$. We assume that $\alpha>-1/2$ and $t<0$ or that $\alpha>0$ and $t>0$. Then, the partition function $\widehat{Z}_n(t)$ defined in~\eqref{Defn PartitionFunction} satisfies
\begin{gather}
\log \widehat Z_n(t,\alpha,V)= \log \widehat Z_n(0,\alpha,V)+\int_0^{s_{n,t}}\frac{\sigma_\alpha^\pm (s)-\alpha^2}{s}ds+\frac{\alpha s_{n,t}}{2}\nonumber\\
\hphantom{\log \widehat Z_n(t,\alpha,V)=}{}
+\frac{n\alpha}{2}(V(z_0)+V(-z_0)-2V(0))+\mathcal O\big(|t|^{1/2}\big) +{\mathcal O}\big(n^{-1}\big)
,\label{as Znthm}
\end{gather}
with uniform error terms as $n\to \infty$ and $t\to 0$, $\pm t>0$.
\end{Thm}

\begin{Remark}
\label{remark improvement theorem}
If we set $\widehat s_{n,t}=-4\pi i nz_0 \psi_V(0)$, we have $\widehat s_{n,t}=s_{n,t}\left(1+\bigO(|t|)\right)$ as $n\to\infty$, $t\to 0$. If $t>0$, it follows from equation \eqref{int sigma} below that one can replace $s_{n,t}$ by $\widehat s_{n,t}$ in~\eqref{as Znthm} without modifying the error terms.
Doing the same when $t<0$ induces an error term of order $\mathcal O (n|t|^{3/2})$.
\end{Remark}
\begin{Remark}
We believe that the expansion \eqref{as Znthm}, with error term $o(1)$, holds for $t>0$ and $-1/2<\alpha<0$ as well, but because of technical obstacles, we do not prove the result in this case.
\end{Remark}
\begin{Remark}
The asymptotic behavior for $\widehat Z_n(t,\alpha,V)$ is described in terms of the Painlev\'e V solutions $\sigma_\alpha^\pm$ and in terms of the partition function $\widehat Z_n(0,\alpha,V)$, which corresponds to a weight function which has only one singularity at $0$. Asymptotics for $\widehat Z_n(0,\alpha,V)$ were obtained in \cite{Krasovsky} in the case where $V$ is quadratic, an important special case which we will return to later in this section.
For general $V$, we have not been able to f\/ind explicit asymptotic expansions in the literature for $\log\widehat Z_n(0,\alpha,V)$, up to decaying terms as $n\to\infty$.
\end{Remark}
\begin{Remark}
By \eqref{sigma0-intro}--\eqref{sigma2-intro}, the integral appearing in \eqref{as Znthm} can be computed explicitly for $\alpha=0,1,2$. We have
\begin{gather}
\int_0^{s_{n,t}}\frac{\sigma^\pm_0(s)}{s}ds=0,\\
\int_0^{s_{n,t}}\frac{\sigma^-_1(s)-1}{s}ds=\log \left(2\sinh \frac{s_{n,t}}{2}\right)-\frac{s_{n,t}}{2}-\log s_{n,t},\label{sigma alpha 1}\\
\int_0^{s_{n,t}}\frac{\sigma^\pm_2(s)-4}{s}ds=\log\left( 4\sinh^2 \frac{s_{n,t}}{2}-s_{n,t}^2\right)-s_{n,t}-4\log s_{n,t}+\log12.\label{sigma alpha 2}
\end{gather}
\end{Remark}

\begin{Remark}Our proofs of Theorems~\ref{Thm PartitionFunction} and~\ref{thm lim kernel} are based on an asymptotic analysis of the orthogonal polynomials on the real line def\/ined by~(\ref{orthogonalpolynomials}), using RH methods~\cite{Deift, Deiftetal, Deiftetal2}. The asymptotic analysis for those orthogonal polynomials shows similarities to the one performed in~\cite{CIK, CK} (following the general method from~\cite{DIK}) for orthogonal polynomials on the unit circle with respect to weights with merging or emerging singularities, where they were used to obtain asymptotics for Toeplitz determinants. Nevertheless, there is no known way to directly derive the asymptotics for the Hankel determinant~\eqref{Hankel}, with respect to the $n$-dependent weight~$w_n$ supported on the full real line, from the asymptotics for Toeplitz determinants.
\end{Remark}

\subsubsection*{Asymptotics for the correlation kernel near~0}

 Denote by $p_j^{(n)}$ the degree $j$ orthonormal polynomials with respect to the weight $w_n(x)$ (def\/ined in \eqref{Def:eigs}), characterized by
\begin{gather}\label{orthogonalpolynomials}
\int_{\mathbb{R}} p_j^{(n)}(x)p_k^{(n)}w_n(x)dx=\delta_{jk}= \bigg{\{}\begin{matrix} 0 \quad \textrm{for $j\neq k$,} \\ 1 \quad \textrm{for $j=k$,} \end{matrix}
\end{gather}
with positive leading coef\/f\/icient $\kappa_j^{(n)}$.
The eigenvalue density~(\ref{Def:eigs}) characterizes a determinantal point process with correlation kernel given by
\begin{gather*}
K_n(x,y)=\sqrt{w_n(x)w_n(y)}\sum_{j=0}^{n-1} p_j^{(n)}(x)p_j^{(n)}(y),
\end{gather*}
or alternatively
\begin{gather}\label{correlationkernelCD}
K_n(x,y)=\sqrt{w_n(x)w_n(y)}\frac{\kappa^{(n)}_{n-1}}{\kappa^{(n)}_n} \frac{p_n^{(n)} (x)p_{n-1}^{(n)}(y)-p_{n-1}^{(n)}(x)p_n^{(n)}(y)}{x-y},
\end{gather}
by the Christof\/fel--Darboux formula. Note that the orthogonal polynomials and the correlation kernel $K_n$ depend on the parameter~$t$ in the weight~$w_n$.
This paper is concerned with the microscopic large $n$ behavior of the eigenvalues near $0$. We will study the scaled correlation kernel $\frac{1}{cn}K_n\left(\frac{u}{cn},\frac{v}{cn}\right)$ for a suitable choice of $c>0$, and obtain asymptotics for it as $n\to\infty$ and $t\to 0$.
We will see that a transition in the asymptotics for the kernel takes place if $t$ is of the order $n^{-2}$, or equivalently if the distance between the two singularities in the weight is of the order~$n^{-1}$, which is of the same order as the typical distance between consecutive eigenvalues.

Local scaling limits near the origin for the correlation kernel $K_n$ are well understood for $t\neq 0$ f\/ixed and for $t=0$, as $n\to\infty$.
If $t\neq 0$ is independent of $n$, then, as a slight generalization of results obtained in \cite{Deiftetal, Deiftetal2}, we have for $u,v\in\mathbb R$ that
\begin{gather}
 \lim_{n \rightarrow \infty} \frac{1}{\psi_V(0)n}K_n\left(\frac{u}{ \psi_V(0)n},\frac{v}{\psi_V(0)n}\right)=\mathbb K^{\sin}(u,v)= \frac{\sin \pi(u-v)}{\pi(u-v)}\label{sinekernel}.
\end{gather}
The convergence is uniform for~$u$,~$v$ in compact subsets of~$\mathbb R$.
If $t=0$, it was proved in~\cite{2} that for~$u$,~$v$ in bounded subsets of $(0,\infty)$ that
\begin{gather}
 \lim_{n \rightarrow \infty} \frac{1}{\psi_V(0)n}K_n\left(\frac{u}{\psi_V(0)n},\frac{v}{\psi_V(0)n}\right)=\mathbb K_{\alpha}^{\textrm{Bessel}}(u,v)\nonumber\\
 \qquad{} = \pi\sqrt{u v} \frac{J_{\alpha+\frac{1}{2}}(\pi u)J_{\alpha-\frac{1}{2}}(\pi v)-J_{\alpha-\frac{1}{2}}(\pi u)J_{\alpha+\frac{1}{2}}(\pi v)}{2(u-v)}, \label{besselkernel}
\end{gather}
where $J_{\nu}$ is the Bessel function of order~$\nu$ (see \cite{NIST} for a reference on Bessel functions and other special functions which appear throughout the paper). The convergence is uniform for~$u$,~$v$ in compact subsets of $(0,\infty)$.

We obtain large $n$ asymptotics for the correlation kernel near the origin as $n\to\infty$ for $t$ small. Of particular interest will be the double scaling limit where $n\to\infty$ and simultaneously~$t\to 0$ in such a way that~$n^2t$ tends to a non-zero constant. This will lead us to a new family of limiting kernels which are built out of functions associated to the Painlev\'e V equation.

\begin{Thm}\label{thm lim kernel} Let $V$ be real analytic, satisfying~\eqref{Cond:V} and one-cut regular with $0$ in the interior of the support $[a,b]$ of the equilibrium measure $\mu_V$. We denote $\psi_V$ for the density of $\mu_V$, given by~\eqref{eq density}, and define
\begin{gather}\label{taunt}
\tau_{n,t}=16\pi^2\psi_V(0)^2n^2t,
\end{gather}
for $t$ real.
\begin{enumerate}\itemsep=0pt
\item[$1.$] As $n\to \infty$ and simultaneously $t\to 0$ in such a way that $\tau_{n,t}\to \pm\infty$, we have
\begin{gather}\label{limit kernelinf}
\lim\frac{1}{\psi_V(0)n}K_{n}\left(\frac{u}{\psi_V(0)n}, \frac{v}{\psi_V(0)n}\right)= \mathbb K^{{\sin}}(u,v),
\end{gather}
for any $u,v\in\mathbb R$, with $\mathbb K^{{\sin}}$ as in \eqref{sinekernel}.
\item[$2.$] As $n\to \infty$ and simultaneously $t\to 0$ in such a way that $\tau_{n,t}\to 0$, we have
\begin{gather}\label{limit kernel0}
\lim\frac{1}{\psi_V(0)n}K_{n}\left(\frac{u}{\psi_V(0)n}, \frac{v}{\psi_V(0)n}\right)= \mathbb K_\alpha^{{\rm Bessel}}(u,v),
\end{gather}
for any $u,v\in\mathbb R{\setminus}\{0\}$, with $\mathbb K_\alpha^{{\rm Bessel}}$ as in~\eqref{besselkernel}.
\item[$3.$] As $n\to \infty$ and simultaneously $t\to 0$ in such a way that $\tau_{n,t}\to \tau \neq 0$, there exists a~limiting kernel $\mathbb K_{\alpha}^{{\rm PV}}$ such that
\begin{gather}\label{limit kernel}
\lim\frac{1}{\psi_V(0)n}K_{n}\left(\frac{u}{\psi_V(0)n}, \frac{v}{\psi_V(0)n}\right)= \mathbb K_{\alpha}^{{\rm PV}}(u,v;\tau),
\end{gather}
for any $u,v\in\mathbb R$ if $\tau<0$, and for any $u,v\in\mathbb R{\setminus}\{\pm\sqrt{\tau}/4\}$ if $\tau>0$.
\end{enumerate}
The limits are uniform for $u$, $v$ in compact subsets of~$\mathbb R$.
\end{Thm}

\begin{Remark}
The limiting kernel $\mathbb K_{\alpha}^{{\rm PV}}(u,v;\tau)$ has the form
\begin{gather} \label{limker}
\mathbb K_{\alpha}^{{\rm PV}}(u,v;\tau)=\frac{\Phi_1(\pi v;\tau)\Phi_2(\pi u;\tau)-\Phi_1(\pi u;\tau)\Phi_2(\pi v;\tau)}{2\pi i (u-v)},
 \end{gather}
 and the functions $\Phi_j$, $j=1,2$, will be def\/ined in Section~\ref{section: PV} for $\tau<0$ and in Section \ref{sec: PV2} for $\tau>0$.
 The easiest way to def\/ine $\Phi_1$ and $\Phi_2$ is via a RH problem which appeared in~\cite{CIK, CK}. Alternatively, they can be characterized as special solutions to a Lax pair which is related to the f\/ifth Painlev\'e equation.
For $\tau<0$, $\Phi_1(u;\tau)$ and $\Phi_2(u;\tau)$ are analytic as functions of $u$ in a neighborhood of~$\mathbb R$; for~$\tau>0$, they are analytic functions of $u$ in a neighborhood of~$\mathbb R{\setminus} \{\pm \frac{\sqrt{\tau}}{4}\}$; they are analytic as functions of~$\tau$ for~$\tau$ in a~neighborhood of $\mathbb R{\setminus}\{0\}$ and continuous in~$\tau$ at~$0$. For $\tau=0$, the limiting kernel $\mathbb K_\alpha^{\rm PV}$ is a special case of a large family of limiting
kernels which appeared in~\cite{Atkin}.
\end{Remark}

\begin{Remark}
In general, the kernel $\mathbb K_{\alpha}^{{\rm PV}}$ is a transcendental function of~$u$,~$v$,~$\tau$. However,
 for $\tau<0$ and $\alpha\in\mathbb N$, and for $\tau>0$ and $\alpha$ even, the kernel $\mathbb K_{\alpha}^{{\rm PV}}$ is algebraic and can be constructed explicitly by a recursive procedure. We will present this procedure in Section~\ref{sec: rec-} for $\tau<0$ and in Section~\ref{sec: rec+} for $\tau>0$.
\end{Remark}

\begin{Remark} \label{limthm}
The limiting kernel $\mathbb K_\alpha^{{\rm PV}}(u,v;\tau)$ degenerates to the sine and Bessel kernels for large and small values of $\tau$. Indeed, we will show that
\begin{gather}\label{lim sine}
\lim_{\tau\to\pm\infty} \mathbb K_{\alpha}^{{\rm PV}}(u,v;\tau)=\mathbb K^{\sin}(u,v),\qquad u,v\in\mathbb R,
\end{gather} and that
\begin{gather}\label{lim Bessel}
\lim_{\tau\to 0} \mathbb K_{\alpha}^{{\rm PV}}(u,v;\tau)=\mathbb K_{\alpha}^{{\rm Bessel}}(u,v),\qquad u,v\in\mathbb R{\setminus}\{0\}.
\end{gather}
This implies that~\eqref{limit kernelinf}, \eqref{limit kernel0}, and \eqref{limit kernel} are consistent: letting $\tau\to \pm\infty$ in \eqref{limit kernel}, we recover~\eqref{limit kernelinf}, and letting $\tau\to 0$, we get~\eqref{limit kernel0}.
\end{Remark}

\subsection{Consequences and applications}

\subsubsection*{Asymptotics for Toeplitz determinants}

In \cite{CIK}, large $n$ asymptotics were obtained for Toeplitz determinants
\begin{gather*}
D_n(f)=\det (f_{j-k})_{j,k=0}^{n-1}, \qquad f_j=\frac{1}{2\pi} \int_0^{2\pi} f\big(e^{i\theta}\big)e^{-ij\theta}d\theta
\end{gather*}
 with an emerging Fisher--Hartwig singularity, depending on parameters~$\alpha$,~$\beta$. In the special case $\beta=0$, the weight has the form
\begin{gather*}
f_t(z)=\big(z-e^t\big)^\alpha\big(z-e^{-t}\big)^{\alpha}z^{-\alpha}e^{-\pi i \alpha}e^{V(z)}, \qquad z=e^{i\theta}, \quad \theta \in [0,2\pi)
\end{gather*}
for $\Re\alpha>-\frac{1}{2}$, $t>0$, and where $V(z)=\sum\limits_{k=-\infty}^{+\infty} V_kz^k$ is analytic on an annulus containing the unit circle. The singularities $e^{\pm t}$ approach the unit circle as $t\to 0$, and form a Fisher--Hartwig type singularity if $t=0$. Asymptotics for Toeplitz determinants with weight functions of this form were also obtained in the context of the 2d Ising model~\cite{WMTB}, see~\cite{DIK2} for a review of Toeplitz determinants and the Ising model.
 Theorem~1.1 in~\cite{CIK} states that
\begin{gather*}
\log D_n(f_t)=nV_0+\alpha n t+\sum_{k=1}^{\infty}k\left(V_k-\alpha\frac{e^{-tk}}{k}\right)\left(V_{-k}-\alpha\frac{e^{-tk}}{k}\right)\nonumber\\
\hphantom{\log D_n(f_t)=}{} +\log\frac{G(1+\alpha)^2}{G(1+2\alpha)}+\alpha^2\log 2nt+\int_0^{2nt}\frac{\sigma_\alpha^-(x)-\alpha^2}{x}dx+o(1),
\end{gather*}
as $n\to\infty$ and simultaneously $t\to 0$ where $G$ is Barnes's G-function.
If $\beta=0$ and $\alpha=1,2$, we can substitute~(\ref{sigma alpha 1}) and~(\ref{sigma alpha 2}), here with $s_{n,t}=2nt$, and obtain the more explicit asymptotic expansions
\begin{gather*}
\log D_n(f_t)=nV_0+\sum_{k=1}^{\infty} kV_kV_{-k}
-\sum_{k=1}^\infty (V_k+V_{-k})e^{-tk}
+\log \frac{\sinh nt}{\sinh t} +t +o(1),
\end{gather*}
for $\alpha=1$, and
\begin{gather*}
\log D_n(f_t)= nV_0+\sum_{k=1}^{\infty} kV_kV_{-k}
-2\sum_{k=1}^\infty (V_k+V_{-k})e^{-tk} + 4t\\
\hphantom{\log D_n(f_t)=}{}
-2\log \big(2\sinh^2 t\big) +\log \big(\sinh ^2nt-n^2t^2\big) +o(1),
\end{gather*}
for $\alpha=2$,
as $n\to\infty$ and $t\to 0$. With more ef\/fort, we can obtain explicit asymptotic expansions for any integer $\alpha$.

In \cite{CK}, asymptotics were obtained for Toeplitz determinants $D_n(f_t)$ with merging Fisher--Hartwig singularities, depending on parameters $\alpha_1$, $\alpha_2$, $\beta_1$, $\beta_2$. Setting $\alpha_1=\alpha_2=\alpha/2$ and $\beta_1=\beta_2=0$, the weight has the form \begin{gather*}
f_t(z)=e^{V(z)}\big|z-e^{it}\big|^{\alpha}\big|z-e^{i(2\pi-t)}\big|^{\alpha},
\end{gather*}
where $\Re \alpha>-\frac{1}{2}$, $t\in (0,\pi)$, and with $V$ again analytic on an annulus containing the unit circle. Then, the weight function has two Fisher--Hartwig type singularities if $t>0$ which merge to a~single one as $t\to 0$. In~\cite[Theorem~1.5]{CK}, the following result is stated:
\begin{gather*}
\log D_n(f_t)= \log D_n(f_0)+\int_0^{-2int}\frac{1}{s}\left(\sigma_{\rm CK}(s)-\frac{\alpha^2}{2}\right) ds\\
\hphantom{\log D_n(f_t)=}{} -\frac{\alpha^2}{2}\log \frac{\sin t}{t}-\frac{\alpha}{2}(V(e^{it})+V(e^{-it})-2V(1))+o(1)
\end{gather*}
as $n\to \infty$ and $t\to 0$. The function $\sigma_{\rm CK}$ is related to $\sigma_\alpha^+$ by the formula
\begin{gather*}
\sigma_{\rm CK}(s)=\sigma_\alpha^+(s)-\frac{\alpha^2}{2}+\frac{\alpha s}{2}.
\end{gather*}
 For $\alpha=2$, we can substitute \eqref{sigma alpha 2} (with $s_{n,t}=-2int$) and obtain
\begin{gather*}
\log D_n(f_t)= nV_0+\sum_{k=1}^{\infty} kV_kV_{-k}- 2\log 2t\sin t +\log \left(n^2t^2-\sin^2nt\right) \\
\hphantom{\log D_n(f_t)=}{} -\left(V\left(e^{it}\right)+V\left(e^{-it}\right)-2V_0\right) +o(1),
\end{gather*}
as $n\to\infty$ and $t\to 0$. With more ef\/fort, we can obtain explicit asymptotic expansions for any even $\alpha$.

\subsubsection*{Extreme values of GUE characteristic polynomials}
Let $H$ be a random $n\times n$ GUE matrix, normalized such that
the joint probability distribution of the eigenvalues is given by
\begin{gather*}
\frac{1}{\widehat Z_n^{\rm GUE}}\prod_{1\leq i<j \leq n} (x_j-x_i)^2 \prod_{j=1}^n e^{-2nx_j^2}dx_j,
\end{gather*}
where $\widehat Z_n^{\rm GUE}=Z_n(\alpha=0,V(x)=2x^2)$ is the normalizing constant.
The limiting mean eigenvalue density is then supported on $[-1,1]$ as $n \to \infty$. Def\/ine the (random) characteristic polynomial
\begin{gather*}
P_n(x)=\det (xI-H).
\end{gather*}
The average of products of powers of characteristic polynomials of the form $\mathbb E\left(\prod\limits_{j=1}^2 |P_n(u_j)|^{\alpha}\right)$, given $u_1,u_2\in\mathbb R$, can be expressed as
\begin{gather} \label{pnZn-0}
\mathbb E\left(\prod _{j=1}^2 |P_n(u_j)|^{\alpha}\right)= \frac{1}{\widehat Z_n^{\rm GUE}}\int_{\mathbb R^n}\!\prod_{1\leq i<j \leq n} \! (x_j-x_i)^2 \prod_{j=1}^n |x_j-u_1|^\alpha|x_j-u_2|^\alpha e^{-2nx_j^2}dx_j.\!\!\!
\end{gather}

The following asymptotic results as $n\to\infty$ were obtained by
Krasovsky \cite{Krasovsky} (see also \cite{BH, FF, Garoni} for $\alpha_j$'s integers):
\begin{gather} \label{formula Krasovsky}
\mathbb E \left( \prod_{j=1}^k|P_n(u_j)|^{2\alpha_j}\right)=F\big(n, (u_i,\alpha_i)_{i=1}^k\big) \left( 1+\mathcal O \left( \frac{\log n}{n}\right)\right),
\end{gather}
where $k=1,2$ and
\begin{gather*} F\big(n,(u_i,\alpha_i)_{i=1}^k\big)=\prod _{j=1}^k C(\alpha_j) \big(1-u_j^2\big)^{\alpha_j^2/2}(n/2)^{\alpha_j^2} e^{(2u_j^2-1-2\log 2)\alpha_jn}\!\!\! \prod_{1\leq i<j\leq k}\!\! (2|u_i-u_j|)^{-2\alpha_i \alpha_j},
\\
 C(\alpha)= 2^{2\alpha^2} \frac{G(\alpha+1)^2}{G(2\alpha+1)},
 \end{gather*}
and $G$ is the Barnes G-function.
If $u_1$ and $u_2$ approach each other, we can use Theorem~\ref{Thm PartitionFunction} to obtain asymptotics for~$\mathbb E \left( \prod\limits_{j=1}^2|P_n(u_j)|^{2\alpha_j}\right)$.
Indeed, by~\eqref{pnZn-0} and~\eqref{Defn PartitionFunction}, it follows that
\begin{gather} \label{pnZn}
\mathbb E\left(\prod _{j=1}^2 |P_n(u_j)|^{\alpha}\right)
= \frac{\widehat Z_n\big(t=\frac{(u_1-u_2)^2}{4},\alpha,V(x)=2(x+(u_1+u_2)/2)^2\big)}{\widehat Z_n^{\rm GUE}},
\end{gather}
which is seen to hold after the simple change of variables $y_n= x_n-\frac{u_1+u_2}{2}$.
 Substituting~\eqref{as Znthm} in the numerator, we obtain asymptotics for~$\mathbb E\left(\prod\limits_{j=1}^2 |P_n(u_j)|^{\alpha}\right)$.

This observation can be used to obtain information about extreme values of $|P_n(x)|$ for large $n$. This problem was investigated in \cite{FyodorovSimm}, and in this context the authors needed large~$n$ asymptotics for integrals of the form (see in particular \cite[Section~2]{FyodorovSimm})
\begin{gather}\label{def In}
I_n(\theta,\rho)= \int_{-\theta}^\theta\int_{-\theta}^\theta \mathbb E\left(\prod _{j=1}^2 |P_n(u_j)|^{\alpha}\right)\prod _{j=1}^2e^{-\alpha n \lim\limits_{k\to\infty}\left(\frac{1}{k}\mathbb E \log |P_k(u_j)|\right) } \rho(u_j) du_j,
\end{gather}
with $\theta\in [0,1]$, $\rho$ strictly positive and continuous on $(-\theta,\theta)$, and $\alpha>0$.
The authors note that one can dif\/ferentiate (\ref{formula Krasovsky}) for $k=1$ with respect to $\alpha$ and evaluate at $\alpha=0$ to obtain
\begin{gather}\label{logPn}
\lim_{n\to\infty}\frac{1}{n}\mathbb E( \log |P_n(u)|)=u^2-1/2-\log2,
\end{gather}
since $C(0)=1$, $C'(0)=0$.
In order to obtain asymptotics for $I_n$ as $n\to\infty$, one needs asymptotics for
$\mathbb E\left(\prod\limits_{j=1}^2 |P_n(u_j)|^{\alpha}\right)$, both in the region where~$u_1$ and~$u_2$ are bounded away from each other (in this region, we can use Krasovsky's result~\eqref{formula Krasovsky}) and in the region where~$u_1$ and~$u_2$ are close to each other (in this region, we need to use our expansion~\eqref{as Znthm}).

As a corollary to Theorem~\ref{Thm PartitionFunction} we have
\begin{Cor}\label{corollary}Let $0<\theta<1$, $\alpha>0$, and let $\rho$ be a strictly positive continuous function on $(-\theta,\theta)$.
Then,
\begin{gather*}
I_n(\theta,\rho)=
\begin{cases}
C_1(\alpha) n^{\alpha^2/2} (1+o(1)) &\textrm{for $\alpha^2<2$,}\\
 C_2 n\log n (1+o(1)) & \textrm{for $\alpha^2=2$,}\\
C_3(\alpha) n^{\alpha^2-1} (1+o(1)) & \textrm{for $\alpha^2>2$,}
\end{cases}
\end{gather*}
as $n \to \infty$, where
\begin{gather*}
C_1(\alpha)=\frac{G(1+\frac{\alpha}{2})^4}{G(1+\alpha)^2} \int_{-\theta}^{\theta}\int_{-\theta}^{\theta}\frac{((1-u_1^2)(1-u_2^2))^{\alpha^{2}/8}}{|u_1-u_2|^{\alpha^2/2}}\rho(u_1)\rho(u_2)du_1 du_2,\\
C_2=2 \frac{G(1+\frac{1}{\sqrt{2}})^4}{G(1+\sqrt{2})^2} \int_{-\theta}^\theta \big(1-u^2\big)^{1/2}\rho(u)^{2}du,\\
 C_3(\alpha)=2^{\alpha^2}\frac{G(\alpha+1)^2}{G(2\alpha+1)}\int_{-\theta}^\theta\! \big(1-u^2\big)^{\frac{\alpha^2-1}{2}}\rho(u)^{2}du
 \int_0^{\infty}\!\exp\left(\int_0^{-2 iv}\!\frac{\sigma_\alpha^+(s)-\alpha^2}{s}ds-i\alpha v\right)dv.
\end{gather*}
\end{Cor}

\begin{Remark}
For $\alpha^2<2$, the main contribution in the integral $I_n(\theta,\rho)$ comes from the asymptotics~\eqref{formula Krasovsky} in the region of integration where $u_1$ and $u_2$ are bounded away from each other. For $\alpha^2>2$, the main contribution comes from the asymptotics in Theorem~\ref{Thm PartitionFunction} in the region of integration where $|u_1-u_2|=\mathcal O(n^{-1})$.
\end{Remark}

\begin{proof}
First, we note that \eqref{formula Krasovsky} can be extended: as $n\to\infty$ and $|u_1-u_2|\to 0$ suf\/f\/iciently slowly such that $n|u_1-u_2|\to\infty$,
\begin{gather} \label{formula Krasovsky extension}
\mathbb E \left( \prod_{j=1}^2|P_n(u_j)|^{\alpha}\right)=F\left(n,u_1,u_2,\frac{\alpha}{2},\frac{\alpha}{2}\right) \left( 1+o(1)\right),
\end{gather}
with the error term uniform for $t_0>|u_1-u_2|>2t_n$ for some $t_0>0$ and $t_n\to 0$ such that $nt_n \to \infty$.
This follows by recalling the identity
\begin{gather}\label{int sigma}
\lim_{s\to -i\infty} \left(\int_0^s \frac{\sigma_\alpha^+ (\widetilde{s})-\alpha^2}{\widetilde s}d\widetilde s+\frac{\alpha s}{2} +\frac{\alpha^2}{2} \log |s|\right)=\log\frac{G\left(1+\frac{\alpha}{2}\right)^4G(1+2\alpha)}{G(1+\alpha)^4},
\end{gather}
from \cite[formula~(1.26)]{CK} and applying it to~\eqref{as Znthm}, which one substitutes into~(\ref{pnZn}).

We now split the integral in \eqref{def In} in two parts: the integral over
\begin{gather}\nonumber
\mathcal A_{\theta, t_n}=\big\{(u_1,u_2)\colon |u_1|,|u_2|<\theta, |u_1-u_2|>2\sqrt{t_n}\big\},
\end{gather}
and the integral over $[-\theta,\theta]^2{\setminus} \mathcal A_{\theta,t_n}$, for some $t_n$ which converges suf\/f\/iciently slowly to~$0$ as \mbox{$n\to\infty$}, slower than~$n^{-2}$.
To compute the contribution of the integral over $\mathcal A_{\theta,t_n}$, we can use~\eqref{logPn},~\eqref{formula Krasovsky}, and~\eqref{formula Krasovsky extension}.
For the contribution of the complement of $\mathcal A_{\theta, t_n}$, we substitute~\eqref{as Znthm} in~\eqref{pnZn} and then, together with~\eqref{logPn}, in \eqref{def In}. Using~\eqref{formula Krasovsky} to compute the ratio $\widehat Z_n(0,\alpha,V)/\widehat Z_n^{\rm GUE}$ and changing variables $t=(u_1-u_2)^2/4$, $u=(u_1+u_2)/2$, we f\/inally obtain
\begin{gather}
I_n(\theta,\rho)=\Bigg(\left( \frac{n}{2}\right)^{\alpha^2/2} C(\alpha/2)^2 \iint_{\mathcal A_{\theta, t_n}}\frac{((1-u_1^2)(1-u_2^2))^{\alpha^{2}/8}}{(2|u_1-u_2|)^{\alpha^{2}/2}}\rho(u_1)\rho(u_2)
du_1 du_2
\nonumber\\
\hphantom{I_n(\theta,\rho)=}{}
+2\left(\frac{n}{2}\right)^{\alpha^2}C(\alpha) \int_{-\theta}^\theta (1-u^2)^{\alpha^2/2}\rho(u)^{2} \nonumber\\
\hphantom{I_n(\theta,\rho)=}{}\times
\left(\int_0^{t_n}\exp\left(\int_0^{\widehat s_{n,t}}\frac{\sigma_\alpha^+(s)-\alpha^2}{s}ds+\frac{\alpha \widehat s_{n,t}}{2}\right)\frac{dt}{\sqrt t}\right)du \Bigg) \left(1+o(1)\right),\label{In2}
\end{gather}
as $n\to \infty$, where
\begin{gather*}
\widehat s_{n,t}=-4\pi in\sqrt{t}\psi_V(0)=-8 in\sqrt{t}\sqrt{1-u^2}
\end{gather*} in terms of the integration variables $t, u$. Here we need that the error terms found in Theorem~\ref{Thm PartitionFunction} are uniform in $-\theta<u<\theta$ for the potential $V(x)=2(x+u)^2$. This is not mentioned elsewhere in this paper but can readily be seen to hold true by inspection of the proof of Theorem~\ref{Thm PartitionFunction}, if $|\theta|<1$. If $\theta=1$, contributions from the region where the singularities approach an edge point of the support of the equilibrium measure will have to be taken into account as well, but this is outside the scope of this paper.

From (\ref{int sigma}) and the small $s$ asymptotics of $\sigma^+_\alpha$ it follows that with $v=4n\sqrt{1-u^2}\sqrt{t}$,
\begin{gather}
\int_0^{t_n}\exp\left(\int_0^{\widehat s_{n,t}}\frac{\sigma_\alpha^+(s)-\alpha^2}{s}ds+\frac{\alpha \widehat s_{n,t}}{2}\right)\frac{dt}{\sqrt t}\nonumber\\
 \qquad{} =\frac{1}{2n\sqrt{1-u^2}}\int_0^{4n\sqrt{1-u^2}\sqrt{t_n}}\exp\left(\int_0^{-2iv}
 \frac{\sigma_\alpha^+(s)-\alpha^2}{s}ds- i\alpha v\right)dv\nonumber\\
 \qquad=\begin{cases}
o\big(n^{-\alpha^2/2}\big)&\textrm{for } \alpha^2<2, \\
\dfrac{1}{4\sqrt{1-u^2}}\dfrac{G(1+\sqrt{2}/2)^4G(1+2\sqrt{2})}{G(1+\sqrt{2})^4}n^{-1}\log n+o\big(n^{-1}\log n\big)&\textrm{for } \alpha^2=2,\\
cn^{-1}+o\big(n^{-1}\big)&\textrm{for } \alpha^2>2,\end{cases}\label{OrderInt}
\end{gather}
as $n\to \infty$, with
\begin{gather*}
c=\frac{1}{2\sqrt{1-u^2}}\int_0^{\infty}\exp\left(\int_0^{-2iv}
 \frac{\sigma_\alpha^+(s)-\alpha^2}{s}ds- i\alpha v\right)dv.
 \end{gather*}
Formulas (\ref{In2}) and (\ref{OrderInt}) yield the corollary, where it is readily seen from (\ref{int sigma}) that the integral appearing in the constant $C_3$ is well def\/ined.
\end{proof}

\subsection{Outline}

In Section~\ref{section: PV}, we recall a model RH problem, introduced in \cite{CIK}, associated to the f\/ifth Painlev\'e equation. We def\/ine the functions $\Phi_1$ and $\Phi_2$ in terms of the solution to the RH problem for $\tau<0$, and describe their relation with the solution $\sigma_\alpha^-$ to the f\/ifth Painlev\'e equation. For $\alpha$ integer, we construct the RH solution recursively. In Section~\ref{sec: PV2}, we describe a second model RH problem which was introduced in~\cite{CK}, we def\/ine $\Phi_1$ and $\Phi_2$ in terms of its solution for $\tau>0$, and we relate them to the Painlev\'e solution $\sigma_\alpha^+$. For $\alpha$ even, we construct the RH solution recursively.
In Section~\ref{section: 4}, we recall the standard RH problem which characterizes the orthogonal polynomials $p_j^{(n)}$, and we def\/ine the $g$-function needed to analyze this RH problem asymptotically. In Section \ref{section: RH-}, we perform a Deift/Zhou steepest descent analysis to obtain large~$n$ asymptotics for the orthogonal polynomials in the case where $t<0$ (i.e., the case where the singularities are complex and approach~$0$). The construction of a local parametrix near~$0$, in terms of the model problem def\/ined in Section~\ref{section: PV}, is crucial here. At the end of Section~\ref{section: RH-}, we will prove Theorems~\ref{thm lim kernel} and~\ref{Thm PartitionFunction} in the case where $t<0$.
In Section~\ref{section: RH+}, we perform a similar asymptotic analysis in the case $t>0$ (i.e., the case where the singularities are real and approach~$0$). Here the local parametrix is constructed in terms of the model problem from Section~\ref{sec: PV2}. This will allow us to prove Theorems~\ref{thm lim kernel} and~\ref{Thm PartitionFunction} for $t>0$.

\section[Model RH problem for $\tau<0$]{Model RH problem for $\boldsymbol{\tau<0}$}\label{section: PV}

The functions $\Phi_1$ and $\Phi_2$ that appeared in the limiting kernel $\mathbb K_\alpha^{{\rm PV}}$ can be def\/ined in terms of a RH problem. We need to distinguish between the case where $\tau>0$ and $\tau<0$.

The f\/irst model RH problem is a special case of the one studied in~\cite[Section 1.3]{CIK}. For our purposes, the parameter~$\beta$ from~\cite{CIK} is set to zero.
The RH problem depends on parameters $\alpha>-1/2$ and $s\in\mathbb C$. The relevant case for us will be $s>0$.

\subsubsection*{RH problem for $\boldsymbol{\Psi^-}$}

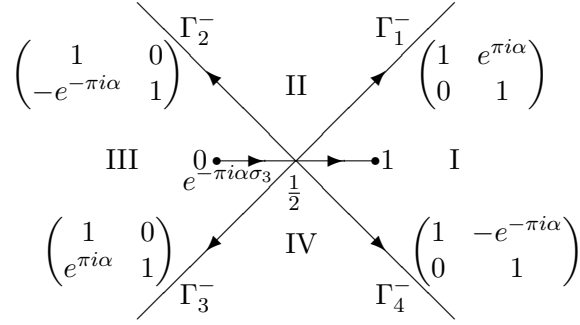
\begin{figure}[t]
\centering
\begin{picture}(120,120)(-5,-65)
 \put(38,-20){$\frac{1}{2}$}
\put(3,-8){$0$}
\put(74,-8){$1$}

\put(72,42){$\Gamma^-_{1}$}
\put(-2,42){$\Gamma^-_{2}$}
\put(-2,-59){$\Gamma^-_{3}$}
\put(72,-59){$\Gamma^-_{4}$}

\put(100,-8){I}
\put(38,20){II}
\put(-30,-8){III}
\put(38,-40){IV}

\put(87,25){$\begin{pmatrix} 1 & e^{\pi i\alpha}\\ 0 &1\end{pmatrix}$}
\put(-66,25){$\begin{pmatrix} 1 &0\\ -e^{-\pi i\alpha} &1\end{pmatrix}$}
\put(-54,-42){$\begin{pmatrix} 1 &0\\ e^{\pi i\alpha} &1\end{pmatrix}$}
\put(85,-42){$\begin{pmatrix} 1 & -e^{-\pi i\alpha}\\0 &1\end{pmatrix}$}
\put(-1,-16){$e^{-\pi i \alpha \sigma_3}$}

\put(-18,-65){\line(1,1){120}}
\put(-18,55){\line(1,-1){120}}
\put(12,-5){\line(1,0){60}}

\put(12,-5){\thicklines\circle*{2.5}}
\put(72,-5){\thicklines\circle*{2.5}}

\put(8,-39){\thicklines\vector(-1,-1){.0001}}
\put(76,29){\thicklines\vector(1,1){.0001}}
\put(8,29){\thicklines\vector(-1,1){.0001}}
\put(76,-39){\thicklines\vector(1,-1){.0001}}
\put(30,-5){\thicklines\vector(1,0){.0001}}
\put(60,-5){\thicklines\vector(1,0){.0001}}
\end{picture}
\caption{The jump contour $\Gamma^-$ and the jump matrices for $\Psi^-$.}\label{figure: contour Psi-}
\end{figure}

\begin{itemize}\itemsep=0pt
\item[(a)] $\Psi^-\colon \mathbb{C}{\setminus} \Gamma^- \rightarrow \mathbb{C}^{2\times 2}$ is analytic, with $\Gamma^-=\Gamma_1^-\cup\Gamma_2^-\cup\Gamma_3^-\cup\Gamma_4^-\cup[0,1]$, and
\begin{gather*}
\Gamma_1^-=\frac{1}{2}+e^{\frac{\pi i}{4}}\mathbb R^+,\ \ \Gamma_2^-=\frac{1}{2}+e^{\frac{3\pi i}{4}}\mathbb R^+,\ \ \Gamma_3^-=\frac{1}{2}+e^{-\frac{3\pi i}{4}}\mathbb R^+,\ \ \Gamma_4^-=\frac{1}{2}+e^{-\frac{\pi i}{4}}\mathbb R^+,
\end{gather*}
oriented as in Fig.~\ref{figure: contour Psi-}.
\item[(b)] $\Psi^-$ has continuous boundary values on $\Gamma^- {\setminus} \{0,\frac{1}{2},1\}$, which we denote by $\Psi_+^-(z)$ if the limit is taken from the left when oriented along the contour, and $\Psi_-^-(z)$ if the limit is taken from the right. We have the jump relations
\begin{alignat*}{3}
&\Psi_+^-(z)=\Psi_-^-(z)\begin{pmatrix} 1 & e^{\pi i\alpha}\\ 0 &1\end{pmatrix} \qquad &&\textrm{for $z \in \Gamma_1^-$,}&\\
&\Psi_+^-(z)=\Psi_-^-(z)\begin{pmatrix} 1 &0\\ -e^{-\pi i\alpha} &1\end{pmatrix} \qquad &&\textrm{for $z \in \Gamma_2^-$,}&\\
&\Psi_+^-(z)=\Psi_-^-(z)\begin{pmatrix} 1 &0\\ e^{\pi i\alpha} &1\end{pmatrix} \qquad &&\textrm{for $z \in \Gamma_3^-$,}&\\
&\Psi_+^-(z)=\Psi_-^-(z)\begin{pmatrix} 1 & -e^{-\pi i\alpha}\\0 &1\end{pmatrix}\qquad &&\textrm{for $z \in \Gamma_4^-$,}&\\
&\Psi_+^-(z)=\Psi_-^-(z)e^{-\pi i\alpha\sigma_3} \qquad &&\textrm{for $z \in (0,1)$,}&
\end{alignat*}
where $\sigma_3=\begin{pmatrix}1&0\\0&-1\end{pmatrix}$ is the third Pauli matrix.
\item[(c)]There exist $q$, $p$, $r$ independent of $z$ (but depending on $s,\alpha$) such that $\Psi^-(z)$ has the following behavior as $z\to \infty$:
\begin{gather}\label{Psi as}
\Psi^-(z)=\left(I+\frac{1}{z}\begin{pmatrix}q&r\\p&-q\end{pmatrix}+\mathcal{O}\big(z^{-2}\big)\right)\exp\left(-\frac{s}{2}z\sigma_3\right).
\end{gather}
\item[(d)] The functions $G(z):=\Psi^-(z)z^{-\frac{\alpha}{2}\sigma_3}$ and $H(z):=\Psi^-(z)(z-1)^{\frac{\alpha}{2}\sigma_3}$ are bounded for~$z$ near~$0$ and~$1$ respectively, and~$\Psi^-$ is bounded near~$1/2$. The branches in the def\/initions of~$G$,~$H$ are such that $z^{ \frac{\alpha}{2}}, (z-1)^{\frac{\alpha}{2}}>0$ on the~$+$ side of $(1,\infty)$.
\end{itemize}

For $\alpha>-1/2$,
it was shown in \cite[Theorem 1.8(i)]{CIK} that the RH problem for $\Psi^-$ has a unique solution for all $s>0$.

It will be useful to note that the following symmetry relation holds:
\begin{gather} e^{\frac{s}{2}\sigma_3}\Psi^-\left(z+\frac{1}{2}\right)=\sigma_1\Psi^-\left(-z+\frac{1}{2}\right)\sigma_1,\qquad \sigma_1=\begin{pmatrix}0&1\\1&0\end{pmatrix}.\label{SymmetryPsi-}\end{gather} It can indeed be verif\/ied that
the left and right hand side satisfy the same, uniquely solvable, RH problem.

For $\tau<0$ and $u\in\mathbb R$, we def\/ine
\begin{gather}\label{def Phi -}
\begin{pmatrix}\Phi_1(u;\tau)\\ \Phi_2(u;\tau)\end{pmatrix}= \Psi^-\left( z=-\frac{2 i u}{s}+\frac{1}{2};s\right)e^{\mp \pi i \alpha /2 \sigma_3}\begin{pmatrix} 1\\-1 \end{pmatrix} \qquad \mbox{for $\pm u >0$},\quad s=\sqrt{-\tau}.
\end{gather}
One verif\/ies using the jump conditions for $\Psi^-$ that $\Phi_1$ and $\Phi_2$ are analytic in a neighborhood of the real line. In fact, the singularities $z=0$ and $z=1$ of $\Psi^-$ are transformed to complex conjugate singularities $u=\pm is/4=\pm i\sqrt{-\tau}/4$ for $\Phi_1$ and $\Phi_2$. We have the asymptotic behavior
\begin{gather}
\begin{pmatrix}
\Phi_1(u;\tau)\\
\Phi_2(u;\tau)
\end{pmatrix}=\left( I +\mathcal{O}\big(u^{-1}\big)\right) \exp\big(\big(iu-\sqrt{-\tau}/4\big)\sigma_3\big)e^{\mp \pi i \alpha/2 \sigma_3}\begin{pmatrix}1\\-1\end{pmatrix}\label{as Phi -}
\end{gather}
as $u\to \pm\infty$.
The kernel $\mathbb K_\alpha^{\rm PV}(u,v;\tau)$ is then given by~\eqref{limker} for~$\tau<0$.

\subsection{Lax pair and Painlev\'e V}

There is a connecton between the RH problem for $\Psi^-$ and the Painlev\'e V equation, which relies on the theory of isomonodromy preserving deformations~\cite{FIKN}. We recall from~\cite{CIK} (in particular Theorem~1.8(ii) and Section~4.3 in that paper) that the RH solution $\Psi=\Psi^-$ solves a~system of linear dif\/ferential equations{\samepage
\begin{gather} \label{Lax}
\Psi_{z}(z; s)= A(z;s) \Psi(z;s),\\
\Psi_s(z;s) = B(z;s)\Psi(z;s), \label{Lax2}
\end{gather}}

\noindent
where the matrices $A$ and $B$ have the form
\begin{gather}
\label{Eqn:A} A(z;s)= A_{\infty}(s)+\frac{A_0(s)}{z} +\frac{A_1(s)}{z-1},\\
 B(z;s)= B_0(s)+B_1(s)z.\label{Eqn:B}
\end{gather}
The $2\times 2$ matrices $A_0$, $A_1$, $A_{\infty}$, $B_0$, $B_1$, which are independent of $z$ but can depend on~$s$ and~$\alpha$, can be parametrized as follows,
\begin{gather}
A_0= \begin{pmatrix} -v+\alpha/2 & uy(v-\alpha) \\ -\frac{v}{uy}&v-\alpha/2\end{pmatrix},\\
A_1= \begin{pmatrix} v-\alpha/2& -y(v-\alpha) \\ \frac{v}{y} &-v+\alpha/2 \end{pmatrix},\\
A_{\infty}=-\frac{s}{2} \sigma_3, \qquad
B_1= -\frac{1}{2}\sigma_3, \qquad
B_0=\frac{A_0+A_1}{s}.\label{B2}
\end{gather}
Here $u$, $v$, and $y$ are functions of $s$ and $\alpha$.
We have the following relation between $u$, $v$ and the functions $q$, $r$, $p$ appearing in the asymptotic expansion~(\ref{Psi as}) for $\Psi^-$ as $z\to\infty$,
\begin{gather}
 \label{v}v=\frac{\alpha}{2}-q-srp,\\
 \label{u}u=1+\frac{sp}{(1-s)p+sp'}.
\end{gather}
The compatibility condition $\Psi_{sz}= \Psi_{z s}$ implies that $A_s-B_{z}+AB-BA=0$, and this implies the following system of equations for $u$,~$v$,~$y$:
\begin{gather}
su_s= su+(\alpha-2v)(u-1)^2,\label{ODE u}\\
sv_s= v\left(u-\frac{1}{u}\right)(v-\alpha),\label{ODE v}\\
sy_s= y\left(-2v+\alpha+uv-u\alpha+\frac{v}{u}-s\right).
\end{gather}
By eliminating $v$ from the top two equations, we obtain a special case of the f\/ifth Painlev\'{e} equation,
\begin{gather} \label{Eqn:PV}
u_{ss}=\left(\frac{1}{2u}+\frac{1}{u-1}\right)u_s^2-\frac{1}{s}u_s
+\frac{(1-u)^2}{s^2}\left(\frac{\alpha^2}{2}\left(u-\frac{1}{u}\right)\right)+\frac{u}{s}-\frac{u(u+1)}{2(u-1)}.
\end{gather}
If we def\/ine $\sigma(s)$ as
\begin{gather}\label{sigma}
\sigma(s)=\int_s^{+\infty}v(\xi)d\xi
\end{gather}
(it was shown in \cite{CIK} that~$v$ decays rapidly as $s\to +\infty$ so that this integral converges),
then~$\sigma$ solves~(\ref{PainleveV}).

\subsection[Recursive construction of $\Psi^-$ if $\alpha\in\mathbb N$]{Recursive construction of $\boldsymbol{\Psi^-}$ if $\boldsymbol{\alpha\in\mathbb N}$}\label{sec: rec-}

For positive integer values of $\alpha$, the RH solution $\Psi^-$ can be constructed explicitly. We will denote~$\Psi^-_\alpha$ instead of~$\Psi^-$ in this section to emphasize the dependence on~$\alpha$. The crucial observation here is that the jump matrices for $\Psi_\alpha^-$ are periodic in~$\alpha$: they remain the same if we replace~$\alpha$ by~$\alpha+2$. This fact indicates that there is a Schlesinger type transformation~\cite{MF} relating~$\Psi_\alpha$ to~$\Psi_{\alpha+2}$. Even more, if we replace $\alpha$ by $\alpha+1$, the jump matrices are modif\/ied in a simple way:
the jumps on the diagonals $\Gamma_1^-$, $\Gamma_2^-$, $\Gamma_3^-$, $\Gamma_4^-$ are the same for the functions $\Psi_{\alpha}^{-}$ and $\sigma_3 \Psi_{\alpha+1}^- \sigma_3$,
whereas they are opposite on $(0,1)$. Based on this observation, and on the fact that the solution for $\alpha=0$ is simple and explicit, we can recursively construct a solution for any integer $\alpha$.

For positive integer $\alpha$ and $s \in (0,+\infty)$, def\/ine $X_\alpha$ as follows,
\begin{gather} \label{DefX}
X_\alpha(z;s)= \begin{cases} \Psi^-_{\alpha}(z;s) & \textrm{for $z \in {\rm II},\, {\rm IV}$,} \\
\Psi^-_{\alpha}(z;s) \begin{pmatrix} 1&(-1)^{\alpha}\\0&1 \end{pmatrix} & \textrm{for $z \in {\rm I}$,} \vspace{1mm}\\
\Psi^-_{\alpha}(z;s) \begin{pmatrix} 1&0\\ (-1)^{\alpha}&1 \end{pmatrix} & \textrm{for $z \in {\rm III}$,}
\end{cases}
\end{gather}
with regions I, II, III, IV as in Fig.~\ref{figure: contour Psi-}. Then, by the RH conditions for $\Psi_\alpha^-$, $X_\alpha$ satisf\/ies the following RH problem.

\subsubsection*{RH problem for $\boldsymbol{X_\alpha}$}
\begin{itemize}\itemsep=0pt
\item[(a)] $X_\alpha\colon \mathbb C{\setminus}[0,1] \to \mathbb C^{2\times 2}$ is analytic.
\item[(b)] For $z\in(0,1)$, we have the jump relation $X_{\alpha,+}(z)=(-1)^\alpha X_{\alpha,-}(z)$.
\item[(c)] As $z \to \infty$,
\begin{gather*} X_\alpha(z)=\big(I+\bigO \big(z^{-1}\big)\big)\exp\left(-\frac{sz}{2}\sigma_3\right).
\end{gather*}
\item[(d)] The functions $G_\alpha$ and $H_\alpha$ def\/ined by
\begin{gather}\label{def Galpha}G_\alpha(z) =X_\alpha(z)\begin{pmatrix}1&0\\-(-1)^{\alpha} &1 \end{pmatrix}z^{-\frac{\alpha}{2}\sigma_3},\\
H_\alpha(z) =X_\alpha(z) \begin{pmatrix}1&-(-1)^{\alpha}\\0&1 \end{pmatrix}(z-1)^{\frac{\alpha}{2}\sigma_3},
\end{gather}
are analytic at $z=0$ and at $z=1$ respectively. The branches in the def\/initions of~$G_{\alpha}$,~$H_{\alpha}$ are such that $z^{ \frac{\alpha}{2}},(z-1)^{\frac{\alpha}{2}}>0$ on the $+$ side of $(1,\infty)$.
\end{itemize}

For $\alpha=0$, this RH problem is easy to solve: we have
\begin{gather*} X_0(z)=\exp\left(-\frac{sz}{2}\sigma_3\right).\end{gather*}
Def\/ine
\begin{gather}\label{defnW}
W_\alpha(z)=\left( \frac{z}{z-1} \right)^{\frac{1}{2} \sigma_3} X_\alpha(z) \sigma_3 X_{\alpha+1}(z)^{-1} \sigma_3,
\end{gather}
where the square root is analytic on $\mathbb C {\setminus} [0,1]$ and positive for large positive~$z$. We will explicitly construct $W_\alpha$ in terms of $G_\alpha$ and $H_\alpha$, and thus we will have the recursive relation
\begin{gather}\label{recrelnX}
X_{\alpha+1}(z)=\sigma_3W_\alpha(z)^{-1}\left( \frac{z}{z-1} \right)^{\frac{1}{2} \sigma_3} X_\alpha(z)\sigma_3 .
\end{gather}
From the jump relations for $X_\alpha$, one verif\/ies that $W_\alpha$ is meromorphic in $z$, with singularities at~$0$ and~$1$, and as $z\to \infty$, we have
\begin{gather} \label{Winfty}
W_\alpha(z)=I+ \mathcal{O} \left(\frac{1}{z} \right).
\end{gather}
By condition (d) of the RH problem for $X_{\alpha}$, we obtain
\begin{alignat}{3}\label{Wat0}
& W_\alpha(z) =\left( \frac{z}{z-1} \right)^{\frac{1}{2} \sigma_3} G_\alpha(z) z^{-\frac{1}{2}\sigma_3}\sigma_3 G_{\alpha+1}(z)^{-1} \sigma_3 \qquad && \textrm{for $z$ near $0$,}& \\
& W_\alpha(z) =\left( \frac{z}{z-1} \right)^{\frac{1}{2} \sigma_3} H_\alpha(z)(z-1)^{\frac{1}{2}\sigma_3}\sigma_3 H_{\alpha+1}(z)^{-1} \sigma_3 \qquad && \textrm{for $z$ near $1$.}& \label{Wat1}
\end{alignat} It follows from (\ref{Winfty})--(\ref{Wat1}) that $W_\alpha$ takes the form
\begin{gather} \label{form for W}
W_\alpha(z)=I+\frac{P_\alpha}{z}+\frac{Q_\alpha}{z-1},
\end{gather}
for matrices $P_\alpha$, $Q_\alpha$ which are independent of~$z$.
Moreover, it is easily verif\/ied that~$X_\alpha$ exhibits the following symmetry,
\begin{gather}\label{Symm X}
X_\alpha\left(z+\frac{1}{2}\right)=e^{-\frac{s}{2}\sigma_3}\sigma_1 X_\alpha\left(-z+\frac{1}{2}\right)\sigma_1,
\end{gather}
which yields the relation
\begin{gather} \label{link PQ}
-e^{-\frac{s}{2}\sigma_3}\sigma_1 P_\alpha\sigma_1 e^{\frac{s}{2}\sigma_3}= Q_\alpha.
\end{gather}
Condition (\ref{Wat0}) gives us that
\begin{gather}
 z^{\frac{1}{2}\sigma_3} G_\alpha(z)^{-1}z^{-\frac{1}{2}\sigma_3} \sigma_3 W_\alpha(z) \quad \textrm{is analytic in a neighborhood of $0$.} \label{Wanalcond00}
\end{gather}
By substituting (\ref{form for W}) into (\ref{Wanalcond00}) it follows that the f\/irst row of $P_\alpha$ is $0$, and consequently, by~(\ref{Wanalcond00}),
\begin{gather}
 z^{\frac{1}{2}\sigma_3} G_\alpha(0)^{-1}z^{-\frac{1}{2}\sigma_3} \sigma_3 W_\alpha(z) \quad \textrm{is analytic in a neighborhood of $0$.} \label{Wanalcond0}
 \end{gather}
This implies
\begin{gather}
G_{\alpha,21}(0)(1+P_{\alpha,22})+G_{\alpha,11}(0)P_{\alpha,21} =0, \label{integereqns1}\\
e^{-s}G_{\alpha,21}(0)P_{\alpha,21}+G_{\alpha,11}(0)P_{\alpha,22} =0. \label{integereqns3}
\end{gather}
By the unique solvability of the RH problem for $X_\alpha$, it follows that this linear system has a~unique solution for $P_{\alpha,22}$ and $P_{\alpha,21}$ for any positive $s$.

In conclusion, since we know $X_0(z)$, we can compute $G_0(0)$ by \eqref{def Galpha}, and from that we can compute $P_0$ by \eqref{integereqns1}, \eqref{integereqns3}. Substituting those in~\eqref{link PQ}, \eqref{form for W}, and \eqref{recrelnX}, we obtain $X_1(z)$. Similarly, once we know $X_1(z)$, we can compute $G_1(0)$, $P_1$, and~$X_2$, and so on.

For $\alpha=1$ and $\alpha=2$, in this way, we obtain the expressions
\begin{gather}
X_1(z) =\frac{1}{(e^s-1)\sqrt{(z-1)z}}\begin{pmatrix} 1+(e^s-1)z & -1 \\ e^s & -e^s+z(e^s-1)\end{pmatrix} e^{-\frac{sz}{2}\sigma_3},\\
X_2(z) =\frac{1}{(1-e^s (2+s^2)+e^{2 s}) (z-1) z}\label{formula X2}\\
\hphantom{X_2(z) =}{} \times \bigg{(}\begin{matrix} 1+ e^s( s-1)
 - 2 z +e^s z(2-2s+s^2)+z^2
-e^s z^2(2+s^2-e^s)
\\e^s(-1 - s +e^s+ z(2+s+e^s(-2+s)))
 \end{matrix}\nonumber \\
\hphantom{X_2(z) =}{} \quad \begin{matrix}1 - e^s + e^s s -z(2+s+e^s(s-2))\\
e^s (-1-s+e^s)+ze^s(2+2s+s^2-2e^s) +z^2 +z^2e^s(-2-s^2+e^s)
 \end{matrix}\bigg{)} e^{-\frac{sz}{2}\sigma_3} .\nonumber
\end{gather}

Expanding $X_\alpha$ as $z\to\infty$, we can obtain expressions for the Painlev\'e~V functions $q(s)$, $r(s)$, $p(s)$ by~(\ref{Psi as}). In particular, for $\alpha=0,1,2$, we obtain
\begin{gather}\label{formulas q}
q_0(s)=0,\qquad q_1(s)=\frac{1}{e^s-1}+\frac{1}{2},\qquad
q_2(s)=\frac{-1+e^s(-2s+e^s)}{1-e^s(2+s^2)+e^{2s}}.
\end{gather}
Using the relation (see \cite[equation~(4.109)]{CIK}) \begin{gather}\sigma_\alpha^-(s)=sq(s)-\frac{\alpha s}{2},\label{sigma q}\end{gather}
we can compute $\sigma_0^-$, $\sigma_1^-$, and $\sigma_2^-$: this gives~\eqref{sigma0-intro}--\eqref{sigma2-intro}.

\subsection[The kernel $\mathbb K_{\alpha}^{{\rm PV}}$ as $\tau \to -\infty$]{The kernel $\boldsymbol{\mathbb K_{\alpha}^{{\rm PV}}}$ as $\boldsymbol{\tau \to -\infty}$} \label{section tau infty}

We follow the steps of \cite[Section~4.1]{CIK} in the analysis as $\tau \to -\infty$. Def\/ine
\begin{gather} \label{tildePsi}
\widetilde \Psi (z)= \Psi^{-}(z) \left(\frac{z-1}{z}\right)^{\frac{\alpha}{2}\sigma_3} ,
\end{gather}
where the power $\alpha/2$ is taken analytic except on $[0,1]$ and tending to $1$ as $z\to +\infty$.
The fraction cancels out the jump of $\Psi^-$ on $(0,1)$ and the singularities at $0$ and $1$, and $\widetilde\Psi$ solves the following RH problem.

\subsubsection*{RH problem for $\boldsymbol{\widetilde\Psi}$}
\begin{itemize}\itemsep=0pt
\item[(a)] $\widetilde\Psi\colon \mathbb{C}{\setminus}(\Gamma_1^-\cup\Gamma_2^-
\cup\Gamma_3^-\cup\Gamma_4^-)$.
\item[(b)] $\widetilde\Psi$ has the jump relations
\begin{alignat*}{3}
&\widetilde\Psi_+(z)=\widetilde\Psi_-(z)\begin{pmatrix} 1 & e^{\pi i\alpha}\left(\frac{z-1}{z}\right)^{-\alpha}\\ 0 &1\end{pmatrix} \qquad &&\textrm{for $z \in \Gamma_1^-$,}&\\
&\widetilde\Psi_+(z)=\widetilde\Psi_-(z)\begin{pmatrix} 1 &0\\ -e^{-\pi i\alpha}\left(\frac{z-1}{z}\right)^\alpha &1\end{pmatrix} \qquad &&\textrm{for $z \in \Gamma_2^-$,}&\\
&\widetilde\Psi_+(z)=\widetilde\Psi_-(z)\begin{pmatrix} 1 &0\\ e^{\pi i\alpha}\left(\frac{z-1}{z}\right)^\alpha &1\end{pmatrix} \qquad &&\textrm{for $z \in \Gamma_3^-$,}&\\
&\widetilde\Psi_+(z)=\widetilde\Psi_-(z)\begin{pmatrix} 1 & -e^{-\pi i\alpha}\left(\frac{z-1}{z}\right)^{-\alpha}\\0 &1\end{pmatrix} \qquad &&\textrm{for $z \in \Gamma_4^-$.}&
\end{alignat*}
\item[(c)]As $z\to \infty$,
\begin{gather}\label{Psi as tilde}
\widetilde\Psi(z)=\left(I+\mathcal{O}\big(z^{-1}\big)\right)\exp\left(-\frac{s}{2}z\sigma_3\right).
\end{gather}
\end{itemize}

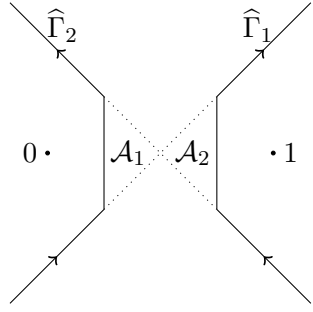
\begin{figure}\centering
\begin{tikzpicture}
\draw [decoration={markings, mark=at position 0.5 with {\arrow[thick]{>}}},
 postaction={decorate},] (-2,-2) -- (-0.75,-0.75) ;
\draw (-0.75,-0.75) -- (-0.75,0.75);
\draw [decoration={markings, mark=at position 0.5 with {\arrow[thick]{>}}},
 postaction={decorate},] (-0.75,0.75) -- (-2,2);

\draw [decoration={markings, mark=at position 0.5 with {\arrow[thick]{>}}},
 postaction={decorate},] (2,-2) -- (0.75,-0.75) ;
\draw (0.75,-0.75) -- (0.75,0.75);
\draw [decoration={markings, mark=at position 0.5 with {\arrow[thick]{>}}},
 postaction={decorate},] (0.75,0.75) -- (2,2);

\draw [dotted] (-0.75,-0.75)--(0.75,0.75);
\draw [dotted] (-0.75,0.75)--(0.75,-0.75);

\draw[fill] (-1.5,0) circle [radius=0.025];
\draw[fill] (1.5,0) circle [radius=0.025];
\node [right] at (1.5,0) {$1$};
\node [left] at (-1.5,0) {$0$};
\node at (0.43,0) {$\mathcal A_2$};
\node at (-0.43,0) {$\mathcal A_1$};
\node at (1.3,+1.7) {$\widehat\Gamma_1$};
\node at (-1.3,+1.7) {$\widehat\Gamma_2$};

\end{tikzpicture}
\caption{Contour $\widehat \Gamma$.}\label{Gammahat}
\end{figure}

Now we will deform the jump contour in such a way that it stays away from $0$, $1$, and $1/2$, to a contour as shown in Fig.~\ref{Gammahat}.
Def\/ine
\begin{gather} \label{hatPsi}
\widehat \Psi (z)=\begin{cases}e^{\frac{s}{4}\sigma_3}\widetilde \Psi(z) \exp \left(\dfrac{s}{2}\left(z-\dfrac{1}{2}\right)\sigma_3\right)&\textrm{for $z \notin \mathcal A_1 \cup \mathcal A_2$,}\vspace{1mm}\\
e^{\frac{s}{4}\sigma_3}\widetilde \Psi(z) \begin{pmatrix}1&0\\\left(\frac{1-z}{z}\right)^\alpha&1\end{pmatrix}\exp \left(\dfrac{s}{2}\left(z-\dfrac{1}{2}\right)\sigma_3\right)&\textrm{for $z \in \mathcal A_1 $,}\vspace{1mm}\\
e^{\frac{s}{4}\sigma_3}\widetilde \Psi(z) \begin{pmatrix}1&\left(\frac{1-z}{z}\right)^{-\alpha}\\0&1\end{pmatrix}\exp \left(\dfrac{s}{2}\left(z-\dfrac{1}{2}\right)\sigma_3\right)&\textrm{for $z \in \mathcal A_2$,}
\end{cases} \end{gather}
with $\mathcal A_1$ and $\mathcal A_2$ given in Fig.~\ref{Gammahat}. Here, $\left(\frac{1-z}{z}\right)^{\pm\alpha}$ is analytic except on $(-\infty,0]\cup[1,+\infty)$.
Then~$\widehat \Psi$ satisf\/ies a small norm RH problem as $s\to +\infty$.

\subsubsection*{RH problem for $\boldsymbol{\widehat\Psi}$}
\begin{itemize}\itemsep=0pt
\item[(a)] $\widehat\Psi\colon \mathbb C {\setminus} \widehat \Gamma \to \mathbb C^{2\times2}$ is analytic, with $\widehat \Gamma=\widehat\Gamma_1\cup\widehat\Gamma_2$ given in Fig.~\ref{Gammahat}.
\item[(b)] $\widehat \Psi$ has the following jumps on $\widehat \Gamma$:
\begin{alignat*}{3}
&\widehat\Psi_+(z)=\widehat\Psi_-(z)\begin{pmatrix} 1 & \left(\frac{1-z}{z}\right)^{-\alpha}e^{-s(z-\frac{1}{2})}\\ 0 &1\end{pmatrix} \qquad &&\textrm{for $z \in \widehat\Gamma_1$,}&\\
&\widehat\Psi_+(z)=\widehat\Psi_-(z)\begin{pmatrix} 1 &0\\ -\left(\frac{1-z}{z}\right)^\alpha e^{s(z-\frac{1}{2})} &1\end{pmatrix} \qquad &&\textrm{for $z \in \widehat\Gamma_2$}.&
\end{alignat*}
As $s\to +\infty$, this means that
\begin{gather*} \widehat \Psi _+(z)=\widehat \Psi_-(z)(I+\mathcal O \left(\exp \left(-cs\right)\right) \end{gather*}
on $\widehat \Gamma_1$ and $\widehat \Gamma_2$, for some constant $c>0$.
\item[(c)] $\widehat \Psi(z)=I+\mathcal O (z^{-1})$ as $z \to \infty$.
\end{itemize}
By standard small norm analysis, we have
\begin{gather}\label{small norm hatPsi}\widehat \Psi (z)=I +\mathcal O\left(\frac{1}{|z|+1}e^{-cs}\right),\qquad s\to+\infty,\end{gather}
uniformly for $z\in\mathbb C{\setminus}\widehat\Gamma$.
Inverting the transformations $\Psi^-\mapsto\widetilde\Psi\mapsto\widehat\Psi$ and using \eqref{def Phi -}, we obtain
\begin{gather*}
\begin{pmatrix}
\Phi_1\big(x;-s^2\big)\\ \Phi_2\big(x;-s^2\big)
\end{pmatrix}=\widehat\Psi\left(\frac{-2ix}{s}+\frac{1}{2};s\right)\left(\frac{-4ix-s}{-4ix+s}\right)^{-\frac{\alpha}{2}\sigma_3}e^{\mp \frac{\pi i\alpha}{2}\sigma_3}e^{ix\sigma_3}\begin{pmatrix}
1\\ -1
\end{pmatrix},\qquad \pm x>0.
\end{gather*}
By \eqref{small norm hatPsi}, as $s\to +\infty$, $x\in\mathbb R$,
\begin{gather*}
\begin{pmatrix}
\Phi_1(\pi u;-s^2)\\ \Phi_2(\pi u;-s^2)
\end{pmatrix}=\begin{pmatrix} e^{ \pi i u}\\ -e^{- \pi i u}\end{pmatrix}+\mathcal{O}\big(s^{-1}\big),
\end{gather*}
which, by \eqref{limker}, yields the sine kernel limit as $s\to +\infty$,
\begin{gather*}
\mathbb K_{\alpha}^{{\rm PV}}\big(u,v;-s^2\big) = \frac{ \sin \pi(u-v)}{\pi(u-v)}+\mathcal O \big(s^{-1}\big) ,
\end{gather*}
and (\ref{lim sine}) is proved. In fact, it is clear from this derivation that \eqref{lim sine} holds also as~$u$,~$v$ tend to inf\/inity together with $s$, as long as $u-v$ converges.

\subsection[The kernel $\mathbb K_{\alpha}^{{\rm PV}}$ as $\tau \to 0$]{The kernel $\boldsymbol{\mathbb K_{\alpha}^{{\rm PV}}}$ as $\boldsymbol{\tau \to 0}$} \label{section tau 0}

We will now study the RH problem for $\Psi^-$ in the limit where $s\to 0$ and prove \eqref{lim Bessel}.

Def\/ine \begin{gather}\label{def Psi1}\Psi^{(1)}(z;s)=e^{\frac{s}{4}\sigma_3}\Psi^-\left(-\frac{2iz}{s}+\frac{1}{2};s\right)\chi(z),
\end{gather} where
\begin{gather}\chi(z)=\begin{cases} e^{-\frac{\pi i \alpha}{2}\sigma_3}&\textrm{for $\Re z>0$,}\\
e^{\frac{\pi i \alpha}{2}\sigma_3}&\textrm{for $\Re z<0$.}\end{cases}\label{def chi}
\end{gather}
Then $\Psi^{(1)}$ satisf\/ies the following RH problem.
\subsubsection*{RH problem for $\boldsymbol{\Psi^{(1)}}$}
\begin{figure}[t]\centering
\begin{picture}(120,138)(-5,-70)
 \put(48,-7){$0$}
\put(36,-34){$-\frac{is}{4}$}
\put(44,20){$\frac{is}{4}$}
\put(75,45){$\widetilde\Gamma^-_{1}$}
\put(-5,45){$\widetilde\Gamma^-_2$}
\put(-5,-62){$\widetilde\Gamma^-_3$}
\put(75,-62){$\widetilde\Gamma^-_4$}
\put(46,-70){$\widetilde\Gamma^-_5$}

\put(86,27){$\begin{pmatrix}1&-1\\0&1\end{pmatrix}$}
\put(-47,27){$\begin{pmatrix}1&-1\\0&1\end{pmatrix} $}
\put(-42,-46){$\begin{pmatrix}1&0\\1&1\end{pmatrix} $}
\put(89,-46){$\begin{pmatrix}1&0\\1&1\end{pmatrix} $}
\put(46,60){$e^{\pi i \alpha \sigma_3}$}

\put(42,25){\line(0,1){45}}
\put(42,-35){\line(0,-1){45}}
\put(-18,-65){\line(1,1){120}}
\put(-18,55){\line(1,-1){120}}

\put(42,25){\thicklines\circle*{2.5}}
\put(42,-5){\thicklines\circle*{2.5}}
\put(42,-35){\thicklines\circle*{2.5}}

\put(42,60){\thicklines\vector(0,1){.0001}}
\put(42,-60){\thicklines\vector(0,1){.0001}}
\put(8,-39){\thicklines\vector(1,1){.0001}}
\put(76,29){\thicklines\vector(1,1){.0001}}
\put(8,29){\thicklines\vector(1,-1){.0001}}
\put(76,-39){\thicklines\vector(1,-1){.0001}}

\end{picture}
\caption{The jump contour $\widetilde\Gamma^-$ and the jump matrices for $\Psi^{(1)}$.}\label{figure: contour Psi(1)}
\end{figure}
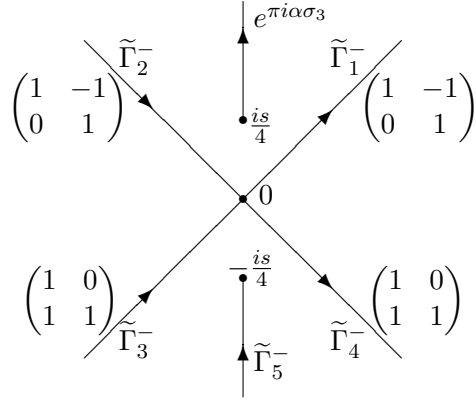

\begin{itemize}\itemsep=0pt
\item[(a)] $\Psi^{(1)}\colon \mathbb{C}{\setminus} \widetilde\Gamma^- \to\mathbb{C}^{2\times 2}$ is analytic, where
\begin{alignat*}{3}
&\widetilde\Gamma= \bigcup_{i=1}^5 \widetilde\Gamma_i^-, \qquad && \widetilde \Gamma_1^-=\{z\colon \arg z=\pi/4\},&\\
& \widetilde\Gamma_2^-=\{z\colon \arg z=3\pi/4\},\qquad &&
\widetilde\Gamma_3^-=\{z\colon \arg z=5\pi/4\},& \\
& \widetilde\Gamma_4^-=\{z\colon \arg z=7\pi/4\}, \qquad && \widetilde\Gamma_5^-=i\mathbb R{\setminus} [-is/4,is/4].&
\end{alignat*}
The orientation is as in Fig.~\ref{figure: contour Psi(1)}. Note that the orientation is dif\/ferent compared to the one in the jump contour for $\Psi^-$.
\item[(b)] $\Psi^{(1)}$ has continuous boundary values on $\widetilde\Gamma^-{\setminus} \{0\}$:
\begin{alignat*}{3}
&{\Psi}^{(1)}_+(z)={\Psi}^{(1)}_-(z)\begin{pmatrix}1&-1\\0&1\end{pmatrix} \qquad &&\textrm{for $z\in \widetilde\Gamma_1^-\cup \widetilde\Gamma_2^-$ ,}& \\
&{\Psi}^{(1)}_+(z)={\Psi}^{(1)}_-(z)\begin{pmatrix}1&0\\1&1\end{pmatrix} \qquad &&\textrm{for $z\in \widetilde\Gamma_3^-\cup \widetilde\Gamma_4^-$, }& \\
&\Psi^{(1)}_+(z)=\Psi^{(1)}_-(z)e^{\pi i \alpha \sigma_3} \qquad &&\textrm{for $z\in \widetilde\Gamma_5^-$.}&
\end{alignat*}
\item[(c)] $\Psi^{(1)}$ has the following behavior as $z\to \infty$:
\begin{gather} \label{Psi1 as}
\Psi^{(1)}(z)=(I+\mathcal{O}(z^{-1}))e^{iz\sigma_3}\chi(z).
\end{gather}
\item[(d0)] $\Psi^{(1)}$ has the following behavior as $z\to -is/4$:
\begin{gather} \nonumber \Psi^{(1)}(z)= \begin{pmatrix} \mathcal{O}(|z+\frac{is}{4}|^{\alpha/2})&\mathcal{O}(|z+\frac{is}{4}|^{-\alpha/2})\\ \mathcal{O}(|z+\frac{is}{4}|^{\alpha/2})&\mathcal{O}(|z+\frac{is}{4}|^{-\alpha/2})\end{pmatrix} .\end{gather}
\item[(d1)] $\Psi^{(1)}$ has the following behavior as $z\to is/4$:
\begin{gather} \nonumber \Psi^{(1)}(z)= \begin{pmatrix} \mathcal{O}(|z-\frac{is}{4}|^{-\alpha/2})&\mathcal{O}(|z-\frac{is}{4}|^{\alpha/2})\\ \mathcal{O}(|z-\frac{is}{4}|^{-\alpha/2})&\mathcal{O}(|z-\frac{is}{4}|^{\alpha/2})\end{pmatrix}. \end{gather}
\item[(e)] $\Psi^{(1)}$ is bounded at $0$.
\end{itemize}
Using \eqref{def Phi -} and \eqref{def Psi1}, we obtain
\begin{gather}\label{PhiPsi1}
\begin{pmatrix}
\Phi_1(x;-s^2)\\ \Phi_2(x;-s^2)
\end{pmatrix}=e^{-\frac{s}{4}\sigma_3}\Psi^{(1)}(x;s)\begin{pmatrix}
1\\ -1
\end{pmatrix}.
\end{gather}

To study $\Psi^{(1)}(z;s)$ as $s\to 0$, we need to recall results from~\cite[Sections~4.2.1 and~4.2.2]{CIK}. The asymptotics as $s\to 0$ can be described
as follows: we have for $|z|>\delta$ with $\delta>0$ arbitrary that
\begin{gather} \label{intrepM}
\Psi^{(1)}(z;s)=\left(I+\mathcal{O}(|s\log|s||)+\mathcal{O}\big(s^{1+2\alpha}\big)\right)M(z)e^{\mp \frac{ \pi i \alpha}{2}\sigma_3} \qquad\text{for $\pm \Re{z}>0$}.
\end{gather}
Here $M$ satisf\/ies the following RH conditions.

\subsubsection*{RH problem for $\boldsymbol{M}$}
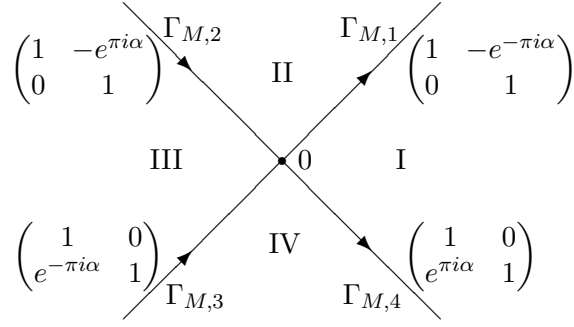
\begin{figure}[t]\centering
\begin{picture}(120,120)(-5,-65)
 \put(48,-8){$0$}
\put(85,-8){I}
\put(38,25){II}
\put(-8,-8){III}
\put(37,-40){IV}

\put(64,42){$\Gamma_{M,1}$}
\put(-3,42){$\Gamma_{M,2}$}
\put(-2,-59){$\Gamma_{M,3}$}
\put(65,-59){$\Gamma_{M,4}$}

\put(88,27){$\begin{pmatrix}1&-e^{-\pi i \alpha}\\0&1\end{pmatrix}$}
\put(-61,27){$\begin{pmatrix}1&-e^{\pi i \alpha}\\0&1\end{pmatrix}$}
\put(-60,-44){$\begin{pmatrix}1&0\\e^{-\pi i \alpha}&1\end{pmatrix}$}
\put(88,-44){$\begin{pmatrix}1&0\\e^{\pi i \alpha}&1\end{pmatrix} $}

\put(-18,-65){\line(1,1){120}}
\put(-18,55){\line(1,-1){120}}

\put(42,-5){\thicklines\circle*{2.5}}

\put(8,-39){\thicklines\vector(1,1){.0001}}
\put(76,29){\thicklines\vector(1,1){.0001}}
\put(8,29){\thicklines\vector(1,-1){.0001}}
\put(76,-39){\thicklines\vector(1,-1){.0001}}

\end{picture}
\caption{The contour $\Gamma_M$ and the jump matrices for $M$.}
\label{ContourGammaM}
\end{figure}
\begin{itemize}\itemsep=0pt
\item[(a)] $M\colon \mathbb{C}{\setminus} \Gamma_M \to\mathbb{C}^{2\times 2}$ is analytic, where
\begin{alignat*}{3} &\Gamma_M= \cup_{i=1}^5 \Gamma_{M,i},&&& \\
& \Gamma_{M,1}=\{z\colon \arg(z)=\pi/4\}, \qquad && \Gamma_{M,2}=\{z\colon \arg(z)=3\pi/4\},& \\
&\Gamma_{M,3}=\{z\colon \arg(z)=5\pi/4\},\qquad && \Gamma_{M,4}=\{z\colon \arg(z)=7\pi/4\}.&
\end{alignat*}
The orientation is as in Fig.~\ref{ContourGammaM}.
\item[(b)] $M$ has continuous boundary values on $\Gamma_M{\setminus} \{0\}$:
\begin{alignat*}{3}
&M_+(z)=M_-(z)\begin{pmatrix}1&-e^{-\pi i \alpha}\\0&1\end{pmatrix} \qquad &&\textrm{for $z\in \Gamma_{M,1}$,}& \\
&M_+(z)=M_-(z)\begin{pmatrix}1&-e^{\pi i \alpha}\\0&1\end{pmatrix} \qquad &&\textrm{for $ z\in\Gamma_{M,2}$,}& \\
&M_+(z)={M}_-(z)\begin{pmatrix}1&0\\e^{-\pi i \alpha}&1\end{pmatrix} \qquad &&\textrm{for $z\in \Gamma_{M,3}$, }& \\
&M_+(z)={M}_-(z)\begin{pmatrix}1&0\\e^{\pi i \alpha}&1\end{pmatrix} \qquad &&\textrm{for $z\in \Gamma_{M,4}$. }&
\end{alignat*}
\item[(c)] $M$ has the following behavior as $z\to \infty$:
\begin{gather*}
M(z)=\big(I+\mathcal{O}\big(z^{-1}\big)\big)e^{i{z}\sigma_3}.
\end{gather*}
\end{itemize}
The RH problem for $M$ has an explicit solution (which is not unique, unless one adds a condition to control the behaviour of $M$ near $0$) given in terms of conf\/luent hypergeometric functions which is used in \cite{CIK,CK}. But in the special case we are dealing with, the conf\/luent hypergeometric functions degenerate to Hankel and modif\/ied Bessel functions $H^{(1)}_{\alpha}$, $H^{(2)}_{\alpha}$, $I_{\alpha}$ and $K_{\alpha}$ (see \cite{NIST}), and we use an explicit formula for $M$ similar to the one given in \cite{2, Vanlessen}. Writing
\begin{gather} \nonumber
K=\frac{1}{\sqrt{2}} \begin{pmatrix}-1&1\\-1&-1\end{pmatrix}e^{-\frac{\pi i}{4}\sigma_3},
\end{gather}
we have
\begin{gather}\label{DefnM}
M(z)=\begin{cases}
\frac{z^{\frac{1}{2}}\sqrt{\pi}}{2}K \begin{pmatrix}
H_{\alpha+\frac{1}{2}}^{(2)}(z)&-iH_{\alpha+\frac{1}{2}}^{(1)}(z)\\
H_{\alpha-\frac{1}{2}}^{(2)}(z)&-iH_{\alpha-\frac{1}{2}}^{(1)}(z)\end{pmatrix}e^{-\frac{\pi i }{2}(\alpha+\frac{1}{2})\sigma_3}\sigma_1\sigma_3 & \textrm{for $z\in {\rm I}$,}\\
K\begin{pmatrix}z^{\frac{1}{2}}\sqrt{\pi }I_{\alpha+\frac{1}{2}}(ze^{-\frac{\pi i }{2}})&-z^{\frac{1}{2}}\frac{1}{\sqrt{\pi}}K_{\alpha+\frac{1}{2}}(ze^{-\frac{\pi i }{2}})\\
-iz^{\frac{1}{2}}\sqrt{\pi }I_{\alpha-\frac{1}{2}}(ze^{-\frac{\pi i }{2}})&-iz^{\frac{1}{2}}\frac{1}{\sqrt{\pi}} K_{\alpha-\frac{1}{2}}(ze^{-\frac{\pi i }{2}})\end{pmatrix}\sigma_1\sigma_3 &\textrm{for $z\in {\rm II}$,}\\
(e^{\pi i}z)^{\frac{1}{2}}\frac{\sqrt{\pi}}{2}K \begin{pmatrix}
-H_{\alpha+\frac{1}{2}}^{(2)}(e^{\pi i}z)&-iH_{\alpha+\frac{1}{2}}^{(1)}(e^{\pi i}z)\\
H_{\alpha-\frac{1}{2}}^{(2)}(e^{\pi i}z)&iH_{\alpha-\frac{1}{2}}^{(1)}(e^{\pi i}z)\end{pmatrix}e^{-\frac{\pi i }{2}(\alpha+\frac{1}{2})\sigma_3}&\textrm{for $z\in {\rm III}$,}\\
K\begin{pmatrix}-iz^{\frac{1}{2}}\sqrt{\pi }I_{\alpha+\frac{1}{2}}(ze^{\frac{\pi i }{2}})&-iz^{\frac{1}{2}}\frac{1}{\sqrt{\pi}}K_{\alpha+\frac{1}{2}}(ze^{\frac{\pi i }{2}})\\
z^{\frac{1}{2}}\sqrt{\pi}I_{\alpha-\frac{1}{2}}(ze^{\frac{\pi i }{2}})&-z^{\frac{1}{2}}\frac{1}{\sqrt{\pi}} K_{\alpha-\frac{1}{2}}(ze^{\frac{\pi i }{2}})\end{pmatrix} &\textrm{for $z\in {\rm IV}$,}
\end{cases}\hspace{-10mm}
\end{gather}
where the square roots and the Hankel and Bessel functions have branch cuts on $(-\infty,0]$.

By (\ref{intrepM}) and (\ref{PhiPsi1}), we have, for $\pm x>0$,
\begin{gather*}
\begin{pmatrix}
\Phi_1\big(x;-s^2\big)\\ \Phi_2\big(x;-s^2\big)
\end{pmatrix}=M(x;s)e^{\mp\frac{\pi i\alpha}{2}\sigma_3}\begin{pmatrix}
1\\ -1
\end{pmatrix}+\mathcal{O}(|s|\log|s|)+\mathcal{O}\big(|s|^{1+2\alpha}\big),\qquad s\to 0.
\end{gather*}
Using \eqref{limker} and substituting in (\ref{DefnM}), one obtains that as $\tau\to 0$:
\begin{gather*}
\mathbb K_{\alpha}^{{\rm PV}}(u,v;\tau)=\frac{\pi\sqrt{uv}}{8(u-v)}\Bigg( \left(H^{(1)}_{\alpha-\frac{1}{2}}(\pi v)+H^{(2)}_{\alpha-\frac{1}{2}}(\pi v)\right)\left(H^{(1)}_{\alpha+\frac{1}{2}}(\pi u)+H^{(2)}_{\alpha+\frac{1}{2}}(\pi u)\right)\\
- \left(H^{(1)}_{\alpha+\frac{1}{2}}(\pi v)+H^{(2)}_{\alpha+\frac{1}{2}}(\pi v)\right)\!\left(H^{(1)}_{\alpha-\frac{1}{2}}(\pi u)+H^{(2)}_{\alpha-\frac{1}{2}}(\pi u)\right)\Bigg)\!+\mathcal{O}\big(|\tau|^{1/2}\log|\tau|\big)+\mathcal{O}\big(|\tau|^{\frac{1}{2}+\alpha}\big).
\end{gather*}
The identity $H_{\alpha}^{(1)}(z)+H_{\alpha}^{(2)}(z)=2J_{\alpha}(z)$ yields our result: as $\tau\to 0$ we have
\begin{gather}
\mathbb K_{\alpha}^{{\rm PV}}(u,v;\tau)=\frac{\pi\sqrt{uv}}{2(u-v)} \left(J_{\alpha-\frac{1}{2}}(\pi v)J_{\alpha+\frac{1}{2}}(\pi u)
- J_{\alpha+\frac{1}{2}}(\pi v)J_{\alpha-\frac{1}{2}}(\pi u)\right)\nonumber\\
\hphantom{\mathbb K_{\alpha}^{{\rm PV}}(u,v;\tau)=}{}
+\mathcal{O}\big(|\tau|^{1/2}\log|\tau|\big)+\mathcal{O}\big(|\tau|^{\frac{1}{2}+\alpha}\big).\label{limBessel<}
\end{gather}
This proves \eqref{lim Bessel}.

\section[Model RH problem for $\tau>0$]{Model RH problem for $\boldsymbol{\tau>0}$}\label{sec: PV2}
 The model RH problem relevant for the case $\tau>0$ is a special case of the one found in~\cite{CK} (up to a simple transformation), and dif\/fers from the one for $\tau<0$, although there are similarities. For instance, the following RH problem, is also connected to the Painlev\'e V equation, and an explicit solution can be constructed for even~$\alpha$ (but not for odd $\alpha$, which was possible when~$\tau<0$). In the RH problem below, the relevant values for the parameters are $\alpha>-1/2$ and $s\in (0,-i\infty)$.

\subsubsection*{RH problem for $\boldsymbol{\Psi^+}$}
\begin{figure}[t]\centering
\begin{picture}(120,112)(-5,-55)

\put(-18,0){$0$}
\put(68,0){$1$}
\put(102,42){${\Gamma}^+_1$}
\put(-61,42){${\Gamma}^+_2$}
\put(-61,-60){${\Gamma}^+_3$}
\put(102,-60){${\Gamma}^+_4$}
\put(27,0){${\Gamma}^+_5$}

\put(100,-8){I}
\put(27,20){II}
\put(-50,-8){III}
\put(27,-50){IV}

\put(115,25){$\begin{pmatrix} 1 & -e^{-\pi i\alpha}\\ 0 &1\end{pmatrix}$}
\put(-119,25){$\begin{pmatrix} 1 &-e^{\pi i\alpha}\\0 &1\end{pmatrix}$}
\put(-119,-43){$\begin{pmatrix} 1 &0\\ e^{-\pi i\alpha} &1\end{pmatrix}$}
\put(117,-43){$\begin{pmatrix} 1 &0\\e^{\pi i\alpha} &1\end{pmatrix} $}
\put(12,-25){$\begin{pmatrix}1&-1\\ 1&0\end{pmatrix}$}

\put(-18,-5){\line(1,0){90}}
\put(-78,-65){\line(1,1){60}}
\put(72,-5){\line(1,1){60}}
\put(-78,55){\line(1,-1){60}}
\put(72,-5){\line(1,-1){60}}

\put(-18,-5){\thicklines\circle*{2.5}}
\put(72,-5){\thicklines\circle*{2.5}}

\put(32,-5){\thicklines\vector(1,0){.0001}}

\put(-52,-39){\thicklines\vector(1,1){.0001}}
\put(106,29){\thicklines\vector(1,1){.0001}}
\put(-52,29){\thicklines\vector(1,-1){.0001}}
\put(106,-39){\thicklines\vector(1,-1){.0001}}

\end{picture}
\caption{The jump contour $\Gamma^+$ and the jump matrices for $\Psi^+$.}\label{figure: contour Psi+}
\end{figure}
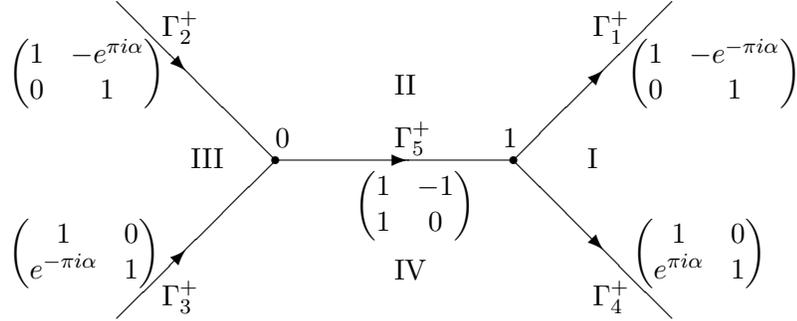

\begin{itemize}\itemsep=0pt
\item[(a)] $\Psi^+\colon \mathbb{C}{\setminus} \Gamma^+ \rightarrow \mathbb{C}^{2\times 2}$ is analytic, with $\Gamma^+=\Gamma_1^+\cup\Gamma_2^+\cup\Gamma_3^+\cup\Gamma_4^+\cup[0,1]$, and
\begin{gather*}\Gamma_1^+=1+e^{\frac{\pi i}{4}}\mathbb R^+,\qquad \Gamma_2^+=e^{\frac{3\pi i}{4}}\mathbb R^+,\qquad \Gamma_3^+=e^{-\frac{3\pi i}{4}}\mathbb R^+,\qquad \Gamma_4^+=1+e^{-\frac{\pi i}{4}}\mathbb R^+,
\end{gather*}
oriented as in Fig.~\ref{figure: contour Psi+}.
\item[(b)]
$\Psi^+$ has continuous boundary values $\Psi_\pm^+$ on $\Gamma^+ {\setminus} \{0,1\}$. We have the jump relations
\begin{alignat*}{3}
&\Psi_+^+(z)=\Psi_-^+(z)\begin{pmatrix} 1 & -e^{-\pi i\alpha}\\ 0 &1\end{pmatrix} \qquad &&\textrm{for $z \in \Gamma_1^+$,}&\\
&\Psi_+^+(z)=\Psi_-^+(z)\begin{pmatrix} 1 &-e^{\pi i\alpha}\\0 &1\end{pmatrix} \qquad &&\textrm{for $z \in \Gamma_2^+$,}&\\
&\Psi_+^+(z)=\Psi_-^+(z)\begin{pmatrix} 1 &0\\ e^{-\pi i\alpha} &1\end{pmatrix} \qquad &&\textrm{for $z \in \Gamma_3^+$,}&\\
&\Psi_+^+(z)=\Psi_-^+(z)\begin{pmatrix} 1 &0\\e^{\pi i\alpha} &1\end{pmatrix} \qquad &&\textrm{for $z \in \Gamma_4^+$,}&\\
&\Psi_+^+(z)=\Psi_-^+(z)\begin{pmatrix}1&-1\\ 1&0\end{pmatrix}\qquad &&\textrm{for $z \in (0,1)$.}&
\end{alignat*}
\item[(c)]There exist $q$, $r$, $p$ (depending on $s$, $\alpha$ but not on $z$) such that~$\Psi^+(z)$ has the following limiting behavior as $z\to \infty$:
\begin{gather*}
\Psi^+(z)=\left(I+\frac{1}{z}\begin{pmatrix}q&r\\ p&- q\end{pmatrix}+\mathcal{O}\big(z^{-2}\big)\right)\exp\left(-\frac{s}{2}z\sigma_3\right).
\end{gather*}
\item[(d)]If $\alpha \notin \mathbb{N}$ then the functions
\begin{gather}G(z):= \Psi^+(z)\begin{pmatrix}1&g\\0&1\end{pmatrix}z^{-\frac{\alpha}{2}\sigma_3}, \label{Psi+d1} \\
H(z):= \Psi^+(z)\begin{pmatrix}1&e^{-\pi i \alpha}g\\0&1\end{pmatrix}(z-1)^{-\frac{\alpha}{2}\sigma_3}\label{Psi+d2}
\end{gather} are bounded for $z$ near $0$ and $1$ respectively when $z$ is in region~IV, with $g=\frac{e^{\pi i \alpha}-1}{2i \sin(\pi \alpha)}$. We let $z^{\frac{\alpha}{2}},(z-1)^{\frac{\alpha}{2}}>0$ on the~$+$ side of $(1,\infty)$. If $\alpha\in \mathbb{N}$ then the functions
\begin{gather}
G(z):= \Psi^+(z)\begin{pmatrix}1&\frac{1-e^{\pi i \alpha}}{2\pi i}\log z\\0&1\end{pmatrix}z^{-\frac{\alpha}{2}\sigma_3},\\
H(z):= \Psi^+(z)\begin{pmatrix}1&\frac{e^{-\pi i \alpha}-1}{2\pi i} \log(z-1)\\0&1\end{pmatrix}(z-1)^{-\frac{\alpha}{2}\sigma_3}
\label{Hint}
\end{gather}
are bounded for $z$ near $0$ and $1$ respectively when $z$ is in region~IV. We let $\log z, \log (z-1)$ $>0$ on the~$+$ side of $(1,\infty)$.
\end{itemize}
The above RH problem is equivalent to the one from \cite[Section~3]{CK}. Indeed, if we def\/ine
\begin{gather}\label{LinkCK}
\Psi_{\rm CK}(\zeta;s)= e^{\frac{s}{4}\sigma_3}\Psi^+\left(\frac{i}{2}(\zeta-i);s\right),
\end{gather}
then $\Psi_{\rm CK}$ solves the RH problem in \cite[Section 3]{CK} in the case $\beta_1=\beta_2=0$, $\alpha_1=\alpha_2=\alpha/2$.

In \cite{CK}, it was proved that the RH problem for $\Psi_{\rm CK}$ is solvable if $\alpha>-1/2$ for every $s\in (0,-i\infty)$, and it follows that the same is true for the RH problem for $\Psi^+$.

For $\tau>0$ and $u\in\mathbb R$, we def\/ine
\begin{gather}\label{def Phi +}
\begin{pmatrix}\Phi_1(u;\tau)\\ \Phi_2(u;\tau)\end{pmatrix}= \Psi^+\left(z=\frac{2 u}{|s|}+\frac{1}{2};s\right)\Delta\left(\frac{2 u}{|s|}+\frac{1}{2}\right), \qquad s=-i\sqrt{\tau},
\end{gather}
where
\begin{gather}\label{Delta}
\Delta(z)=\begin{cases}
\begin{pmatrix}e^{-\frac{\pi i \alpha}{2}}\\ -e^{\frac{\pi i \alpha}{2}}\end{pmatrix}&\textrm{for $\Re z>1$,}\vspace{1mm}\\
\begin{pmatrix}0\\-1\end{pmatrix}&\textrm{for $0<\Re z<1$,}\vspace{1mm}\\
\begin{pmatrix}e^{\frac{\pi i \alpha}{2}}\\ -e^{-\frac{\pi i \alpha}{2}}\end{pmatrix} &\textrm{for $\Re z<0$,}
\end{cases}
\end{gather}
and where $\Psi^+(z;s)$ has to be understood as $\Psi_+^+(z;s)$ for $z\in (0,1)$.
Then, $\Phi_1(u;\tau), \Phi_2(u;\tau)$ are smooth on the real line, except at the singularities $\pm\frac{\sqrt{\tau}}{4}$.
As $u\to \pm\infty$, $\Phi_1$ and $\Phi_2$ have the asymptotic behavior
\begin{gather}\label{as Phi +}
\begin{pmatrix}
\Phi_1(u;\tau)\\
\Phi_2(u;\tau)
\end{pmatrix}=\left( I +\mathcal{O}\big(u^{-1}\big)\right) \exp((iu+i\sqrt{\tau}/4)\sigma_3)\Delta\left(\frac{2 u}{\sqrt{\tau}}+\frac{1}{2}\right),
\end{gather}
similar to \eqref{as Phi -}.

\subsection{Lax pair and Painlev\'e~V}
The results about the Lax pair and the relation with the Painlev\'e V equation hold for $\Psi^+$ as well as for $\Psi^-$. Indeed, $\Psi=\Psi^+$ satisf\/ies the Lax pair (\ref{Lax})--(\ref{Lax2}), where $A$ and $B$ can be parameterized as in (\ref{Eqn:A})--(\ref{B2}). The functions $v$, $u$ def\/ined in (\ref{v}), (\ref{u}) in terms of $q,r,p$, solve equations (\ref{ODE u}), (\ref{ODE v}), and $\sigma$ def\/ined as in (\ref{sigma}) solves the $\sigma$-form (\ref{PainleveV}) of Painlev\'e V. It has to be noted that the functions $q(s)$, $r(s)$, $p(s)$, $v(s)$, $u(s)$, $\sigma(s)$ def\/ined here for $\tau>0$ are in general dif\/ferent from their counterparts in Section~\ref{section: PV} for $\tau<0$.

\subsection[Recursive construction for $\alpha\in 2\mathbb N$]{Recursive construction for $\boldsymbol{\alpha\in 2\mathbb N}$}\label{sec: rec+}

Similarly to the case $\tau<0$, we def\/ine
\begin{gather}\label{DefX+} X_{\alpha}(z;s)= \begin{cases} \Psi_{\alpha}^+(z;s) & \textrm{for $z \in {\rm I}, {\rm III}$,} \\
\Psi_{\alpha}^+(z;s) \begin{pmatrix} 1&1\\0&1 \end{pmatrix} & \textrm{for $z \in {\rm II}$,} \vspace{1mm}\\
\Psi_{\alpha}^+(z;s) \begin{pmatrix} 1&0\\ 1&1 \end{pmatrix} & \textrm{for $z \in {\rm IV}$,} \end{cases} \end{gather}
with regions I, II, III, IV as in Fig.~\ref{figure: contour Psi+}.
Then, it is easily verif\/ied that, for even $\alpha$, $X_\alpha$ satisf\/ies the RH problem for $X_\alpha$ stated in Section~\ref{sec: rec-}. It is important to note that this is not the case for~$\alpha$ odd. It follows that the formulas for $X_2$ and $q_2$ given in \eqref{formula X2} and \eqref{formulas q} still hold for $\tau>0$. By~\eqref{sigma q}, we obtain the formula \eqref{sigma2-intro} for $\sigma_2^+$.

\subsection[The kernel $\mathbb K_\alpha$ as $\tau\to +\infty$]{The kernel $\boldsymbol{\mathbb K_\alpha}$ as $\boldsymbol{\tau\to +\infty}$}
An asymptotic analysis as $s=-i\sqrt{\tau}\to -i\infty$ of the RH problem for $\Psi_{\rm CK}$, equivalent to the RH problem for $\Psi^+$, has been performed in \cite[Section 5]{CK}.
The following result can be extracted from that analysis:
\begin{gather} \label{s infty small norm +} {\Psi}_+^{+}(x;s) \begin{pmatrix} 1&1\\0&1\end{pmatrix} e^{\frac{sx}{2}\sigma_3} =I+\mathcal{O}\left(\frac{1}{s(|x|+1)}\right) \qquad \textrm{for $x\in (0,1)$,}
\end{gather}
as $s\to -i\infty$.
Substituting \eqref{s infty small norm +} and \eqref{def Phi +} in \eqref{limker}, we obtain that the Painlev\'e kernel $\mathbb K_\alpha^{{\rm PV}}$ converges to the sine kernel, proving~(\ref{lim sine}) whenever $|u-v|$ remains bounded in the limit where $\tau\to +\infty$:
\begin{gather*}
\mathbb K_{\alpha}^{{\rm PV}}(u,v;\tau)=\mathbb K^{\sin}(u,v)+\bigO \big(\tau^{-1/2}\big).
\end{gather*}

\subsection[The kernel $\mathbb K_\alpha^{\rm PV}$ as $\tau\to 0$]{The kernel $\boldsymbol{\mathbb K_\alpha^{\rm PV}}$ as $\boldsymbol{\tau\to 0}$}

Recall that the parameter $s$ in the model problem for $\Psi^+$ is negative imaginary.
Def\/ine
\begin{gather} \label{DefnPsi2}
{\Psi}^{(2)}(z;s) =e^{\frac{s}{4}\sigma_3} \Psi^+\left(\frac{2z}{|s|}+\frac{1}{2};s\right) \chi\left(z\right),\\
\chi(z) =\label{Defnchi}\begin{cases}
I,&\mbox{for $-\frac{|s|}{4}<\Re z<\frac{|s|}{4}$,}\\
e^{\frac{\pi i\alpha}{2}\sigma_3},&\mbox{for $\Re z<-\frac{|s|}{4}$,}\\
e^{\frac{-\pi i\alpha}{2}\sigma_3},&\mbox{for $\Re z>\frac{|s|}{4}$.}\\
\end{cases}
\end{gather}

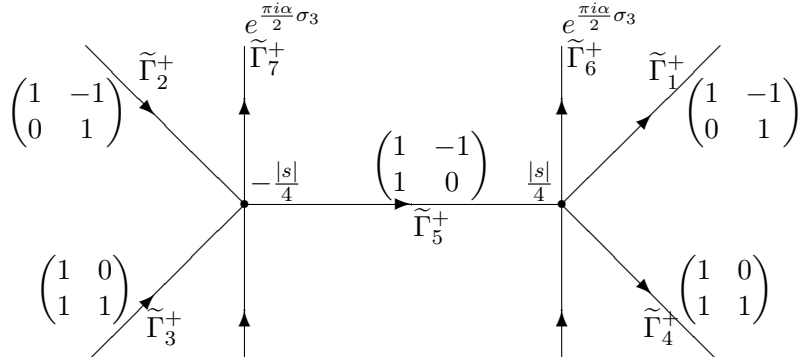
\begin{figure}[t]\centering
\begin{picture}(120,120)(-5,-50)

\put(-16,0){$-\frac{|s|}{4}$}
\put(88,0){$\frac{|s|}{4}$}
\put(135,42){$\widetilde{\Gamma}_1^+$}
\put(-58,42){$\widetilde{\Gamma}_2^+$}
\put(-55,-55){$\widetilde{\Gamma}_3^+$}
\put(133,-55){$\widetilde{\Gamma}_4^+$}
\put(46,-17){$\widetilde{\Gamma}_5^+$}
\put(104,47){$\widetilde{\Gamma}_6^+$}
\put(-16,47){$\widetilde{\Gamma}_7^+$}

\put(148,27){$\begin{pmatrix}1&-1\\0&1\end{pmatrix}$}
\put(-108,27){$\begin{pmatrix}1&-1\\0&1\end{pmatrix}$}
\put(-97,-40){$\begin{pmatrix}1&0\\1&1\end{pmatrix}$}
\put(145,-40){$\begin{pmatrix}1&0\\1&1\end{pmatrix}$}
\put(30,7){$\begin{pmatrix} 1&-1\\1&0\end{pmatrix}$}
\put(102,60){$e^{\frac{\pi i \alpha}{2}\sigma_3} $}
\put(-18,60){$e^{\frac{\pi i \alpha}{2}\sigma_3} $}

\put(-18,-5){\line(1,0){120}}
\put(-78,-65){\line(1,1){60}}
\put(102,-5){\line(1,1){60}}
\put(-78,55){\line(1,-1){60}}
\put(102,-5){\line(1,-1){60}}
\put(102,-65){\line(0,1){120}}
\put(-18,-65){\line(0,1){120}}

\put(-18,-5){\thicklines\circle*{2.5}}
\put(102,-5){\thicklines\circle*{2.5}}

\put(45,-5){\thicklines\vector(1,0){.0001}}

\put(-52,-39){\thicklines\vector(1,1){.0001}}
\put(136,29){\thicklines\vector(1,1){.0001}}
\put(-52,29){\thicklines\vector(1,-1){.0001}}
\put(136,-39){\thicklines\vector(1,-1){.0001}}

\put(102,-45){\thicklines\vector(0,1){.0001}}
\put(102,35){\thicklines\vector(0,1){.0001}}
\put(-18,-45){\thicklines\vector(0,1){.0001}}
\put(-18,35){\thicklines\vector(0,1){.0001}}

\end{picture}
\caption{The contour $\widetilde{\Gamma}^+$ and the jump matrices for $\Psi^{(2)}$.}\label{ContourGammaHat}
\end{figure}

 Then ${\Psi}^{(2)}$ satisf\/ies the following RH problem:
\subsubsection*{RH problem for $\boldsymbol{\Psi^{(2)}}$}
\begin{itemize}\itemsep=0pt
\item[(a)] ${\Psi}^{(2)}\colon \mathbb{C}{\setminus} {\widetilde\Gamma^+} \to \mathbb{C}$ is analytic, where ${\widetilde\Gamma^+}$ is as in Fig.~\ref{ContourGammaHat},
\begin{alignat*}{3} &{\widetilde\Gamma^+}= \cup_{i=1}^8 {\widetilde\Gamma^+}_i, \qquad && {\widetilde\Gamma^+}_1=\{z\colon \arg(z-|s|/4)=\pi/4\},& \\
&
{\widetilde\Gamma}_2^+=\{z\colon \arg(z+|s|/4)=3\pi/4\} ,\qquad &&
{\widetilde\Gamma}_3^+=\{z\colon \arg(z+|s|/4)=5\pi/4\} ,& \\
&{\widetilde\Gamma^+}_4=\{z\colon \arg(z-|s|/4)=7\pi/4\},\qquad &&
{\widetilde\Gamma^+}_5=[-|s|/4,|s|/4],& \\
&{\widetilde\Gamma^+}_6=\{z\colon \Re z=|s|/4\} ,\qquad &&
{\widetilde\Gamma^+}_7=\{z\colon \Re z=-|s|/4\}.&
\end{alignat*}
\item[(b)] ${\Psi}^{(2)}$ has continuous boundary values on $\widetilde\Gamma^+{\setminus} \big\{{-}\frac{|s|}{4},\frac{|s|}{4}\big\}$ given by
\begin{alignat*}{3}
& {\Psi}^{(2)}_+(z)={\Psi}^{(2)}_-(z)\begin{pmatrix}1&-1\\0&1\end{pmatrix}\qquad &&\textrm{for $z\in {\widetilde\Gamma}_1^+\cup {\widetilde\Gamma}_2^+$,}&\\
&{\Psi}^{(2)}_+(z)={\Psi}^{(2)}_-(z)\begin{pmatrix}1&0\\1&1\end{pmatrix}\qquad &&\textrm{for $z\in {\widetilde\Gamma^+}_3\cup {\widetilde\Gamma}_4^+$,}&\\
&{\Psi}^{(2)}_+(z)={\Psi}^{(2)}_-(z)\begin{pmatrix} 1&-1\\1&0\end{pmatrix}\qquad &&\textrm{for $z\in {\widetilde\Gamma^+}_5$,}&\\
&{\Psi}^{(2)}_+(z)= {\Psi}^{(2)}_-(z) e^{\frac{\pi i \alpha}{2}\sigma_3} \qquad &&\textrm{for $z\in {\widetilde\Gamma^+}_6\cup {\widetilde\Gamma^+}_7$.}&
\end{alignat*}
\item[(c)]
${\Psi}^{(2)}(z)$ has the following asymptotic behavior as $z\to \infty$:
\begin{gather} \label{asymptoticsPsi2}
{\Psi}^{(2)}(z)=
\big(I+\mathcal{O}\big(z^{-1}\big)\big)e^{iz\sigma_3}\chi (z ).
\end{gather}
\item[(d)] $\Psi^{(2)}(z)$ inherits its local behavior at $-\frac{|s|}{4}$ and $\frac{|s|}{4}$ from $\Psi^+(z)$ through \eqref{Psi+d1}--\eqref{Hint} and~(\ref{DefnPsi2}).
\end{itemize}

When $s\to -i0$ it was shown in \cite[Section~6.2]{CK} that, for $z$ bounded away from~$0$,
\begin{gather}\label{intrepM2}
{\Psi}^{(2)}(z;s)
\chi^{-1}(z)M^{-1}(z)
=I+\epsilon(z,s),
\end{gather}
where $M$ is as in (\ref{DefnM}), $\epsilon(z,s)= \mathcal{O}(|s|^2)+\mathcal{O}(|s|^{2+4\alpha})$ for $2\alpha \notin \mathbb{Z}$, and $\epsilon(z,s)= \mathcal{O}(|s|\log |s|)$ for $2\alpha \in \mathbb{Z}$.

We substitute (\ref{intrepM2}) and \eqref{DefnPsi2} into \eqref{def Phi +} to f\/ind that as $\tau \to 0 $, $\pm x>0$,
\begin{gather*}
\begin{pmatrix}
\Phi_1\big(x;-s^2\big)\\ \Phi_2\big(x;-s^2\big)
\end{pmatrix}
=M(x)e^{\mp \frac{\pi i \alpha}{2}}\begin{pmatrix}1\\ -1 \end{pmatrix} +\epsilon(x,s).
\end{gather*}

By~\eqref{limker}, we obtain the Bessel kernel after similar calculations as for the case $\tau<0$, proving~(\ref{lim Bessel}).

\section{RH problem for orthogonal polynomials}\label{section: 4}

\subsection{RH problem for orthogonal polynomials and dif\/ferential identity}

The correlation kernel \eqref{correlationkernelCD} for the eigenvalues in the random matrix ensemble~\eqref{matrixmeasure} can be expressed in terms of the solution of a RH problem which characterizes the orthogonal polyno\-mials~$p_j^{(n)}$ with respect to the weight $w_n$ def\/ined in~\eqref{Def:eigs}.

Def\/ine $Y$ by
\begin{gather} \label{MatrixY}
Y(z;n)=\begin{pmatrix}
\displaystyle \frac{1}{\kappa^{(n)}_n}p^{(n)}_n(z)	&	\displaystyle \frac{1}{2\pi i \kappa_n^{(n)}} \int_{\mathbb{R}} \frac{p^{(n)}_n(x)}{x-z} w_n(x) dx \\
-2\pi i \kappa^{(n)}_{n-1}p^{(n)}_{n-1}(z)	&\displaystyle -\kappa^{(n)}_{n-1} \int_{\mathbb{R}} \frac{p^{(n)}_{n-1}(x)}{x-z}w_n(x)dx
\end{pmatrix},
\end{gather}
where $\kappa_j^{(n)}>0$ is the leading coef\/f\/icient of the normalized orthogonal polynomial~$p_j^{(n)}$. It is well-known~\cite{FIK} that~$Y$ can be characterized as the unique solution of the following RH problem.

\subsubsection*{RH problem for $\boldsymbol{Y}$}
\begin{itemize}\itemsep=0pt
\item[(a)] $Y$ is analytic in $\mathbb{C}{\setminus} \mathbb{R}$.
\item[(b)] $Y$ has the following jump relations on $\mathbb{R}$, except at the singularities $\pm \sqrt{t}$ in the case where $t>0$:
\begin{gather*}Y_+(x)=Y_-(x)\begin{pmatrix}1&w_n(x)\\0&1 \end{pmatrix}.
\end{gather*}
\item[(c)] As $z\to \infty$,
\begin{gather*}
Y(z)=\left(I+\mathcal{O} \left( z^{-1} \right)\right) \begin{pmatrix} z^n &0\\ 0& z^{-n} \end{pmatrix}.
\end{gather*}
\item[(d)] If $t>0$ and $\alpha<0$, $Y(z)$ has the behavior
\begin{gather*}Y(z)=\begin{pmatrix} \mathcal{O} (1)& \mathcal{O}\big(|z\mp \sqrt{t}|^{\alpha}\big) \\ \mathcal{O}(1) &\mathcal{O}\big(|z\mp \sqrt{t}|^{\alpha}\big) \end{pmatrix}, \end{gather*}
as $z\rightarrow \pm \sqrt{t}$. If $t>0$ and $\alpha\geq 0$, $Y$ is bounded near $\pm\sqrt{t}$.
\end{itemize}
We can express the eigenvalue correlation kernel $K_n$ and the logarithmic $t$-derivative of the partition function $\widehat Z_n$ exactly in terms of $Y$.

\begin{Prop}\quad
\begin{itemize}\itemsep=0pt
\item[(1)] Let $K_n$ be the correlation kernel defined in~\eqref{correlationkernelCD}. For $x,y \in \mathbb R$, we have
\begin{gather}\label{CorrelationkernelY} K_n(x,y)=\frac{\sqrt{w_n(x)w_n(y)}}{2\pi i}\frac{ \begin{pmatrix}0&1\end{pmatrix} Y_+^{-1}(y)Y_+(x) \begin{pmatrix}1\\0\end{pmatrix}}{x-y}.
\end{gather}
\item[(2)] Let $\widehat Z_n(t)=\widehat Z_n(t,\alpha,V)$ be the partition function defined in~\eqref{Defn PartitionFunction}. Write $z_0=\sqrt t $ when $t>0$ and $z_0=i\sqrt{-t}$ when $t<0$.
For $t<0$, $\alpha>-\frac{1}{2}$ and for $t>0$, $\alpha\geq 0$, we have the differential identity
\begin{gather}\label{diff id}
\frac{d}{dt} \log \widehat Z_n(t)=-\frac{\alpha }{2z_0}\left(\left( Y^{-1}\frac{dY}{dz}\right)_{22}(z_0)-\left(Y^{-1}\frac{dY}{dz}\right)_{22}(-z_0)\right). \end{gather}
\end{itemize}
\end{Prop}
\begin{proof}
The formula for the correlation kernel (\ref{CorrelationkernelY}) follows after a straightforward calculation from (\ref{correlationkernelCD}) and \eqref{MatrixY}.

To obtain (\ref{diff id}), we follow ideas from \cite{Krasovsky} and \cite{CK}. In the remaining part of this proof, we suppress the dependence on $n$ in our notations and we write $p_j$, $\kappa_j$, and $w$ instead of $p_j^{(n)}$, $\kappa_j^{(n)}$, and~$w_n$. We start with the well known formula
\begin{gather*}
\widehat Z_n(t)=\prod_{j=0}^{n-1}\kappa_j^{-2}.
\end{gather*}
Using the orthogonality of the polynomials and the Christof\/fel--Darboux formula, it follows that
\begin{align*}
\frac{d}{dt}\log \widehat Z_n(t) & =-2 \sum_{j=0}^{n-1}\frac{\partial_t \kappa_j}{\kappa_j}=-2\sum_{j=0}^{n-1} \int_{-\infty}^{\infty}p_j(x)\partial_t(p_j(x))w(x)dx\\
& = -\int_{-\infty}^{\infty}\partial_t \left( \sum_{j=0}^{n-1}p_j^2(x)\right)w(x) dx\\
& =-\int_{-\infty}^{\infty}\partial_t \left( \frac{\kappa_{n-1}}{\kappa_n}(\partial_x(p_n(x))p_{n-1}(x)-p_n(x)\partial_x(p_{n-1}(x)))\right)w(x)dx.
\end{align*}
By orthogonality, it is straightforward to check that \begin{gather}\label{dt1}
 \frac{d}{dt} \log \widehat Z_n(t)=-n\frac{\partial_t\kappa_{n-1}}{\kappa_{n-1}}+\frac{\kappa_{n-1}}{\kappa_{n}}(J_1-J_2),\\
 J_1=\int^{\infty}_{-\infty}\partial_{t}p_n(x)\partial_xp_{n-1}(x)w(x)dx,\qquad J_2=\int^{\infty}_{-\infty}\partial_{x}p_n(x)\partial_{t}p_{n-1}(x)w(x)dx.\label{J1J2}
\end{gather}
 We start by evaluating $J_1$:
\begin{gather} \label{formulachangedifferential}
J_1=-\int_{-\infty}^{\infty}p_n(x)\partial_t(\partial_x p_{n-1}(x)w(x))dx+\int_{-\infty}^{\infty}\partial_t(p_n(x)\partial_xp_{n-1}(x)w(x))dx.
\end{gather}
The rightmost term in this expression vanishes, as one can see by taking the derivative outside the integral and using orthogonality.
By the form of the weight function~$w$ given in~\eqref{Def:eigs}, we obtain
\begin{gather}\label{formulaJ1}
 J_1=\frac{\alpha }{2z_0}\big(I^{(1)}_+-I^{(1)}_-\big),\\
 I^{(1)}_{\pm}=\int_{-\infty}^{\infty}\frac{p_n(x)\partial_x p_{n-1}(x)w(x)}{x\mp z_0} dx.\label{integral I}
\end{gather}
Writing $\partial_x(p_{m}(x))|_{x=y}=p_{m,x}(y)$, we f\/ind that
\begin{align*}
I^{(1)}_{\pm}& =\int^{\infty}_{-\infty}\frac{p_n(x)(p_{n-1,x}(x)-p_{n-1,x}(\pm z_0))w(x)}{x\mp z_0} dx +p_{n-1,x}(\pm z_0)\int^{\infty}_{-\infty}\frac{p_n(x)w(x)}{x\mp z_0} dx \\
& =p_{n-1,x}(\pm z_0)\int^{\infty}_{-\infty}\frac{p_n(x)w(x)}{x\mp z_0} dx.
\end{align*}
We proceed in a similar fashion for $J_2$ and f\/ind that
\begin{gather}
J_2=-n\frac{\kappa_n\partial_t\kappa_{n-1}}{\kappa_{n-1}^2}+\frac{\alpha}{2z_0}p_{n,x} (z_0)\int^{\infty}_{-\infty}\frac{p_{n-1}(x)w(x)}{x-z_0} dx\nonumber\\
\hphantom{J_2=}{}
 -\frac{\alpha }{2z_0}p_{n,x}(-z_0)\int^{\infty}_{-\infty}\frac{p_{n-1}(x)w(x)}{x+z_0} dx.\label{formulaJ2}
\end{gather}
Substituting \eqref{formulaJ1}--\eqref{formulaJ2} into \eqref{dt1}, and using \eqref{MatrixY}, it is straightforward to obtain~\eqref{diff id}.
\end{proof}

\begin{Remark}
When $t>0$ and $\alpha<0$, the integrals in~\eqref{integral I} are not def\/ined. Using techniques similar to those in~\cite{Krasovsky}, one can regularize those integrals and obtain a dif\/ferential identity similar to~\eqref{diff id}. We refer the reader to~\cite{CK, Krasovsky} for more details.
\end{Remark}

\subsection[The $g$-function and normalization of the RH problem]{The $\boldsymbol{g}$-function and normalization of the RH problem}

We proceed with a transformation of the RH problem for $Y$ to one which is normalized to the identity at inf\/inity, and which has suitable jump conditions. This is a standard part of the Deift/Zhou steepest descent method~\cite{DeiftZhou} which we apply in order to get asymptotics for~$Y$ as $n\to\infty$ with $t$ small. Def\/ine
\begin{gather}
T(z)=e^{\frac{n\ell}{2}\sigma_3}Y(z)e^{-ng(z)\sigma_3}e^{\frac{-n\ell}{2}\sigma_3},\label{def T}
\end{gather}
where $g$ is def\/ined by
\begin{gather}\label{def g} g(z)=\int \log(z-x) d\mu_V(x) ,\end{gather}
with $\mu_V$ the equilibrium measure given by \eqref{eq density}, minimizing (\ref{energy}) and satisfying the condi\-tions~(\ref{EulLag1}),~(\ref{EulLag2}).
 Clearly
 $g(z)= \log z + \mathcal{O}\left(z^{-1}\right)$ as $z\to \infty$, and it follows that $T$ satisf\/ies the following RH problem.

 \subsubsection*{RH problem for $\boldsymbol{T}$}
\begin{itemize}\itemsep=0pt
\item[(a)]
$T$ is analytic in $\mathbb{C}{\setminus} \mathbb{R}$.
\item[(b)]
$T$ has the following jump relation for $x \in \mathbb{R}$ (except for $x=\pm\sqrt{t}$ if $t>0$),
\begin{gather}\label{jump T}
T_+(x)=T_-(x)
\begin{pmatrix}
e^{-n(g_+(x)-g_-(x))} &|x^2-t|^{\alpha} e^{n(g_+(x)+g_-(x)-V(x)+\ell)} \\
0 & e^{n(g_+(x)-g_-(x))}
\end{pmatrix}.
\end{gather}
\item[(c)]
$T(z)=I+\mathcal{O}(z^{-1}) \quad \text{as $z\rightarrow \infty$}.$
\item[(d)] If $t>0$, $\alpha<0$, as $z\to \pm\sqrt{t}$, we have
\begin{gather}\nonumber
T(z)=
\begin{pmatrix} \mathcal{O} (1)& \mathcal{O}\big(|z\mp \sqrt{t}|^{\alpha}\big) \\ \mathcal{O}(1) &\mathcal{O}\big(|z\mp \sqrt{t}|^{\alpha}\big) \end{pmatrix}.
\end{gather}
If $t>0$ and $\alpha\geq 0$, $T$ is bounded near $\pm\sqrt{t}$.
\end{itemize}
By the fact that $V$ is one-cut regular, we have~(\ref{eq density}) with $h$ real analytic and positive on $[a,b]$. By~(\ref{def g}) and the Euler--Lagrange conditions (\ref{EulLag1}),~(\ref{EulLag2}), it can be shown that $g$ satisf\/ies the following properties which we will need later on:
\begin{itemize}\itemsep=0pt
\item[(a)] $g_+(x)-g_-(x)=\begin{cases} 0& \textrm{for $x>b$,} \\
2\pi i& \textrm{for $x<a$,} \\
2\pi i\int_x^b h(\xi)((\xi-a)(b-\xi))^{\frac{1}{2}}d\xi& \textrm{for $x\in [a,b]$,}
 \end{cases}$
\item[(b)] $g_+(x)+g_-(x)-V(x)+\ell<0$ for $x>b$ and $x<a$,
\item[(c)] $g_+(x)+g_-(x)-V(x)+\ell= \begin{cases}
2\pi \int_x^a h(\xi)\left((\xi-a)(\xi-b)\right)^{1/2}d\xi & \textrm{for $x<a$,}\\
0 &\textrm{for $x\in [a,b]$,}\\
-2\pi \int_b^x h(\xi)\left((\xi-a)(\xi-b)\right)^{1/2}d\xi & \textrm{for $x>b$.}
\end{cases}$
\end{itemize}
In part (a), the square root is positive for $s\in (a,b)$, and in part (c), the square roots are positive for $s>b$ and negative for $s<a$, with branch cut on $[a,b]$. The function $h$ is analytic on a~neighborhood $\mathcal U$ of $\mathbb{R}$, and we def\/ine the function $\phi\colon \mathcal U {\setminus} (-\infty,b]\to \mathbb{C}$ as
\begin{gather}
\phi(z)=\pi\int_b^z h(\xi)\left((\xi-a)(\xi-b)\right)^{1/2}d\xi,\label{phi}
\end{gather}
with the square root analytic on $\mathbb C{\setminus}[a,b]$ and positive for $\xi>b$, and with the path of integration not crossing $(-\infty,b]$.
By property~(a) of the function~$g$, it follows that
\begin{gather*}
\mp(g_+(x)-g_-(x))=2\phi_{\pm}(x) \qquad \textrm{for $x\in(a,b)$.}
\end{gather*}
Combining this with property (c) for the function $g$ we obtain
\begin{gather} \nonumber
T_+(x)=T_-(x)
\begin{pmatrix}
e^{2n\phi _+(x)} &|x^2-t|^{\alpha} \\
0 & e^{2n\phi _-(x)}
\end{pmatrix}\qquad \textrm{for $x\in[a,b]$.}
\end{gather}
By property (c) of the function $g$ it follows that
\begin{gather}\label{Consequencepropc}
g_+(x)+g_-(x)-V(x)+\ell=\begin{cases}
-2\phi(x)& \textrm{for $x>b$,}\\
-2\phi_+(x) -2\pi i& \textrm{for $x<a$.}
 \end{cases}
\end{gather}
Combined with property (a) of $g$ it follows that
\begin{gather*}
T_+(x)= T_-(x)
\begin{pmatrix}
1 &|x^2-t|^{\alpha} e^{-2n{\phi}(x)} \\
0 &1
\end{pmatrix}
\qquad \textrm{for $ x<a$ and for $b<x$},
\end{gather*}
where we note that $e^{-2n{\phi_+}(x)}=e^{-2n{\phi_-}(x)}$ by~(\ref{Consequencepropc}).

\section[Asymptotic analysis for $T$ as $n\to \infty$ for $t<0$]{Asymptotic analysis for $\boldsymbol{T}$ as $\boldsymbol{n\to \infty}$ for $\boldsymbol{t<0}$}\label{section: RH-}
In what follows, we will perform further transformations on $T$ in order to obtain a RH problem for which we can compute large~$n$ asymptotics. We f\/irst complete the analysis in the case~$t<0$.

\subsection{Opening of the lens}

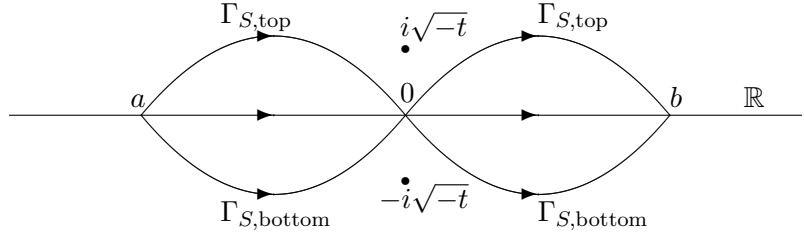
\begin{figure}[t]\centering
\begin{picture}(100,80)(-5,-40)
 \put(40,0){$0$}
\put(-62,-2){$a$}
\put(142,-2){$b$}
\put(40,25){$i\sqrt{-t}$}
\put(32,-40){$-i\sqrt{-t}$}
\put(-108,-5){\line(1,0){300}}

 \put(92,25){\thicklines\vector(1,0){.0001}}
 \put(92,-5){\thicklines\vector(1,0){.0001}}
 \put(92,-35){\thicklines\vector(1,0){.0001}}
\put(-8,25){\thicklines\vector(1,0){.0001}}
\put(-8,-5){\thicklines\vector(1,0){.0001}}
\put(-8,-35){\thicklines\vector(1,0){.0001}}

\put(170,-2){$\mathbb{R}$}
\put(-28,30){$\Gamma_{S,\textrm{top}}$}
\put(92,30){$\Gamma_{S,\textrm{top}}$}
\put(-28,-45){$\Gamma_{S,\textrm{bottom}}$}
\put(92,-45){$\Gamma_{S,\textrm{bottom}}$}

 \qbezier(-58,-5)(-8,55)(42,-5)
 \qbezier(-58,-5)(-8,-65)(42,-5)
 \qbezier(42,-5)(92,55)(142,-5)
 \qbezier(42,-5)(92,-65)(142,-5)

\put(42,20){\thicklines\circle*{2.5}}
\put(42,-30){\thicklines\circle*{2.5}}
\end{picture}
\caption{The contour $\Gamma_S$.}\label{ContourS}
\end{figure}

We deform the jump contour for $T$, which is the real line, to a lens-shaped contour around $[a,b]$ as in Fig.~\ref{ContourS}.
It is important that the singularities $\pm i\sqrt{-t}$ lie in the region outside of the lens, and that the lenses lie within the region~$\mathcal U$ of analyticity of~$h$. Since we want to obtain asymptotics which are valid for~$t$ arbitrary small, we choose the lenses to pass through~$0$, and with shape independent of~$t$ and~$n$.

Def\/ine
\begin{gather}
S(z)=\begin{cases}T(z)& \textrm{for $z$ outside the lens,}\\
T(z)\begin{pmatrix}1&0\\
-\frac{e^{2n\phi(z)}}{(z^2-t)^{\alpha}} &1
\end{pmatrix}&
\textrm{for $z$ in the upper parts of the lens,}\\
T(z)\begin{pmatrix}
1&0\\
\frac{e^{2n\phi(z)}}{(z^2-t)^{\alpha}} &1
\end{pmatrix} &\textrm{for $z$ in the lower parts of the lens.}
\end{cases}\label{def S}
\end{gather}
The function $(z^2-t)^\alpha$ is chosen with branch cuts on $(-i\infty,-i\sqrt{-t}]\cup [i\sqrt{-t},+i\infty)$, outside the lens-shaped region, and positive for real $z$. It is such that $(z^2-t)^\alpha=|z^2-t|^\alpha$ for~$z\in(a,b)$.

It follows that $S$ solves the following RH problem.

\subsubsection*{RH problem for $\boldsymbol{S}$}
\begin{itemize}\itemsep=0pt
\item[(a)]
$S$ is analytic for $z\in \mathbb{C} {\setminus} \Gamma_S$ where $\Gamma_S=\Gamma_{S,\textrm{top}}\cup \Gamma_{S,\textrm{bottom}} \cup \mathbb{R}$ is the contour shown in Fig.~\ref{ContourS}.
\item[(b)] $S$ has the following jump relations on $\Gamma_S$:
\begin{alignat*}{3}
& S_+(x)=
S_-(x)\begin{pmatrix}
1 &|x^2-t|^{\alpha} e^{-2n{\phi}(x)} \\
0 &1
\end{pmatrix},
\qquad &&\mbox{ for } x\in\mathbb R{\setminus}[a,b],& \\
&S_+(x)= S_-(x) \begin{pmatrix} 0& |x^2-t|^{\alpha}\\-\frac{1}{|x^2-t|^{\alpha}} &0\end{pmatrix}, \qquad &&\mbox{for $x\in (a,b)$,}& \\
&S_+(z)= S_-(z) \begin{pmatrix} 1 & 0 \\ \frac{e^{2n\phi(z)}}{(z^2-t)^{\alpha}} &1 \end{pmatrix}, \qquad &&\mbox{for $z\in \Gamma_{S,\textrm{top}}\cup \Gamma_{S,\textrm{bottom}}$.} &
\end{alignat*}
\item[(c)]
$S(z)=I+\mathcal{O} \left( z^{-1} \right)$ as $z\rightarrow \infty$.
\item[(d)]
$S(z)$ is bounded as $z\to 0$, as $z\to a$, and as $z\to b$.
\end{itemize}

\subsection[Global parametrix for $t<0$]{Global parametrix for $\boldsymbol{t<0}$}
The jump matrices for $S$ tend to the identity matrix exponentially fast everywhere on $\Gamma_S$, with the exception of the interval $(a,b)$ and small neighborhoods of $a$, $b$, and $0$. We construct a~global parametrix~$N$ which solves the RH problem for~$S$, ignoring exponentially small jumps and small neighborhoods of $a$, $b$, and $0$. We will show later on that the global parametrix is a~good approximation of~$S$ away from~$a$,~$b$,~$0$.

\subsubsection*{RH problem for $\boldsymbol{N}$}
\begin{itemize}\itemsep=0pt
\item[(a)] $N\colon \mathbb{C}{\setminus} [a,b] \to \mathbb{C}^{2\times 2}$ is analytic.
\item[(b)] $N$ has the following jump relation on $(a,b)$, where the orientation is taken from left to right:
\begin{gather*}N_+(x) =N_-(x)\begin{pmatrix} 0& |x^2-t|^{\alpha}\\-\dfrac{1}{|x^2-t|^{\alpha}} &0\end{pmatrix} \qquad\textrm{for $x\in (a,b)$.}\end{gather*}
\item[(c)]$N(z)=I+\mathcal{O} \left( \frac{1}{z} \right)\textrm{ as $z\rightarrow \infty$}$.
\end{itemize}

Def\/ine \begin{gather}
N(z)=D(\infty)^{-1}N_0(z)D(z),\label{def N}
\end{gather} where
\begin{gather}\label{def N2}
N_0(z)=\begin{pmatrix} 1&i\\i &1\end{pmatrix} \begin{pmatrix} \gamma(z) &0 \\ 0 &\gamma(z)^{-1} \end{pmatrix} \begin{pmatrix} 1&i\\i&1\end{pmatrix} ^{-1}, \qquad \gamma(z)=\left( \frac{z-b}{z-a}\right)^{1/4},
\end{gather}
 with\looseness=-1 $\left( \frac{z-b}{z-a}\right)^{1/4}$ analytic on $\mathbb C{\setminus}[a,b]$ and positive for $z>b$. $D(z)$ depends on $\alpha$ and $t$ and is given by
\begin{gather}
\label{N3}D(z) = d(z)^{-\sigma_3},\\
\label{def N4}d(z) = \exp \left(-\alpha \frac{\left((z-a)(z-b)\right)^{1/2}}{2\pi}\left( \int_a^b \frac{\log |\xi^2 -t| d\xi}{(\xi-z)\sqrt{(\xi-a)(b-\xi)}} \right)\right),
\end{gather}
with $\sqrt{(z-a)(z-b)}$ analytic in $\mathbb C{\setminus}[a,b]$ and positive for $z>b$, and with $\sqrt{(\xi-a)(b-\xi)}$ positive for $\xi\in(a,b)$.
Then, it is straightforward to verify that $N$ solves the RH problem for~$N$.
Moreover, as $z\to a$ and $z\to b$, $N(z)$ has the following behavior:
\begin{alignat}{3} \label{Na} & N(z)=\mathcal{O}\big(|z-a|^{-1/4}\big) \qquad && \textrm{as $ z \rightarrow a$,}& \\
 \label{Nb} & N(z)=\mathcal{O}\big(|z-b|^{-1/4}\big) \qquad && \textrm{as $ z \rightarrow b$.}&
 \end{alignat}
To verify this, one f\/irst shows that $N_0(z)$ satisf\/ies the conditions \eqref{Na},~\eqref{Nb}. It then remains to show that~$d$ is bounded near~$a$ and~$b$. This can be seen by deforming the integral between~$a$ and~$b$ to a contour surrounding the interval~$[a,b]$ and applying a residue argument.

\subsection[Local parametrix at $a$ and $b$]{Local parametrix at $\boldsymbol{a}$ and $\boldsymbol{b}$} \label{locparamaandb}
We need to construct local parametrices near the endpoints~$a$ and~$b$. This construction, using the Airy function, is standard~\cite{Deift, Deiftetal, Deiftetal2}.
The explicit formula for the local parametrices near~$a$ and~$b$ is not relevant to later calculations, we only need that these local parametrices exist.

Let $U_b$ ($U_a$) be an open set which is independent of $n$ and which contains $b$ ($a$). The following RH problems can be solved in terms of the Airy function. We do not give details about this construction here.
\subsubsection*{RH problem for $\boldsymbol{P^{(b)}}$}
\begin{itemize}\itemsep=0pt
\item[(a)] $P^{(b)}\colon U_b{\setminus} \Gamma_S \to \mathbb{C}^{2\times 2}$ is analytic on $U_b{\setminus}\Gamma_S$.
\item[(b)] $P^{(b)}$ has exactly the same jump relations as $S$ has for $z \in \Gamma_S \cap U_b$.
\item[(c)] $P^{(b)}(z)=(I+\mathcal{O}(n^{-1}))N(z)$ as $n\to \infty$, uniformly for $z\in \partial U_b$ and $|t|$ suf\/f\/iciently small.
\end{itemize}

\subsubsection*{RH problem for $\boldsymbol{P^{(a)}}$}
\begin{itemize}\itemsep=0pt
\item[(a)] $P^{(a)}\colon U_a{\setminus} \Gamma_S \to \mathbb{C}^{2\times 2}$ is analytic on $U_a{\setminus} \Gamma_S$.
\item[(b)] $P^{(a)}$ has exactly the same jump relations as $S$ has for $z \in \Gamma_S \cap U_a$.
\item[(c)] $P^{(a)}(z)=(I+\mathcal{O}(n^{-1}))N(z)$ as $n\to \infty$, uniformly for $z\in \partial U_a$ and $|t|$ suf\/f\/iciently small.
\end{itemize}

\subsection{Local parametrix at the origin}
Let $U_0$ be a suf\/f\/iciently small disk around the origin, which is independent of $n$. We want to construct a local parametrix which has the same jumps as $S$ in the disk $U_0$ and which matches with the global parametrix up to order $1/n$ at the boundary of $U_0$. This matching condition has to be valid in the double scaling limit where $n\to\infty$ and simultaneously $t\to 0$ in such a way that $n^2t$ tends to a non-zero constant, or equivalently such that $\tau_{n,t}\to \tau<0$, with $\tau_{n,t}$ def\/ined in~\eqref{taunt}.

\subsubsection*{RH problem for $\boldsymbol{P}$}
\begin{itemize}\itemsep=0pt
\item[(a)] $P\colon U_0{\setminus} \Gamma_S\to \mathbb{C}^{2\times 2}$ is analytic.
\item[(b)] $P$ has the same jump relations as $S$:
\begin{gather*}
 P_-^{-1}(z)P_+(z)=S_-^{-1}(z)S_+(z) \qquad \text{for $z \in U_0 \cap \Gamma_S$.}
 \end{gather*}
\item[(c)]$P(z)=\left(I+\mathcal{O}\left(n^{-1}\right)\right)N(z)$ uniformly for~$z$ on the boundary~$\partial U_0$, as $n\to\infty$ and simultaneously~$t\to 0$ in such a way that~$n^2t$ tends to a non-zero constant.
\end{itemize}

We will construct a solution to the RH problem for $P$ in terms of the model RH solution $\Psi^-$ studied in Section~\ref{section: PV}.

\subsubsection{Construction of the parametrix}

Def\/ine $\lambda\colon U_0\to \mathbb{C}$ by
\begin{gather}\label{def lambda}
 \lambda(z)=-in\times \begin{cases}
 \phi(z) - \dfrac{\phi(i\sqrt{-t})-\phi(-i\sqrt{-t})}{2} &\textrm{for $\Im z>0$},\vspace{1mm}\\
 -\phi(z)-\dfrac{\phi(i\sqrt{-t})-\phi(-i\sqrt{-t})}{2} & \textrm{for $\Im z<0$}.
 \end{cases}
\end{gather}
By (\ref{phi}), $\lambda$ is a conformal map in $U_0$ (if this one is chosen suf\/f\/iciently small, but containing $\pm i \sqrt{-t}$), and we have \begin{gather*}\lambda\big(i\sqrt{-t}\big)=-\lambda\big({-}i\sqrt{-t}\big)=-\frac{in}{2}\big(\phi\big(i\sqrt{-t}\big)+\phi\big({-}i\sqrt{-t}\big)\big),
\end{gather*}
and
\begin{gather}\nonumber
\lambda'(0)=n\pi\psi_V(0).
\end{gather}
We def\/ine $s_{n,t}$ such that $\lambda$ sends $i\sqrt{-t}$ to $is_{n,t}/4$ and $-i\sqrt{-t}$ to $-is_{n,t}/4$:
\begin{gather}\label{def s}
s_{n,t}=-4i \lambda\big(i\sqrt{-t}\big)=-2n\big(\phi\big(i\sqrt{-t}\big)+\phi\big({-}i\sqrt{-t}\big)\big)=-4n\Re\phi\big(i\sqrt{-t}\big)>0.
\end{gather}
By (\ref{phi}), we have
\begin{gather}
s_{n,t}=-4\pi n\Re\left(\int_0^{i\sqrt{-t}}h(s)\left((s-a)(s-b)\right)^{1/2}ds\right)\nonumber\\
\hphantom{s_{n,t}}{}
=4n\pi\sqrt{-t}\psi_V(0)+\bigO\big(n|t|^{3/2}\big)=\sqrt{-\tau_{n,t}}+\bigO\big(n|t|^{3/2}\big),\label{s2}
\end{gather}
as $t\to 0$, $n\to\infty$, with $\tau_{n,t}$ as in \eqref{taunt}.

There remains some freedom to choose the contour~$\Gamma_S$, which has to be of the shape shown in Fig.~\ref{ContourS}, but is otherwise arbitrary.
We choose $\Gamma_S$ such that the upper and lower lenses are mapped to the straight lines $\widetilde\Gamma_1^-,\ldots, \widetilde\Gamma_4^-$ (see Fig.~\ref{figure: contour Psi(1)}) by~$\lambda$.

Recall the def\/inition of $\Psi^{(1)}$ from~\eqref{def Psi1}.
We take $P$ of the form
\begin{gather} \label{defnP}
P(z)=E(z)\Psi^{(1)}(\lambda(z);s_{n,t}) W(z).
\end{gather}
Here, $E$ has to be analytic in $U_0$, and $W$ takes the form
\begin{gather}\label{def W}
W(z)=\begin{cases}
\left((z^2-t)^{\alpha/2}e^{-n\phi(z)}\right)^{\sigma_3} \sigma_3\sigma_1 &\textrm{for $\Im z>0$,}\\
\left( (z^2-t)^{-\alpha/2}e^{n\phi(z)}\right)^{\sigma_3}&\textrm{for $\Im z<0$.}
\end{cases}
\end{gather}
It is constructed in such a way that $P$ has the required jump conditions on $\Gamma_S\cap U_0$.
The branch cuts of $(z^2-t)^{\alpha/2}$ are taken on the lines $(-i\infty, -i\sqrt{-t})$ and $(i\sqrt{-t}, i \infty)$. $E$ is an analytic matrix-valued function on $U_0$ which is constructed such that $P(z)$ approximates $N(z)$ on the boundary~$\partial U_0$:
\begin{gather}\label{DefnE}
E(z)=N(z)W^{-1}(z)\left(\frac{\lambda(z)-\frac{is_{n,t}}{4}}{\lambda(z)+\frac{is_{n,t}}{4}}\right)^{\frac{\alpha}{2}\sigma_3}e^{-i\lambda(z) \sigma_3}e^{\frac{\pi i \alpha}{2}\sigma_3},
\end{gather}
where we have chosen the branch cuts for $z$ in $(i\sqrt{-t}, i \infty)$ and $(-i\sqrt{-t},-i\infty)$ such that
\begin{gather}
\left(\frac{\lambda(z)-\frac{is_{n,t}}{4}}{\lambda(z)+\frac{is_{n,t}}{4}}\right)^{\frac{\alpha}{2}\sigma_3}=1 +\bigO \big(n^{-1}\big) \qquad \textrm{for $\Re(z)>0$ and $z \in \partial U$,}\label{branchfrac}
\end{gather}
as $n \to \infty$ and $s_{n,t}$ remains bounded.
Note that $N(z)W^{-1}(z)$ is continuous across the real line and that $ W^{-1}\left(\frac{\lambda(z)-\frac{is_{n,t}}{4}}{\lambda(z)+\frac{is_{n,t}}{4}}\right)^{\frac{\alpha}{2}\sigma_3}$ has no singularities at $\pm i\sqrt{-t}$. No other component of $E$ has discontinuities. It follows that $E$ is analytic on $U_0$.

\begin{Prop} \label{P0theorem-}
Let $P(z)$ be defined as in \eqref{defnP}. Then, $P$ satisfies the RH problem for~$P$.
\end{Prop}

\begin{proof} Conditions (a) and (b) are satisf\/ied by construction. We verify~(c). For $z \in \partial U_0$ we have
\begin{gather} \label{boundaryeqn}
P(z)N^{-1}(z)= E(z)\Psi^{(1)}({\lambda(z)};s_{n,t})W(z) N^{-1}(z).
\end{gather}
It follows from (\ref{Psi1 as}) that
\begin{gather}
\Psi^{(1)} (\lambda(z);s_{n,t}) \left(\frac{\lambda(z)-\frac{is_{n,t}}{4}}{\lambda(z)+\frac{is_{n,t}}{4}}\right)^{\frac{\alpha}{2}\sigma_3} e^{-i\lambda(z)\sigma_3} e^{\frac{\pi i\alpha}{2} \sigma_3} = I+\mathcal{O}\big(n^{-1}\big),\label{matching2}
\end{gather}
as $n\to\infty$,
where $\mathcal{O} (n^{-1})$ is uniform for $s_{n,t}$ in compact subsets of $(0,\infty)$ and $z\in \partial U_0$ and branches are chosen as in~(\ref{branchfrac}). By~\eqref{def W}, we f\/ind that
\begin{gather} \label{boundaryeqnerror}
P(z)N^{-1}(z)= E(z) \big(I+\mathcal{O}\big(n^{-1}\big)\big)E^{-1}(z).
\end{gather}
The proposition follows if $E(z)$ is bounded uniformly on $\partial U_0$. By (\ref{DefnE}) it is clear that the only part of $E$ which could diverge is $W^{-1}(z)e^{-i\lambda(z)\sigma_3}$, which contains $e^{n\phi(z)-i\lambda(z)}$. To deal with this we note that
\begin{gather}
\label{boundedE}
e^{n\phi(z)-i\lambda(z)}=e^{-\frac{n}{2}(\phi(i\sqrt{-t})-\phi(-i\sqrt{-t}))},
\end{gather}
which has modulus $1$ because of the symmetry $\overline{\phi(\bar z)}=\phi(z)$. Substituting this into (\ref{DefnE}), one indeed f\/inds that $E(z)=\mathcal{O}(1)$ uniformly on the boundary $\partial U_0$.
\end{proof}

\subsection{Small norm RH problem} \label{smallnorm}
Def\/ine $R(z)$ by
\begin{gather}\label{def R}
R(z)=\begin{cases}
S(z)P(z)^{-1} &\textrm{for $z\in U_0$,}\\
S(z)P^{(a)}(z)^{-1} &\textrm{for $z\in U_a$,}\\
S(z)P^{(b)}(z)^{-1} &\textrm{for $z\in U_b$,}\\
S(z)N(z)^{-1} &\textrm{for $z\in \mathbb{C} {\setminus} (U_0\cup U_a\cup U_b)$.}
\end{cases}
\end{gather}
It follows from the conditions of the RH problems for $P$, $P^{(a)}$, $P^{(b)}$ and $N$ that $R(z)$ satisf\/ies
\subsubsection*{RH problem for $\boldsymbol{R}$}
\begin{itemize}\itemsep=0pt
\item[(a)] $R$ is analytic on $\mathbb{C}{\setminus} \Gamma_R$ where \begin{gather*} \Gamma_R= \partial U_a\cup \partial U_b \cup \partial U_0 \cup (\Gamma_S{\setminus} ((a,b)\cup U_0\cup U_a \cup U_b))).\end{gather*}
\item[(b)] $R$ has the following jump relations on $\Gamma_R$:
\begin{gather*} R_+(z)=R_-(z)\big(I+\mathcal{O}\big(n^{-1}\big)\big)\end{gather*}
as $n\to\infty$, uniformly for $z\in\Gamma_R$, in the double scaling limit where $\tau_{n,t}\to\tau<0$.
\item[(c)] $R(z)$ has the following asymptotics as $z\to \infty$:
\begin{gather*}
R(z)= I+\mathcal{O}\big(z^{-1}\big).
\end{gather*}
\end{itemize}

We have that $R$ solves a RH problem, normalized at inf\/inity and with jump matrices that are uniformly close to the identity matrix. This type of RH problem is often called a small norm RH problem, and it implies that $R$ itself is uniformly close to $I$:
 in the double scaling limit where $n\to\infty$ and $t\to 0$ in such a way that $\tau_{n,t}\to \tau<0$, we have
\begin{gather}\label{asymp R}
R(z)=I+\mathcal{O}\big(n^{-1}\big),
\end{gather}
where the error term $\mathcal{O}(n^{-1})$ is uniform for $z\in\mathbb C{\setminus}\Gamma_R$.

\begin{Remark}
\label{remark: uniform R}
In fact, it can be shown that the estimate~\eqref{asymp R} holds not only if $n\to\infty$, $t\to 0$ in such a way that $\tau_{n,t}\to \tau<0$, but also if $\tau_{n,t}\to 0$ and if $\tau_{n,t}\to -\infty$. To see this, we f\/irst recall from
Proposition~3.1 in~\cite{CIK} that~(\ref{matching2}) holds uniformly for $s_{n,t}\in(0,\infty)$ (or $\tau_{n,t} <0$) as $n\to \infty$, and it follows that the same holds for~(\ref{boundaryeqnerror}). Since $E$ is uniformly bounded, this implies that condition (b) of the RH problem for $R$ is uniform for $\tau_{n,t}<0$ as $n\to \infty$, and thus~\eqref{asymp R} is too.
\end{Remark}

\subsection[Proof of Theorem \ref{thm lim kernel} for $t<0$]{Proof of Theorem \ref{thm lim kernel} for $\boldsymbol{t<0}$}\label{section: proof double}

If $\tau_{n,t}\to \tau<0$, we have by \eqref{s2} that $s_{n,t}\to s\in(0,+\infty)$. We have that $s_{n,t}=4n\pi \psi_V(0)\sqrt{-t}+\mathcal{O}(n^{-1})$ as $n\to \infty$.

First, recall formula~(\ref{CorrelationkernelY}) which expresses the correlation kernel in terms of~$Y$. Expres\-sing~$Y$ in $T$ by~(\ref{def T}), we obtain
\begin{gather*}
K_n(x,y)=|x^2-t|^{\alpha/2}e^{n\left(g_+(x)+\frac{\ell-V(x)}{2}\right)}|y^2-t|^{\alpha/2}e^{n\left(g_+(y)+\frac{\ell-V(y)}{2}\right)}\frac{ \begin{pmatrix}0&1\end{pmatrix} T_+^{-1}(y)T_+(x) \begin{pmatrix}1\\0\end{pmatrix}}{2\pi i(x-y)}.
\end{gather*}
Using the identity $g_+-\frac{1}{2}V+\ell/2=-\phi_+$ on $(a,b)$ and substituting $S$ into the equation, by (\ref{def S}), we obtain
\begin{gather}\label{PropCorrKern}
K_n(x,y)=\frac{1}{2\pi i(x-y)} \begin{pmatrix} -\frac{e^{n\phi_+(y)}}{|y^2-t|^{\alpha/2}}& \frac{|y^2-t|^{\alpha/2}}{e^{n\phi_+(y)}} \end{pmatrix}S_+^{-1}(y)S_+(x) \begin{pmatrix}e^{-n\phi_+(x)}|x^2 -t|^{\alpha/2}\\ e^{n\phi_+(x)} |x^2-t|^{-\alpha/2} \end{pmatrix}.
\end{gather}

Now, we use the fact that $S=RP$ for $z$ close to the origin. Substituting this and the def\/inition (\ref{defnP}) of $P$, we obtain for $x,y \in U_0$:
\begin{gather*}
K_n(x,y)=\frac{1}{2\pi i (x-y)} \begin{pmatrix} 1&1 \end{pmatrix} \left(\Psi^{(1)}_+\right)^{-1}(\lambda(y);s_{n,t})E_+^{-1}(y)\\
\hphantom{K_n(x,y)=}{} \times R_+^{-1}(y)R_+(x) E_+(x)\Psi_+^{(1)}(\lambda(x);s_{n,t})\begin{pmatrix}1\\-1\end{pmatrix} .
\end{gather*}
Up to this point, no approximations have been made. Now, we use the fact that $R=I+\bigO(n^{-1})$ in the double scaling limit and the analyticity of~$E$ and~$R$, which imply that
\begin{gather*}
E_+^{-1}(y)R_+^{-1}(y)R_+(x) E_+(x)=I+\mathcal{O}(x-y),
\end{gather*}
as $x\to y$ for $n$ suf\/f\/iciently large.
We f\/ind, for f\/ixed $u,v$, that
\begin{gather*} \frac{1}{\psi_V(0) n}K_{n,t}\left(x_n=\frac{u}{ n\psi_V(0)},y_n =\frac{v}{n\psi_V(0)}\right)\\
\qquad{} =\frac{1}{2\pi i(u-v)}\begin{pmatrix}1&1\end{pmatrix} \left(\Psi^{(1)}_+\right)^{-1}(\lambda(y_n);s_{n,t})\Psi_+^{(1)}( \lambda (x_n);s_{n,t})\begin{pmatrix}1\\-1\end{pmatrix}+ \mathcal{O}\big(n^{-1}\big) ,
 \end{gather*}
 as $n\to\infty$, $t\to 0$.
By (\ref{def Phi -}) and (\ref{def Psi1}) we have
\begin{gather} \label{defnPhi}
\begin{pmatrix}\Phi_1(\lambda;\tau)\\ \Phi_2(\lambda;\tau)
\end{pmatrix}=e^{-\frac{s}{4}\sigma_3}\Psi^{(1)}\big(\lambda;s=\sqrt{-\tau}\big)\begin{pmatrix} -1\\1\end{pmatrix},
\end{gather}
and it follows from \eqref{s2} that
\begin{gather}\label{final lim kernel}
\frac{1}{n\psi_V(0) }K_{n,t}\left(x_n=\frac{u}{ n\psi_V(0)},y_n= \frac{v}{n\psi_V(0)}\right)= \mathbb K_{\alpha}^{{\rm PV}}\big(\lambda(x_n),\lambda(y_n);-s_{n,t}^2\big)+ \mathcal{O}\big(n^{-1}\big),
\end{gather}
as $n\to\infty$, $t\to 0$; with $\mathbb K_\alpha^{\rm PV}$ given in~\eqref{limker}.
In the case where $\tau_{n,t}\to \tau<0$,
\begin{gather} \lambda (x_n)=\pi u+\mathcal O \big(n^{-1}\big), \qquad -s_{n,t}^2=\tau_{n,t}+\mathcal O\big(n^{-1}\big) \label{errors} \end{gather}
 which proves \eqref{limit kernel} by continuity of the kernel.
If $\tau_{n,t}\to -\infty$, the error terms in~(\ref{errors}) become of order $\mathcal O(n|t|^{3/2})$. Nevertheless we can apply~\eqref{lim sine} (which holds whenever $u-v$ has a limit as $\tau\to -\infty$) in~\eqref{final lim kernel} to obtain
\begin{gather*}
\lim_{s\to\infty}\mathbb K_{\alpha}^{{\rm PV}}\big(\lambda(x_n),\lambda(y_n);-s_{n,t}^2\big) =\mathbb K^{\textrm{sin}}(\lambda(x_n),\lambda(y_n))+\mathcal{O} \big(s_{n,t}^{-1}\big),
\end{gather*}
as $n\to \infty,s_{n,t}\to \infty$. This yields~\eqref{limit kernelinf}, since $\lambda(x_n)-\lambda(y_n)=\pi(u-v)+\mathcal O(n^{-1})$ as $n\to \infty$, uniformly in $t<t_0$ for some suf\/f\/iciently small $t_0>0$.
Finally, if $\tau_{n,t}\to 0$, we substitute~\eqref{lim Bessel} into~\eqref{final lim kernel} and obtain~\eqref{limit kernel0}.
This completes the proof of Theorem~\ref{thm lim kernel}.

\subsection[Proof of Theorem \ref{Thm PartitionFunction} for $t<0$]{Proof of Theorem \ref{Thm PartitionFunction} for $\boldsymbol{t<0}$}
We start from the dif\/ferential identity (\ref{diff id}). Following the transformations $Y\mapsto T\mapsto S \mapsto R$ given in~(\ref{def T}), (\ref{def S}) and~(\ref{def R}), and recalling the def\/inition~(\ref{defnP}) of the local parametrix, we f\/ind the following form for~$Y^{-1}Y'$. Write $j=1$ for $\Im z>0$ and $j=2$ for $\Im z<0$. Then
\begin{gather}
\big(Y^{-1}Y'\big)_{2,2}(z)=
\Big( B(z) +
W^{-1}(z)W'(z)+e^{-ng(z)\sigma_3}\big(e^{ng(z)\sigma_3}\big)'\Big)_{2,2}\nonumber\\
\hphantom{\big(Y^{-1}Y'\big)_{2,2}(z)=}{}
+\big(\big(\big(\Psi^{(1)}\big)^{-1}\big(\Psi^{(1)}\big)'_z\big) (\lambda(z))\big)_{j,j},\label{diff idv2}
\end{gather}
where $()'$ in general means the derivative with respect to the main argument, $()'_z$ means the derivative with respect to~$z$ , and where
\begin{gather*}
B(z) =\big(\Psi^{(1)}(\lambda)\big)^{-1}(RE)^{-1}(z)(RE)'(z)\Psi^{(1)}(\lambda).
\end{gather*}
We will show that $B(z)$ is bounded uniformly for $s\in(0,\infty)$ and $n$ large, but we will f\/irst calculate the other terms in the above expression. We have
\begin{gather*}
\Big(\big(W^{-1}W'\big)(z)+e^{-ng(z)\sigma_3}\big(e^{ng(z)\sigma_3}\big)'\Big)_{2,2} =-\frac{n}{2}V'(z)+\frac{\alpha /2}{z-z_0}+\frac{\alpha /2}{z+z_0} ,
\end{gather*}
where we used
\begin{gather}
\phi'_{\pm}(x)+g'_{\pm}(x) =\frac{V'}{2}. \label{formula V'}
\end{gather}
By condition (d) of the RH problem for $\Psi^-$ and (\ref{def Psi1}), it follows that
\begin{gather}
\big(\big(\Psi^{(1)}\big)^{-1}\big(\Psi^{(1)}\big)_{z}'\big)_{1,1} (\lambda(z) )\nonumber\\
\qquad{} =-\frac{2i\lambda'(z_0)}{s_{n,t}}\big(H^{-1}H'\big)_{1,1}(1) -\frac{\alpha/2}{z-z_0}+\mathcal O(1) \qquad \textrm{as $z\to z_0$,}
\end{gather}
and
\begin{gather}
\big(\big(\Psi^{(1)}\big)^{-1}\big(\Psi^{(1)}\big)_{z}'\big)_{2,2} (\lambda(z) )\nonumber\\
\qquad{} =-\frac{2i\lambda'(-z_0)}{s_{n,t}}\big(G^{-1}G'\big)_{2,2}(0)
 -\frac{\alpha/2}{z+z_0}+\mathcal O(1) \qquad \textrm{as $z\to -z_0$,}\label{zto-z0}
\end{gather}
where both $\mathcal O(1)$ terms are uniform for large $n$ and suf\/f\/iciently small $|t|$. Now, we use the following lemma, which we will prove below.
\begin{Lemma} \label{lemmapart}\quad
\begin{itemize}\itemsep=0pt
\item[$(1)$] The following identity holds:
\begin{gather*}
\big(G^{-1}G'\big)_{2,2}(0)=-\big(H^{-1}H'\big)_{1,1}(1)=\frac{s}{2}+\frac{\sigma^{-}_\alpha (s)}{\alpha}-\frac{\alpha}{2}.
\end{gather*}
\item[$(2)$] $B(z)$ is bounded uniformly as $n\to \infty$ for $t<t_0$ for some sufficiently small~$t_0$.
\end{itemize}
\end{Lemma}
Equations \eqref{diff idv2}--\eqref{zto-z0} and Lemma~\ref{lemmapart} can be substituted into~\eqref{diff id} to obtain
\begin{gather}
\frac{d}{dt} \log \widehat Z_n(t,\alpha,V)=-\frac{\alpha }{2z_0}\Bigg(\frac{4i\lambda'(z_0)}{s_{n,t}}\left(\frac{s_{n,t}}{2}+\frac{\sigma^{-}_\alpha (s_{n,t})}{\alpha}-\frac{\alpha}{2}\right)+\frac{\alpha }{2z_0}\nonumber\\
\hphantom{\frac{d}{dt} \log \widehat Z_n(t,\alpha,V)=}{}
+\frac{n}{2}\left(V'(-z_0)-V'(z_0)\right)
+B(z_0)-B(-z_0)+\mathcal{O} (1)\Bigg)\label{diff id 3}
\end{gather}
as $n \to \infty$ and for suf\/f\/iciently small $|t|$. After a straightforward calculation
in which one uses~\eqref{asymsigma-} and~\eqref{s2},
Theorem~\ref{Thm PartitionFunction} follows upon integration.

\begin{proof}[Proof of Lemma \ref{lemmapart}]
Equation (\ref{SymmetryPsi-}) implies that
\begin{gather*}
\big(H^{-1}H'\big)_{1,1}(1)=-\big(G^{-1}G'\big)_{2,2}(0).
\end{gather*}
We proceed to evaluate $(G^{-1}G')_{2,2}(0)$. Substituting
\begin{gather*}
\Psi(\zeta)=G(\zeta)\zeta^{\frac{\alpha}{2}\sigma_3} \qquad \textrm{for $\zeta$ in a neighborhood of $0$,}
\end{gather*}
into \eqref{Lax}, (\ref{Eqn:A}) and evaluating the terms of order $z^{-1}$, we f\/ind that
\begin{gather}
\big(G\sigma_3G^{-1}\big)_{2,2}(0)=\frac{2v}{\alpha}-1. \label{Gsigma3G-1}\end{gather}
Writing
\begin{gather*}
G(z)=G_0\big(1+G_1z+\mathcal O\big(z^{2}\big)\big),\qquad z\to0,
\end{gather*}
we also have \begin{gather}\label{equation G1}\big(G^{-1}G'\big)(0)=G_1,
\end{gather}
and, by (\ref{Eqn:B}),
\begin{gather}
\frac{\partial}{\partial s}G_{0} =B_0G_0,\\
\frac{\partial}{\partial s}G_{1} =G_0^{-1}B_1G_0=-\frac{1}{2}G_0^{-1}\sigma_3G_0.\label{equation G1b}
\end{gather}
By (\ref{Gsigma3G-1}), \eqref{equation G1} and \eqref{equation G1b}, we obtain
\begin{gather} \frac{\partial}{\partial s}\big(G^{-1}G'\big)_{2,2}(0)=-\frac{v}{\alpha}+\frac{1}{2}. \label{G-1G'}
\end{gather}
By analyzing at the behavior of $\Psi$ as $s\to +\infty$ in Section~\ref{section tau infty}, one f\/inds that as $s\to +\infty$,
\begin{gather}\label{G-1G'asym}
\big(G^{-1}G'\big)_{2,2}(0)=-\frac{\alpha}{2} +\frac{s}{2}+o(1).
\end{gather}
Integrating \eqref{G-1G'} using \eqref{sigma} and comparing \eqref{G-1G'asym} to \eqref{asymsigma-}, one f\/inds that
\begin{gather*}
\big(G^{-1}G'\big)_{2,2}(0)=\frac{s}{2}+\frac{\sigma(s)}{\alpha}-\frac{\alpha}{2}.
\end{gather*}
It remains to show that $B(z)$ is bounded.
It follows from~\eqref{asymp R} and Remark~\ref{remark: uniform R} that
\begin{gather*}
 R^{-1}(z)R'(z)=\mathcal O \big(n^{-1}\big) \qquad \textrm{as $n \to \infty, t \to 0$.}
 \end{gather*}
 We will show that $E$ and $E^{-1}E'$ are bounded uniformly as $n \to \infty$ and $t\to 0$.
For $z$ in a neighborhood of $z_0$, we write
\begin{gather*}
E(z) =\prod_{j=1}^4 f_j(z),\\
f_1(z) =D^{-1}(\infty)N_0(z),\\
f_2(z) =D(z) \sigma_1\sigma_3 (z+z_0)^{-\frac{\alpha}{2}\sigma_3}\left(\lambda(z)+\frac{is_{n,t}}{4}\right)^{-\frac{\alpha}{2}\sigma_3}n^{\frac{\alpha}{2}\sigma_3},\\
f_3(z) =(z-z_0)^{-\frac{\alpha}{2}\sigma_3}\left(\lambda-\frac{is_{n,t}}{4}\right)^{\frac{\alpha}{2}\sigma_3}n^{-\frac{\alpha}{2}\sigma_3},\\
f_4(z) =e^{n\phi(z)\sigma_3}e^{-i\lambda(z)\sigma_3}e^{\frac{\pi i \alpha}{2}\sigma_3}.
\end{gather*}
Each $f_j$, $j=1,\ldots, 4$ is uniformly bounded at $\pm z_0$ as $n\to \infty $ and $t \to 0$. This is clear for~$f_1$ and~$f_3$. To see this for~$f_2$, one can perform a contour integral to f\/ind that
\begin{gather*}
 D(z_0)=\begin{pmatrix}\mathcal{O}\left(z_0^{-\alpha } \right)&0\\0&\mathcal{O}\left(z_0^{\alpha } \right)\end{pmatrix} .\end{gather*}
For~$f_4$, it follows from the def\/inition of $\lambda$ and the fact that $\phi(z)=\overline{\phi (\overline z)}$. One can do the same for $z$ in a neighborhood of~$-z_0$.
Thus
\begin{gather*}
B(\pm z_0)=\lim_{z\to \pm z_0}\big(\Psi^{(1)}\big)^{-1}(\lambda(z);s_{n,t})\mathcal O (1) \Psi^{(1)}(\lambda(z);s_{n,t}),
\end{gather*}
which is bounded as long as $s$ lies in compact subsets of $(0,+\infty)$.
To see that $B(\pm z_0)$ is bounded as well as $s\to \infty$, we substitute in $\widehat \Psi$ from (\ref{hatPsi}) and use the fact that it converges to the identity. To see that $B(\pm z_0)$ is bounded as $s\to 0$ we refer to the small~$x$ behavior of~$\Psi$ in~\cite{CIK} (formu\-las~(4.23), (4.53), (4.61), (4.64), (4.67)).
\end{proof}

\section[Asymptotic analysis of $T$ as $n\to \infty$ for $t>0$]{Asymptotic analysis of $\boldsymbol{T}$ as $\boldsymbol{n\to \infty}$ for $\boldsymbol{t>0}$}\label{section: RH+}

In this section, we will analyse asymptotically the RH problem for~$T$ in the case where~$t>0$, which means that the singularities $\pm \sqrt{t}$ are real and approach the origin as~$t\to 0$. Many of the notations here will be the same as the ones in Section~\ref{section: RH-} to emphasize the parallels in the analysis. We will refer to many arguments given in the previous section, so it is recommended to read Section~\ref{section: RH-} before this one.
\subsection[Opening of the lens for $t>0$]{Opening of the lens for $\boldsymbol{t>0}$}

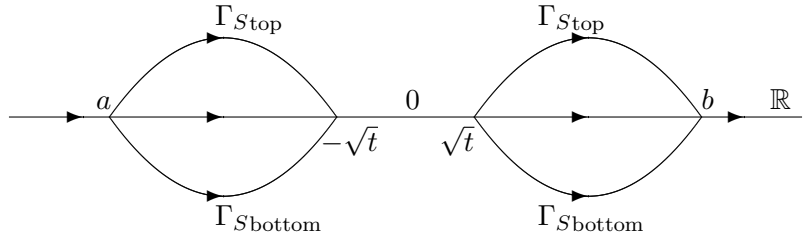
\begin{figure}[t]\centering
\begin{picture}(100,85)(-5,-50)
 \put(42,-2){$0$}
\put(55,-18){$\sqrt{t}$}
\put(10,-18){$-\sqrt{t}$}
\put(-75,-2){$a$}
\put(154,-2){$b$}
\put(-108,-5){\line(1,0){300}}

 \put(111,25){\thicklines\vector(1,0){.0001}}
 \put(111,-5){\thicklines\vector(1,0){.0001}}
 \put(111,-35){\thicklines\vector(1,0){.0001}}
\put(-27,25){\thicklines\vector(1,0){.0001}}
\put(-27,-5){\thicklines\vector(1,0){.0001}}
\put(-27,-35){\thicklines\vector(1,0){.0001}}
\put(170,-5){\thicklines\vector(1,0){.0001}}
\put(-80,-5){\thicklines\vector(1,0){.0001}}

\put(180,-2){$\mathbb{R}$}
\put(-30,30){${\Gamma_S}_{\textrm{top}}$}
\put(92,30){${\Gamma_S}_{\textrm{top}}$}
\put(-30,-46){${\Gamma_S}_{\textrm{bottom}}$}
\put(92,-46){${\Gamma_S}_{\textrm{bottom}}$}

 \qbezier(-70,-5)(-27,55)(16,-5)
 \qbezier(-70,-5)(-27,-65)(16,-5)
 \qbezier(68,-5)(111,55)(154,-5)
 \qbezier(68,-5)(111,-65)(154,-5)
\end{picture}
\caption{The contour ${\Gamma_S}$.}\label{lensesing}
\end{figure}

We open the lens as in Fig.~\ref{lensesing}. We def\/ine ${S}$ in the same way as for $t<0$ in \eqref{def S}, except that we use the contour pictured in Fig.~\ref{lensesing} instead of the one in Fig.~\ref{ContourS}, and except for the fact that
we def\/ine $(z^2-t)^{\alpha}$ such that it is positive on $(a,-\sqrt{t})$ and $(\sqrt{t},b)$:
\begin{gather*}
\big(z^2-t\big)^\alpha=\big|z^2-t\big|^\alpha e^{i\alpha\arg(z-\sqrt{t})}e^{i\alpha\arg(z+\sqrt{t})},
\end{gather*}
with the conventions $-3\pi/2<\arg(z-\sqrt{t})<\pi/2$ and $-\pi/2<\arg(z+\sqrt{t})<3\pi/2$.

The RH problem for $S$ is the following.
\subsubsection*{RH problem for $\boldsymbol{S}$}
\begin{itemize}\itemsep=0pt
\item[(a)]
${S}$ is analytic for $z\in \mathbb{C} {\setminus}{\Gamma_S}$ where ${\Gamma_S}={\Gamma_S}_{\textrm{top}}\cup {\Gamma_S}_{\textrm{bottom}} \cup \mathbb{R}$ is the union of the lens and the real line, see Fig.~\ref{lensesing}.
\item[(b)] ${S}$ has the following jump relations on ${\Gamma_S}$,
\begin{alignat*}{3}
&{S}_+(x)=
{S}_-(x)\begin{pmatrix}
1 &|x^2-t|^{\alpha} e^{-2n\phi(x)} \\
0 &1
\end{pmatrix}\qquad
&&\text{for $x<a$ and $b<x$}, \\
 \nonumber
&{S}_+(x)= {S}_-(x)\begin{pmatrix} 0& |x^2-t|^{\alpha}\\-\frac{1}{|x^2-t|^{\alpha}} &0\end{pmatrix} \qquad&&\text{for $x\in (a,-\sqrt{t})\cup (\sqrt{t},b)$,}& \\
&{S}_+(z)= {S}_-(x) \begin{pmatrix} 1 & 0 \\ \frac{e^{2n\phi(z)}}{(z^2-t)^{\alpha}} &1 \end{pmatrix} \qquad &&\text{for $z\in {\Gamma_S}_{\textrm{top}}\cup {\Gamma_S}_{\textrm{bottom}}$,} & \\
&{S}_+(x)= {S}_-(x)
\begin{pmatrix}
e^{2n\phi _+(x)} &|x^2-t|^{\alpha} \\
0 & e^{2n\phi _-(x)}
\end{pmatrix}\qquad &&\text{for $x\in[-\sqrt{t},\sqrt{t}]$}.&
\end{alignat*}
\item[(c)]
${S}(z)=I+\mathcal{O} \left( z^{-1} \right)$ as $z\rightarrow \infty$.
\item[(d)] For $\alpha<0$, ${S}$ has the following asymptotics as $z\rightarrow \pm \sqrt{t}$:
\begin{gather*}
 {S}(z)=\begin{pmatrix} \mathcal{O}(1)&\mathcal{O}\big(|z\mp \sqrt{t}|^{\alpha}\big)\\
\mathcal{O}(1) &\mathcal{O}\big(|z\mp \sqrt{t}|^{\alpha}\big)\end{pmatrix}.\end{gather*}
For $\alpha\geq 0$ the following asymptotics hold:
\begin{alignat*}{3}
& {S}(z)= \begin{pmatrix} \mathcal{O}(1)&\mathcal{O}(1)\\
\mathcal{O}(1) &\mathcal{O}(1)\end{pmatrix} \qquad &&\text{as $z\to \pm \sqrt{t}$ from outside the lense,}& \\
& {S}(z)=\begin{pmatrix} \mathcal{O}\big(|z\mp \sqrt{t}|^{-\alpha}\big)&\mathcal{O}(1)\\
\mathcal{O}\big(|z\mp \sqrt{t}|^{-\alpha}\big) &\mathcal{O}(1)\end{pmatrix}\qquad &&\text{as $z\to \pm \sqrt{t}$ from inside the lense.}&
\end{alignat*}
\item[(e)] $S(z)$ is bounded as $z\to a$ and as $z\to b$.
\end{itemize}

\subsection{Global parametrix}
As for $t<0$, we now ignore small neighborhoods of $a$, $b$, $0$ and exponentially small jumps.

The solution to the RH problem obtained in this way is constructed as for $t<0$: $N$ is given by (\ref{def N})--(\ref{def N4}), where the parameter $t$ is positive in~(\ref{def N4}).

The function $N$ now has singularities at $ \pm \sqrt{t}$.
Applying a contour deformation argument on the integral in (\ref{def N4}), we obtain the following behavior,
\begin{gather}{N}(z)=\bigO\big(|z\pm\sqrt{t}|^{-\frac{\alpha}{2}\sigma_3}\big),\qquad z\to \mp \sqrt{t}.\label{N sing}
\end{gather}

\subsection[Local parametrices at endpoints $a$ and $b$]{Local parametrices at endpoints $\boldsymbol{a}$ and $\boldsymbol{b}$}
The local parametrices at the endpoints $a$ and $b$ remain unchanged with respect to Section~\ref{locparamaandb}. They satisfy the same jumps as~$S$ near~$a$ and~$b$, and as $n\to\infty$, they match with the global parametrix $N$ on the boundary of f\/ixed disks around~$a$ and~$b$.

\subsection{Local parametrix at the origin}
Let ${U}_0$ be a suf\/f\/iciently small disk around the origin which is independent of~$n$. We will construct a~function~$P$ which satisf\/ies the following conditions.

\subsubsection*{RH problem for $\boldsymbol{P}$}
\begin{itemize}\itemsep=0pt
\item[(a)] ${P}\colon { U}_0{\setminus} \Gamma_S\to \mathbb{C}^{2\times 2}$ is analytic.
\item[(b)] ${P}$ has same jump relations as $S$:
\begin{gather*} {P}_-^{-1}(z){P}_+(z)={S}_-^{-1}(z){S}_+(z) \qquad \text{for $z \in {U}_0 \cap {\Gamma_S}$.} \end{gather*}
\item[(c)]As $n\to\infty$ and simultaneously $t\to 0$ in such a way that $n^2t$ tends to a non-zero constant, we have
\begin{gather*} {P}(z)=\big(I+\mathcal{O}\big(n^{-1}\big)\big){N}(z),\end{gather*} uniformly for~$z$ on the boundary~$\partial {U}_0$.
\end{itemize}
These RH conditions are exactly the same as in the case $t<0$, but it has to be noted that the function~$S$ and its jump contour are dif\/ferent here than in the case $t<0$. Therefore the construction of the local parametrix dif\/fers from the construction done before.

In analogy with the case $t<0$, def\/ine
\begin{gather}\label{def lambda+}
 \lambda(z)=-in\times \begin{cases}
 \phi(z) - \dfrac{\phi_+(\sqrt{t})+\phi_+(-\sqrt{t})}{2} &\textrm{for $\Im z>0$},\\
 -\phi(z)-\dfrac{\phi_+(\sqrt{t})+\phi_+(-\sqrt{t})}{2} & \textrm{for $\Im z<0$}.
 \end{cases}
\end{gather}
By (\ref{phi}), $\lambda$ is a conformal map in $U_0$, \begin{gather*}\lambda\big(\sqrt{t}\big)=-\lambda\big({-}\sqrt{t}\big)=-\frac{in}{2}\big(\phi_+\big(\sqrt{t}\big)-\phi_+\big({-}\sqrt{t}\big)\big),\end{gather*}
and
\begin{gather*}
\lambda'(0)=n\pi\psi_V(0).
\end{gather*}
We have that $\lambda(z)$ is a conformal map which sends $\sqrt{t}$ to $|s_{n,t}|/4$ and $-\sqrt{t}$ to $-|s_{n,t}|/4$, with
\begin{gather}\label{def s+}
s_{n,t}=-4i \lambda\big(\sqrt{t}\big)=-2n\big(\phi_+\big(\sqrt{t}\big)-\phi_+\big({-}\sqrt{t}\big)\big).
\end{gather}
By (\ref{phi}), we have
\begin{gather}\label{s2+}
s_{n,t}=-2\pi in \int_{-\sqrt{t}}^{\sqrt{t}}\psi_V(s)ds=-4in\pi\sqrt{t}\psi_V(0)+\bigO\big(nt^{3/2}\big)=-i\sqrt{\tau_{n,t}}+\bigO\big(nt^{3/2}\big),
\end{gather}
as $t\to 0$, $n\to\infty$, similar to the case $t<0$.

Recall the def\/inition of $\Psi^{(2)}$ in \eqref{DefnPsi2}.
We now f\/ix the lens by requiring that it is mapped to the jump contours $\widetilde\Gamma_1$, $\widetilde\Gamma_2$, $\widetilde\Gamma_3$, and $\widetilde\Gamma_4$ of
${\Psi}^{(2)}$ by $\lambda$. We search for $P$ in the form
\begin{gather}\label{def P2}
{P}(z)={E}(z){\Psi}^{(2)}(\lambda(z);s_{n,t}) {W}(z),
\end{gather}
where
\begin{gather*}
{W}(z)=\begin{cases}
\big(\big(z^2-t\big)^{\alpha/2}e^{-n\phi(z)}\big)^{\sigma_3} \sigma_3\sigma_1& \textrm{for $\Im z>0$, $\Re \lambda(z)\notin \big[{-}\frac{|s_{n,t}|}{4},\frac{|s_{n,t}|}{4}\big]$,}\vspace{1mm}\\
\big( \big(z^2-t\big)^{-\alpha/2}e^{n\phi(z)}\big)^{\sigma_3} &\textrm{for $\Im z<0$, $\Re \lambda(z)\notin \big[{-}\frac{|s_{n,t}|}{4},\frac{|s_{n,t}|}{4}\big]$,} \vspace{1mm}\\
\big(\big(z^2-t\big)^{\alpha/2}e^{\frac{\pi i\alpha}{2}}e^{-n\phi(z)}\big)^{\sigma_3} \sigma_3\sigma_1& \textrm{for $\Im z>0$, $\Re \lambda(z)\in \big[{-}\frac{|s_{n,t}|}{4},\frac{|s_{n,t}|}{4}\big]$,}\vspace{1mm}\\
\big( \big(z^2-t\big)^{-\alpha/2}e^{-\frac{\pi i\alpha}{2}}e^{n\phi(z)}\big)^{\sigma_3} &\textrm{for $\Im z<0$, $\Re \lambda(z)\in \big[{-}\frac{|s_{n,t}|}{4},\frac{|s_{n,t}|}{4}\big]$.}
 \end{cases}
\end{gather*}
The above construction was done in such a way that ${W}$ induces the correct jump relations for~$P$. ${E}$ is an analytic function which has to be such that the matching condition~(c) of the RH problem for~$P$ is valid for $z$ on the boundary of ${U}_0$.
Therefore, we recall the def\/inition of $\chi$ in~\eqref{Defnchi} and let
\begin{gather*}
{E}(z)={N}(z){W}^{-1}(z)e^{-i\lambda(z)\sigma_3} \chi(\lambda(z))^{-1}.
\end{gather*}

\begin{Prop}
${P}(z)$, defined as in \eqref{def P2}, satisfies the RH problem for ${P}$.
\end{Prop}

\begin{proof} Condition (a) and condition (b) hold by construction. We proceed to prove (c). By the def\/inition of $E(z)$, we have
\begin{gather*}
{P}(z) {N}^{-1}(z)={E}(z){\Psi}^{(2)}(\lambda(z);s_{n,t})e^{-i\lambda(z)\sigma_3} \chi(\lambda(z))^{-1} {E}^{-1}(z). \end{gather*}
For $z\in\partial U_0$, it follows from (\ref{asymptoticsPsi2}) that
\begin{gather*}
{\Psi}^{(2)}\big(\lambda(z);s_{n,t}\big)=\big(I+\mathcal{O}\big(n^{-1}\big)\big)\chi(\lambda(z))e^{i\lambda(z)\sigma_3}\end{gather*}
as \looseness=-1 $n\to \infty$, where the $\mathcal{O} (n^{-1})$ is uniform for $s_{n,t}$ in compact subsets of $(0,-i\infty)$ and $\lambda(z)$ suf\/f\/i\-cient\-ly large. Using the same calculations as in the proof of Proposition~\ref{P0theorem-} we f\/ind that~$E$ is bounded on~$\partial U_0$ and so the result follows in the same manner as in the proof of Proposition~\ref{P0theorem-}.
\end{proof}

\subsection{Small norm RH problem}
A small norm RH problem can be constructed in a similar way as in Section~\ref{smallnorm}: def\/ine
\begin{gather*}
{R}(z)=\begin{cases}
{S}(z){P}(z)^{-1}& \mbox{for $z\in {U}_0$,}\\
S(z){P}^{(a)}(z)^{-1}& \mbox{for $z\in {U}_a$,}\\
S(z){P}^{(b)}(z)^{-1}& \mbox{for $z\in {U}_b$,}\\
S(z){N}(z)^{-1}&\mbox{for $z\in \mathbb{C} {\setminus} ({U}_0\cup {U}_a\cup {U}_b)$}.
 \end{cases}
\end{gather*}

Then $R$ satisf\/ies the following RH problem, similar to the case where $\tau<0$.
\subsubsection*{RH problem for $\boldsymbol{R}$}
\begin{itemize}\itemsep=0pt
\item[(a)] $R$ is analytic on $\mathbb{C}{\setminus} \Gamma_R$ where \begin{gather*} \Gamma_R= \partial U_a\cup \partial U_b \cup \partial U_0 \cup (\Gamma_S{\setminus} ((a,b)\cup U_0\cup U_a \cup U_b)).\end{gather*}
\item[(b)] $R$ has the following jump relations on $\Gamma_R$:
\begin{gather*} R_+(z)=R_-(z)\big(I+\mathcal{O}\big(n^{-1}\big)\big)\end{gather*}
as $n\to\infty$, uniformly for $z\in\Gamma_R$, in the double scaling limit where $\tau_{n,t}\to\tau>0$.
\item[(c)] $R(z)$ has the following asymptotics as $z\to \infty$:
\begin{gather*}
R(z)= I+\mathcal{O}\big(z^{-1}\big).
\end{gather*}
\end{itemize}

In the double scaling limit where $n\to\infty$ and at the same time $t\to 0$ in such a way that $n^2t$ tends to a non-zero constant, it follows that
\begin{gather}\label{R as 2}
{R}(z)=I+\mathcal{O}\big(n^{-1}\big),
\end{gather}
uniformly for $z$ of\/f the jump contour for $R$.
As for $t<0$, one can extend this result to the cases where $n\to \infty$ and $t\to 0$ in such a way that $n^2t\to \infty$ and in such a way that $n^2t\to 0$. We do not give the details here, as the arguments are very similar to those in \cite{CK}: see equations (5.1), (5.17), (5.18), (5.25) in \cite{CK} for $n^2t\to\infty$ and equations (6.1), (6.28), (6.32), and (4.6) in~\cite{CK} for $n^2t\to 0$.

\subsection{Proof of Theorem \ref{thm lim kernel}} \label{pf main thm +}

In the same way as when $t<0$, we can f\/ind an expression for the correlation kernel in terms of~${S}$ when $t>0$:
\begin{gather}
{K}_n(x,y)=\frac{1}{2\pi i(x-y)}\begin{pmatrix} -\dfrac{e^{n\phi_+(y)}}{|y^2-t|^{\alpha/2}}\delta(y)& \dfrac{|y^2 -t|^{\alpha/2}}{e^{n\phi_+(y)}} \end{pmatrix}\nonumber\\
\hphantom{{K}_n(x,y)=}{} \times{S}_+^{-1}(y){S}_+(x) \begin{pmatrix}|x^2 -t|^{\alpha/2}e^{-n\phi_+(x)}\\ e^{n\phi_+(x)} |x^2-t|^{-\alpha/2} \delta(x) \end{pmatrix}\label{correlkernhat}
\end{gather}
where
\begin{gather} \delta(x)= \begin{cases} 0 & \textrm{for $x\in \big({-}\sqrt{t},\sqrt{t}\big)$,}\\
1& \textrm{for $x \in \big(a,-\sqrt{t}\big)\cup\big(\sqrt{t},b\big)$.}
\end{cases} \end{gather}

Substituting the small norm RH solution $R$ into the formula for the correlation kernel in~(\ref{correlkernhat}), one obtains
\begin{gather}
{K}_n(x,y)=\frac{1}{2\pi i(x-y)} \begin{pmatrix} -\dfrac{e^{n\phi_+(y)}}{|y^2-t|^{\alpha/2}}\delta(y)& \dfrac{|y^2 -t|^{\alpha/2}}{e^{n\phi_+(y)}} \end{pmatrix}{W}_+^{-1}(y)\big({\Psi}_+^{(2)}\big)^{-1}({\lambda(y)};s_{n,t})\!\!\label{EqnforkernelPsi2}\\
\hphantom{{K}_n(x,y)=}{}\times {E}_+^{-1}(y)R_+(y)R^{-1}_+(x){E}(x){\Psi}_+^{(2)}({\lambda(x)};s_{n,t}){W}(x) \begin{pmatrix}e^{-n\phi_+(x)}|x^2 -t|^{\alpha/2}\\ e^{n\phi_+(x)} |x^2-t|^{-\alpha/2} \delta(x) \end{pmatrix} .\nonumber
\end{gather}
From the asymptotic behavior \eqref{R as 2} and the analyticity of $R$ and $y$, we have
\begin{gather*}
{E}_+^{-1}(y)R_+(y)R_+(x){E}(x)
=I+\mathcal{O}(x-y),
\end{gather*}
as $x\to y$ for $n$ suf\/f\/iciently large.
Substituting this into (\ref{EqnforkernelPsi2}) along with the def\/inition for ${W}$, one f\/inds that
\begin{gather*}
{K}_n(x,y)=\frac{1}{2\pi i(x-y)} \begin{pmatrix} 1& \delta(y) \end{pmatrix}\big({\Psi}_+^{(2)}\big)^{-1}(\lambda(y);s_{n,t})\\
\hphantom{{K}_n(x,y)=}{} \times (I+\mathcal{O}(x-y)){\Psi}^{(2)}_+(\lambda(x);s_{n,t})\begin{pmatrix}\delta(x)\\ -1 \end{pmatrix} .
\end{gather*}
Now, we can use \eqref{def Phi +} and \eqref{DefnPsi2} to express $\Psi^{(2)}$ in terms of the functions~$\Phi_1$ and~$\Phi_2$.
We have the relation
\begin{gather}\label{Phi+ in Psi2}
\begin{pmatrix}
\Phi_1(u;\tau)\\
\Phi_2(u;\tau)
\end{pmatrix}=e^{-\frac{s}{4}\sigma_3}
\begin{cases}
\Psi^{(2)}_+\big(u;-i\sqrt{\tau}\big)
\begin{pmatrix}0\\-1\end{pmatrix},& -|s|/4<u<|s|/4,\vspace{1mm}\\
\Psi^{(2)}\big(u;-i\sqrt{\tau}\big)\begin{pmatrix}1\\-1\end{pmatrix},& u\in\mathbb R{\setminus} [-|s|/4,|s|/4],
\end{cases}
\end{gather}
and after a similar calculation to the one in Section~\ref{section: proof double}, using the fact that $\lambda'(0)=\pi n\psi_V(0)$,
one f\/inds that Theorem~\ref{thm lim kernel} holds for~$\tau_{n,t}\in (0,\infty)$.

\subsection[Proof of Theorem \ref{Thm PartitionFunction} for $t>0$]{Proof of Theorem \ref{Thm PartitionFunction} for $\boldsymbol{t>0}$}
We assume here that $\alpha> 0$. The formulas (\ref{diff idv2}), (\ref{formula V'}) which we obtained in the proof of Theorem~\ref{Thm PartitionFunction} for $t<0$, hold also in the case $t>0$, if we replace $\Psi^{(1)}$ by $\Psi^{(2)}$.
We obtain from~(\ref{DefnPsi2}) and condition d) in the RH problem for $\Psi^{+}$ that as $z \to \sqrt{t}$,
\begin{gather*}
\big(\big(\Psi^{(2)}\big)^{-1}\big(\Psi^{(2)}\big)'_z\big)_{2,2}(\lambda(z))
=-\frac{\alpha/2}{z-\sqrt{t}}+\frac{2\lambda'(z)}{|s|}\big(H^{-1}H'\big)_{2,2}(1)+\mathcal O (1 ),
\end{gather*}
where the error term is uniform for large $n$ and suf\/f\/iciently small $t$. Likewise as $z\to -\sqrt{t}$,
\begin{gather*}
\big(\big(\Psi^{(2)}\big)^{-1}\big(\Psi^{(2)}\big)'_z\big)_{2,2}(\lambda(z))
=-\frac{\alpha/2}{z+\sqrt{t}}+\frac{2\lambda'(z)}{|s|}\big(G^{-1}G'\big)_{2,2}(0)+\mathcal O (1 ),
\end{gather*}
where the error term is uniform for large $n$ and suf\/f\/iciently small $t$.
Thus, by \eqref{diff id}, we obtain
\begin{gather}
\frac{d}{dt} \log \widehat Z_n(t,\alpha,V) =-\frac{\alpha }{2z_0}\Bigg(\frac{2i\lambda'(z_0)}{s_{n,t}}\big(G^{-1}G'\big)_{2,2}(0)-\frac{2i\lambda'(-z_0)}{s_{n,t}}\big(H^{-1}H'\big)_{2,2}(1)\nonumber\\
 \hphantom{\frac{d}{dt} \log \widehat Z_n(t,\alpha,V) =}{} +\frac{\alpha }{2z_0}+\frac{n}{2} (V'(-z_0)-V'(z_0) )
+B(z_0)-B(-z_0)+\mathcal{O} (1)\Bigg)\label{diff id 4}
\end{gather}
as $n\to \infty$ and suf\/f\/iciently small~$t$.
As in the case $t<0$, we have
\begin{gather*}
\big(G^{-1}G'\big)'_{2,2}(0) =-\big(H^{-1}H'\big)'_{2,2}(1)=\frac{\sigma'(s)}{\alpha}+\frac{1}{2},
\end{gather*}
and we can use the large $s$ asymptotics given in~\cite[equation~(5.1)]{CK} to f\/ind that
\begin{gather*}
\big(G^{-1}G'\big)_{2,2}(0)=-\big(H^{-1}H'\big)_{2,2}(1)=\frac{\sigma(s)}{\alpha}+\frac{s}{2}-\frac{\alpha}{2}.
\end{gather*}
In the same manner as for $t<0$, one can show that
\begin{gather*}
B(\pm z_0)=\lim_{z\to \pm z_0}\big(\Psi^{(2)}\big)^{-1}(\lambda(z);s_{n,t})\mathcal O (1) \Psi^{(2)}(\lambda(z);s_{n,t}), \end{gather*}
which is uniformly bounded for~$n$ large and $s\in (0,-i\infty)$.
To show this for $s\to -i\infty$, one uses equations~(5.10) and (5.15)~from \cite{CK}; while
to show this for $s\to 0$, one uses equations~(6.10),~(6.19) and~(6.23)~from \cite{CK}. Substituting the formula for~$G$ and~$H$ into~(\ref{diff id 4}) with~$B$ bounded yields Theorem~\ref{Thm PartitionFunction} for $t,\alpha>0$ after integration.

\subsection*{Acknowledgements}
The authors are grateful to I.~Krasovsky and N.~Simm for useful discussions. They were supported by the European Research Council under the European Union's Seventh Framework Programme (FP/2007/2013)/ ERC Grant Agreement 307074 and by the Belgian Interuniversity Attraction Pole P07/18.

\pdfbookmark[1]{References}{ref}
\LastPageEnding

\end{document}